%% file: main.tex
\documentclass[preprint]{elsarticle}
\input{aux/packages}
\input{aux/macros_global}

\input{aux/macros_local}

\begin{document}
\hypersetup{
    linkcolor=blue, 	
    anchorcolor=black, 	
    citecolor=green, 	
    filecolor=cyan, 	
    urlcolor=magenta,   
    pdftitle={DSEF1},
    pdfauthor={E. Emanuel Rapsch}
    }
    
\input{contents/frontmatter}

\tableofcontents

\input{contents/text}

\phantomsection
\addcontentsline{toc}{section}{References}
\bibliography{main}
\bibliographystyle{plain}

\newpage
\appendix
\input{contents/appendix}

\end{document}

%% file: aux/packages.tex
\usepackage[utf8]{inputenc}
\usepackage[T1]{fontenc}
\usepackage[english]{babel}
\usepackage{amsmath}
\usepackage{amssymb}
\usepackage{amsthm} 
\usepackage{amscd}
\usepackage{bm}
\usepackage{mathtools}
\usepackage{mathrsfs}
\usepackage{fancyhdr}
\usepackage{fancyvrb}
\usepackage{lmodern}
\usepackage[top=4cm]{geometry}
\usepackage[shortlabels]{enumitem}
\setlist[enumerate,1]{
{1.},
ref={\arabic*}
}
\setlist[enumerate,2]{
{(a)},
ref={\theenumi{}(\alph*)}
}
\setlist[itemize]{
label={--}
}
\usepackage{hyperref}
\hypersetup{
    colorlinks=true,
    linkcolor=blue, 	
    anchorcolor=black, 	
    citecolor=green, 	
    filecolor=cyan, 	
    urlcolor=magenta,
    pdfpagemode=FullScreen,
    }
\usepackage{graphicx}
\usepackage{tikz-cd}
\usetikzlibrary{positioning}

%% file: aux/macros_global.tex
\allowdisplaybreaks
\theoremstyle{plain}
\newtheorem{thm}{Theorem}[section]
\newtheorem{corollary}[thm]{Corollary}
\newtheorem{lemma}[thm]{Lemma}
\newtheorem{proposition}[thm]{Proposition}
\theoremstyle{definition}
\newtheorem{definition}[thm]{Definition}
\newtheorem{remark}[thm]{Remark}
\newtheorem{example}[thm]{Example}
\numberwithin{equation}{subsection}
\newcommand{\mf}{\mathfrak}
\newcommand{\mc}{\mathcal}
\newcommand{\ms}{\mathscr}
\newcommand{\op}{\operatorname}

\newcommand{\id}{\operatorname{id}}
\newcommand{\im}{\operatorname{im}}

\newcommand{\ev}{\operatorname{ev}}

\newcommand{\R}{\mathbb R}
\newcommand{\Z}{\mathbb Z}
\newcommand{\N}{\mathbb N}
\newcommand{\Nast}{{\N^{\ast}}}

\newcommand{\e}{\varepsilon}

\newcommand{\bigmid}{~ \Big |~}

\newcommand{\A}{\mathbb{A}}
\newcommand{\B}{\mathbb{B}}

\newcommand{\Pot}{{\mathcal P}}

\newcommand{\T}{\mathbb T}

\renewcommand{\P}{{\mathbb P}}
\renewcommand{\d}{{\mathrm d}}
\newcommand{\X}{{\bm{\mathsf X}}}
\newcommand{\x}{{\bm{\mathsf x}}}

\newcommand{\Tr}{{\bm{\mathsf T}}}
\newcommand{\y}{{\bm{\mathsf y}}}
\newcommand{\w}{{\bm{\mathsf w}}}

%% file: aux/macros_local.tex
\makeatletter
\def\ps@pprintTitle{
  \let\@oddhead\@empty
  \let\@evenhead\@empty
  \let\@oddfoot\@empty
  \let\@evenfoot\@oddfoot
  }
\makeatother
\newcommand{\Today}{9th November 2024}
\usepackage{etoolbox}
\patchcmd{\MaketitleBox}{\footnotesize\itshape\elsaddress\par\vskip36pt}{\footnotesize\itshape\elsaddress\par\parbox[b][36pt]{\linewidth}{\vfill\hfill\textnormal{\Today}\hfill\null\vfill}}{}{}%
\patchcmd{\pprintMaketitle}{\footnotesize\itshape\elsaddress\par\vskip36pt}{\footnotesize\itshape\elsaddress\par\parbox[b][36pt]{\linewidth}{\vfill\hfill\textnormal{\Today}\hfill\null\vfill}}{}{}%

%% file: contents/frontmatter.tex
\begin{frontmatter}

\title{Decision making in stochastic extensive form I: \\Stochastic decision forests}
\author[1]{E.\ Emanuel Rapsch}
\ead{rapsch@math.tu-berlin.de}
\affiliation[1]{organization={Institut für Mathematik, Technische Universität Berlin},
                addressline={Straße des 17.\ Juni 136}, 
                city={10623 Berlin},
                country={Germany}
                }

\begin{abstract}
	A general theory of stochastic decision forests is developed to bridge two concepts of information flow: decision trees and refined partitions on the one side, filtrations from probability theory on the other. Instead of the traditional ``nature'' agent, this framework uses a single lottery draw to select a tree of a given decision forest. Each ``personal'' agent receives dynamic updates from an own oracle on the lottery outcome and makes partition-refining choices adapted to this information. 
    This theory addresses a key limitation of existing approaches in extensive form theory, which struggle to model continuous-time stochastic processes, such as Brownian motion, as outcomes of ``nature'' decision making. Additionally, a class of stochastic decision forests based on time-indexed action paths is constructed, encompassing a wide range of models from the literature and laying the groundwork for an approximation theory for stochastic differential games in extensive form.
\end{abstract}


\begin{keyword}
Extensive form games \sep Dynamic games \sep Stochastic games \sep Decision making \sep Sequential decision theory \sep Stochastic processes\smallskip

\JEL C73 \sep D81 \smallskip

\MSC[2020] 91A15 \sep 91A18 \sep 
91B06 
\end{keyword}

\end{frontmatter}

%% file: contents/text.tex
\section*{Introduction}
\addcontentsline{toc}{section}{Introduction}

This paper is the first in a series aimed at developing a unified theory of stochastic games and decision problems in extensive form. Representing a decision problem in extensive form means specifying it in terms of what the author suggests calling ``extensive form characteristics'', namely, the flow of information about past choices and exogenous events, along with the set of choices available to decision makers. Although classical theory, as established by von Neumann and Morgenstern in \cite{Neumann1944} and furthered by Kuhn in \cite{Kuhn1950,Kuhn1953Extensive}, relies on strong finiteness assumptions, the concept itself is very general and broadly applicable. In a series of papers, including \cite{AlosFerrer2005,AlosFerrer2008,AlosFerrer2011Comment}, and in the monograph \cite{AlosFerrer2016}, Alós-Ferrer and Ritzberger develop an abstract, highly general theory of extensive form games and decision problems and give a concrete order-theoretic characterisation of its own boundaries. Beyond these boundaries, the notions of strategy, outcome, and equilibrium lose rigorous decision-theoretic meaning when applied to the extensive form characteristics of a given decision problem.

Hence, any decision problem and game exhibiting extensive form characteristics lies within these boundaries or is, at least, some sort of limit of objects within these boundaries. The inclusion, or the precise meaning of this limit, allows us to rigorously determine the decision-theoretic meaning of strategies, outcomes, equilibrium, etc.\ for the given decision problem. This opens a new perspective on problems in continuous time in particular -- a domain where the subtlety of this issue has long been recognised (see, e.g.\ \cite{Simon1989,Stinchcombe1992} and the references therein). Based on the maximality result in \cite{Stinchcombe1992}, systematically emphasised in \cite{AlosFerrer2008,AlosFerrer2011Comment}, which advocates restricting outcomes to certain piecewise constant paths, \cite{AlosFerrer2015} introduces, for the first time, a rigorous extensive form foundation of continuous-time games and decision problems involving such outcomes.

Having said that, an important question remains to be addressed: How can the ``extensive form characteristics'' of stochastic games, decision problems exposed to randomness, or those endowed with randomisation devices be modelled? The standard approach in game and decision theory is to introduce a ``nature'' agent executing a behaviour strategy (described in, e.g.\ \cite[Subsection~2.2.3]{AlosFerrer2005} or \cite{Fudenberg1991,AlosFerrer2016}, with roots in the works of Shapley \cite{Shapley1953} and Harsanyi \cite{Harsanyi1967,Harsanyi1968a,Harsanyi1968b}). However, a fundamental problem arises in the continuous-time case: many relevant modelling applications -- from economics (see, e.g.\ \cite{Riedel2017,Grenadier1996Strategic}) and finance (see \cite{Delbaen2006Mathematics,Pham2009Continuous,Bayer2023Rough,Lasry2007}) to engineering (see \cite{Cohen2023Optimal,Huang2006} and the references therein) and reinforcement learning (see, for example, \cite{Guo2023Reinforcement}) -- require exogenous noise represented by stochastic processes whose paths are anything but piecewise constant, with Brownian motion serving as a paradigmatic case (see, for instance, \cite{Pham2009Continuous,Karatzas1998Methods,Cohen2015Stochastic} for an overview of both theory and applications). By no means this implies ``the world to be Brownian'' or any similar assumption; rather, such models are widely used, be it because it appears reasonable to allow for unpredictably small time lags until the next exogenous information revelation, be it because certain quantities arguably evolve in decision paths of low regularity, be it for mathematical and computational convenience. One can certainly argue that these models exhibit ``extensive form characteristics'' without fitting into the Alós-Ferrer--Ritzberger framework (\cite{AlosFerrer2016}). This creates a fundamental issue: the lack of a rigorous decision-theoretic foundation of a large class of relevant models. 

Despite this, stochastic control and differential game theory, based on stochastic analysis, have found a pragmatic way to deal with the ``extensive form characteristics'' of a continuous-time decision problem (see, e.g.\ \cite{Cohen2015Stochastic,Carmona2018}). While extensive form theory is based on graph theory and refined partitions, stochastic analysis models exogenous information through filtrations on a given measurable space of scenarios. These are strictly less restrictive objects because $\sigma$-algebras need not be generated by partitions of the sample space. This works as if a one-shot lottery draw selects a scenario $\omega$ at the beginning, without communication, and then, as time progresses, more and more properties of $\omega$ become known to the decision makers. For instance, this could be knowledge about the realised path of a Brownian motion. From a decision-theoretic perspective, there is no disadvantage in using this weaker structure for modelling exogenous information. 

However, when it comes to the actual decision makers, the stochastic analysis-inspired approach is not compatible with the extensive form paradigm in the strict sense. This is because it does not explain strategies, randomisation, outcomes, or (dynamic, e.g.\ subgame-perfect) equilibria in terms of a decision tree-like object and choices locally available at moves. In this approach, ``strategies'' are typically defined as stochastic processes satisfying a certain measurability condition (e.g.\ progressive measurability) and such that a given stochastic differential equation, depending on them in a non-anticipating way, has a unique solution in some sense. This contradicts the problem's own extensive form characteristics in the strict sense, due to the \emph{ex post} restriction of the strategy space, which implies that the availability of a choice depends on future decisions (for a detailed discussion of this issue, see the introduction in \cite{AlosFerrer2015}). The ``outcome'' is taken to be the solution to the differential equation, and ``equilibria'' (or ``optimal controls'') are defined with implicit reference to the dynamic programming principle, but without an extensive form foundation. This not only fundamentally undermines the use of the equilibrium concept, but also introduces potential confusion, particularly regarding the decision-theoretic meaning of these ``equilibria'' (for example, concerning the notions of ``closed'' and ``open loop'' ``equilibria'' in stochastic differential games, see the discussion in \cite[p.\ 72--76]{Carmona2018}), as well as the precise definition of subgames (for example, in stochastic timing games, see \cite{Riedel2017} for further discussion).

Thus, there is room for an extensive form theory that models exogenous information through filtration-like objects, while modelling endogenous decision making with decision tree-like objects and partition-refining choices. This approach synthesises both frameworks in a productive way, which the author refers to as \emph{stochastic extensive forms}. This is the aim of this project, with the present paper forming its first part. The theory can be used to construct extensive form games involving general stochastic processes as noise. Moreover, it provides a way to precisely characterise how stochastic control and differential games can be approximated by stochastic extensive form decision problems. For instance, it applies to the stochastic timing game, which is typically not at all formulated in extensive form (see, in particular, \cite{Riedel2017}, where strategic form games are stacked according to a notion of ``consistency'', which is justified \emph{a priori} by an analogy to discrete time and \emph{a posteriori} in \cite{Steg2018} by a ``discrete time with an infinitesimal grid'' approximation argument reminiscent of \cite{Simon1987,Simon1989}). While a central motivation for this project stems from continuous-time problems, the theory is not confined to this domain. It is fundamentally about pushing the boundaries of extensive form theory as far as logically possible, with a focus set on probability. This makes it possible to formally establish links between uses of concepts such as equilibrium, which otherwise might seem ad hoc. It also allows for a clear outline of the decision-theoretic structure that arises naturally from the combination of these two ideas, distinguishing it from the additional structure sometimes added in modelling based on the modeller's discretion or specific perspectives on a given problem. For example, stochastic extensive form theory is independent of the representation of noise through a ``nature'' agent's virtual decision making, and in that sense, it is naturally stochastic. \smallskip

The first step in this endeavour, which is taken in the present paper, is to develop a theory of stochastic decision forests. These can be understood as forests of decision trees, where each tree corresponds to exactly one exogenous scenario, equipped with a similarity structure across trees, which identifies moves. Decision trees, as graph-theoretical objects, are the traditional base model for extensive form games and decision problems. However, as pointed out in \cite{AlosFerrer2005}, the refined partitions-based representation, using the set of all maximal chains, exhibits not only strong duality properties but also makes dynamic decision making amenable to the traditional decision-theoretic paradigm of choice under uncertainty: acts which assign a consequence to any state, following Savage's framework (\cite{Savage1972Foundations}). In this refined partitions approach, acts translate into strategies, mapping each move to a choice determined by a set of outcomes still possible at that move from the agent's perspective. This is how uncertainty about choices (whether past or future choices of oneself or others, and present choices of opponents) -- that is, endogenous information -- is handled. 
The modelling of uncertainty about exogenous noise, or exogenous information, in this paper, however, is fundamentally different. It is the aforementioned similarity structure, consisting in so-called random moves, that can be equipped with a filtration-like object which dynamically reveals information about the realised exogenous scenario. This allows for general stochastic processes to model noise without running into outcome generation problems along that dimension. The adaptedness of strategies with respect to exogenous information can be based on a concept of adapted choices that is introduced in the final section.

The fundamental motivation for this theory arises in a large class of concrete applications. Thus, its presentation is tightly accompanied by examples. While some of them remain rather pedagogical, a general class of stochastic decision forests based on time-indexed action paths is also constructed. A small set of readily verifiable axioms allows for many different time regimes, many different specifications of the outcomes (e.g.\ paths of timing games, also in the case of the scenario-dependent expiration of certain options), and general stochastic noise. Of course, it also includes the deterministic case. To the best of the author's knowledge, this theoretical unification of action path-based decision problems with extensive form characteristics within a single framework is a second new contribution of this paper in its own right.\smallskip

Before concluding the introduction, three remarks remain to be made. First, this is the first of a series of three papers; thus, similar to the opening movement of a piano sonata, its first part is self-contained and addresses an independent research problem in its own right, but clearly prepares for and is closely related to the other two. Second, as pointed out by Aumann (see \cite{Aumann2020}), game theory arguably is ``interactive decision theory'', while decision theory is, in a trivial sense, single-player game theory. Still, game theory and decision theory are not the same; the former is more concerned with the ``interactive'' aspect, while the latter focuses more on the ``decision'' facet. This paper is more concerned with the latter, so it primarily employs the corresponding terminology. However, when the context allows, ``game'' is used instead of ``decision problem'', or vice versa, and similarly, the terms ``decision maker'', ``agent'', and ``player''  are used interchangeably. Note, however, that the concept of a decision maker is not formalised in this first paper; its use is purely motivational and of no structural importance at this stage. Finally, the relevant proofs of all new theorems, propositions, lemmata, and claims in examples can be found in the corresponding subsection of the appendix.

\section*{Notation}
\begin{itemize}[label=--]
    \item $\N = \Z_+$ = the natural numbers including zero, $\Nast = \N \setminus \{0\}$, $\R$ = the real numbers, $\R_+ = \{x\in\R \mid x \ge 0\}$;
    \item a function $f\colon D \to V$ from a set $D$ called \emph{domain} to a set $V$ is a subset of $D\times V$ such that for all $x\in D$ there is a unique $y\in V$ with $(x,y)\in f$, and this unique $y$ is denoted by $f(x)$; in other words, $f$ is described through its graph; as an abbreviation, the constant map on $D$ with value $y\in V$ is denoted by $y_{D} = D \times \{y\}$;
    \item $\mc P(A) = \mc P A$ = the set of subsets of a given set $A$, $\mc P(f) = \mc P f$ = the function $\mc P A \to\mc P B, M \mapsto \{ f(m) \mid m\in M\}$ for a given function $f\colon A \to B$ between two sets $A$ and $B$;\footnote{$\mc P$ defines a covariant endofunctor on the category of sets.}
    \item $\im f = (\mc P f)(D)$ = the image of a set-theoretic function $f\colon D \to V$;
    \item $\bigcup M $ = the union of a set $M$ = the set of all $x$ that are the element of some $S\in M$, also written $\bigcup_{i\in I} S_i$ in case $M$ is the image of some function $I\ni i\mapsto S_i$, for some set $I$;
    \item $|M|$ = if $M$ is not an element of some real coordinate space $\R^d$, $d\in\N$, then this is the cardinality of the set $M$;
    \item $>$ = the strict partial order associated to a weak partial order $\ge$;
    \item $\ms E|_D = \{E \cap D \mid E\in\ms E\}$, for any $\sigma$-algebra $\ms E$ and any $D\in\ms E$;
    \item $\ms E_1 \vee \ms E_2$ = the smallest $\sigma$-algebra on $\Omega$ containing both $\ms E_1$ and $\ms E_2$, given a set $\Omega$ and sets $\ms E_1,\ms E_2 \subseteq \mc P(\Omega)$.
\end{itemize}

\section{Decision forests}\label{sec:def}

The basic object of classical extensive form decision and game theory is the decision tree. In the stochastic generalisation presented here, nature does not act as an agent taking decisions dynamically, but is simply replaced with a device ``randomly'' selecting the decision tree the ``personal'' agents follow during their decision making process. This implies that we consider \emph{decision forests}, rather than decision trees. Although this term is not new (as testified by \cite{Rokach2016Decision}), we consider it in the context of abstract decision theory, and more precisely within the -- decision-theoretically natural -- refined partitions framework, and in the aim of making this framework amenable to exogenous noise in the sense of general probability theory. This framework has originally been developed for trees, rather than forests, pioneered in \cite[Section~8]{Neumann1944}, and developed in much more generality in \cite{AlosFerrer2005,AlosFerrer2008} and subsequent papers, much of which is covered in the monograph \cite{AlosFerrer2016}. In this section, we fix some order-theoretic language, present a definition and interpretation of decision forests within the refined partitions approach, and show that decision forests are nothing but forests of decision trees.

\subsection{Order- and graph-theoretic conventions}\label{subs:conventions}
We first recall some basic definitions from graph and order theory, thereby fixing conventions used in this text, which combine those from \cite{Bollobas2013Modern,AlosFerrer2005,Davey2002}. In a partially ordered set (in short, \emph{poset}) $(N,\ge)$ a \emph{chain} is a subset $c\subseteq N$ such that for all $x,y\in c$, $x\ge y$ or $y\ge x$ holds true. A \emph{maximal chain} is a chain that is maximal as a chain with respect to set inclusion in $\mc P(N)$. $x\in N$ is called a \emph{maximal element} iff there is no $y\in N$ other than $x$ such that $y\ge x$. $x\in N$ is called \emph{maximum} iff for all $y\in N$, $x \ge y$. For $x\in N$, the \emph{principal up-set} $\uparrow x$ and \emph{principal down-set} $\downarrow x$ are defined by
\[ \uparrow x = \{y\in N \mid y\ge x\}, \qquad \downarrow x = \{y\in N \mid x \ge y\}. \]

Moreover, in this text, a poset $(F,\ge)$ is called a \emph{forest} iff for every $x\in F$, $\uparrow x$ is a chain. A forest $(F,\ge)$ is called \emph{rooted} iff $F\neq\emptyset$ and for every $x\in F$, $\uparrow x$ contains a maximal element of $(F,\ge)$. A forest $(T,\ge)$ is called a \emph{tree} iff for every $x,y\in T$, $(\uparrow x) \cap (\uparrow y) \neq \emptyset$. Given a forest $(F,\ge)$, the elements $x\in F$ are called \emph{nodes}. Nodes $x\in F$ such that $\downarrow x = \{x\}$ are called \emph{terminal}. We state the following lemma, fundamental for what follows. It can actually be seen as an explicitly order-theoretic reformulation of a basic result from graph theory (see the discussion in \cite[Section~I.1]{Bollobas2013Modern}).\footnote{As the claim that it is a reformulation requires proof, and also for the reader's convenience, a proof can be found in the appendix.}

\begin{lemma}\label{lemma:partion_of_forest}
    For any forest $(F, \ge)$ there exists a unique partition $\mc F$ of $F$ into trees such that for all $x, y \in F$ with $x \ge y$ there is $T \in \mc F$ with $x, y \in T$. If $(F,\ge)$ is rooted, then for any $T\in\mc F$, $(T,\ge)$ is a rooted tree and has a maximum.
\end{lemma}

The elements of $\mc F$ are called \emph{connected components} of $(F,\ge)$. The maximum of a rooted tree $(T,\ge)$ is called the \emph{root}. The \emph{roots} of a forest $(F,\ge)$ are the roots of its connected components. A \emph{decision forest} (\emph{decision tree}) is a rooted forest (tree, respectively) $(F,\ge)$ such that all $x,y\in F$ with $x\neq y$ can be separated by some maximal chain $c\subseteq F$, that is, $c\cap \{x,y\}$ is a singleton. A \emph{move} in a decision forest $(F,\ge)$ is a non-terminal node $x\in F$. In the present text, let us call a decision forest $(F,\ge)$ \emph{(everywhere) non-trivial} iff some (any, respectively) root is a move.

If $V$ is a set, a \emph{$V$-poset} is a subset $N\subseteq \mc P(V)$. The name derives from the fact that $(N,\supseteq)$ defines a poset of subsets of $V$ ordered by set inclusion.

\begin{remark}\label{rmk:def_tree_AFK}
    Note that the definition of a tree used in this text is an order-theoretic transcription of a graph-theoretical concept. What is called forest here, corresponds indeed to a ``forest'', also named ``acyclic graph'' in graph theory (see \cite[Sections I.1, I.2]{Bollobas2013Modern}), but is called ``tree'' in \cite[Definition~1]{AlosFerrer2005}. This latter terminology adapts the use of this term in order theory (see \cite{Davey2002} and compare the discussion in \cite[Remark~2]{AlosFerrer2005}). In the present text the use of forests with multiple connected components, describing exogenous scenarios, is central, and hence using the term ``tree'' for this may be misleading. Hence, the definitions of the present text insist stronger on the graph-theoretical aspect as presented in \cite{Bollobas2013Modern}, and especially on the botanic metaphor of trees and forests, than do those in \cite{AlosFerrer2005,Davey2002}, though without adding unnecessary discreteness assumptions. 
    
    Note, however, that a ``rooted tree'' in the sense of \cite[Definition~1]{AlosFerrer2005} corresponds to a rooted tree in this text because the \cite{AlosFerrer2005}-definition of ``rooted'' demands the existence of a maximum, not only of maximal elements for principal up-sets. But a rooted forest in the sense of this text need not be a ``rooted tree'' in the sense of \cite{AlosFerrer2005}. 

    It should also be noted that we use the partial order $\ge$ rather than $\le$ because of the representation used in the subsequent subsection, following the decision- and game-theoretic texts \cite{AlosFerrer2005,AlosFerrer2016}, though this may differ from the usual convention in other contexts.
\end{remark}

\subsection{Decision forests}

The following definition is, formally, a transcription of the characterisation of (a subclass of) ``game trees'' in \cite[Theorem~3]{AlosFerrer2005}. So, although the object is formally not new, we look at it from a slightly different perspective which is why it is recalled here. First, it emphasises the fact that there can be multiple connected components, but restricts the attention to those \cite{AlosFerrer2005}-``game trees'' whose connected components are rooted. Second, the terminology in this text is different also in that it insists on the purely decision-theoretic aspect. What is called ``decision forest (or tree) over a set'' and ``decision path'' here, respectively, is called ``game tree`` and ``play'' in \cite{AlosFerrer2005}. Third, while \cite{AlosFerrer2005}-``game trees'' are defined via certain set-theoretic properties which can be rather easily verified in applications, and are then characterised via the so-called ``representation by plays'', the present text perceives the decision-theoretic essence of decision forests over sets rather as being the duality between outcomes and nodes expressed by that representation, and thus uses it as the definition.

\begin{definition}\label{def:decision_forest}
    Let $V$ be a set. A \emph{decision forest on $V$} is a $V$-poset $F$ such that:
    \begin{enumerate}
        \item\label{def:decision_forest.rooted_forest} $(F,\supseteq)$ is a rooted forest;
        \item\label{def:decision_forest:repr_by_dec_paths} $F$ is \emph{its own representation by decision paths}, that is, if $W$ denotes the set of maximal chains in $(F,\supseteq)$, and for every $y\in F$, $W(y) = \{w\in W \mid y\in w\}$, then there is a bijection $f\colon V \to W$ such that for every $y\in F$, $(\mc P f)(y) = W(y)$.
    \end{enumerate}
    $F$ is called \emph{decision tree on $V$} iff, in addition, for all $x,y\in F$ there is $z\in F$ with $z\supseteq x\cup y$.

    The nodes, terminal nodes, and moves of $(F,\supseteq)$ are also called \emph{nodes}, \emph{terminal nodes}, and \emph{moves} of $F$, respectively, and the elements of $V$ are called \emph{outcomes}. The set of moves of $F$ is denoted by $X(F)$ or $X$ in short.
\end{definition}

Following \cite{AlosFerrer2005}, $V$ can be seen as the set of possible \emph{outcomes} of the game, and the forest $(F,\supseteq)$ specifies how, as the decision problem is tackled dynamically, the set of realisable outcomes becomes smaller and smaller. Thus, elements of $F$ can be interpreted as nodes. See \cite{AlosFerrer2005} for handy criteria on the set $F$ characterising the situation of it being a decision forest on $V$: essentially, this corresponds to $F$ being a ``game tree'' on $V$ in the sense of \cite[Definition~4]{AlosFerrer2005} such that all connected components of $(F,\supseteq)$ have maximal elements.

Condition \ref{def:decision_forest}.\ref{def:decision_forest.rooted_forest} clarifies that the predecessors of any node $x\in F$ can be totally ordered and contain a root, and therefore constitute an account of the past with a beginning, a history. The set $W$ of maximal chains can be interpreted as the set of \emph{decision paths}. Condition \ref{def:decision_forest}.\ref{def:decision_forest:repr_by_dec_paths} says that possible outcomes and decision paths are in one-to-one correspondence, in such a way that a node is contained in some decision path iff the corresponding outcome is contained in that node. Even more is true: according to \cite[Theorem~3, Corollary~2]{AlosFerrer2005} and Remark~\ref{rmk:def_tree_AFK}, the bijection $f$ is uniquely determined. More precisely:

\begin{proposition}[\cite{AlosFerrer2005}]\label{prop:f(v)=uparrow v}
    Let $F$ be a decision forest on some set $V$ and $f\colon V \to W$ be a map as in Definition~\ref{def:decision_forest}, where $W$ is the set of maximal chains in $(F,\supseteq)$. Then,
    \[ \forall v\in V\colon\quad f(v) = \uparrow \{v\} = \{x\in F \mid v\in x\}. \]
\end{proposition}
Thus $V$, which formally is the set of possible outcomes, and $W$, which formally is the set of possible decision paths, can be identified in a uniquely determined way, namely via the $f$ above. Under this identification, we have the following duality statement:
\[ x\in f(v) \qquad \Longleftrightarrow \qquad v\in x, \]
for all nodes $x\in F$ and all outcomes $v\in V$. As a consequence, in the subsequent sections of the paper the set a decision forest is defined on is denoted by $W$ rather than by $V$, consistent with other pieces of the literature, notably with \cite{AlosFerrer2005,AlosFerrer2016}. \smallskip

There is a second duality, namely between two approaches to dynamic decision theory: graphs, based on tree-like objects, and refined partitions, based on a set of outcomes. This is made precise in the following two propositions, which are contained in \cite[Lemma 14 and Theorem~1]{AlosFerrer2005}, following Remark~\ref{rmk:def_tree_AFK}:

\begin{proposition}[\cite{AlosFerrer2005}]\label{prop:decision_forest_over_set_is_decision_forest}
    Let $F$ be a decision forest on a set $V$. Then the poset $(F,\supseteq)$ is a decision forest.
\end{proposition}

\begin{proposition}[\cite{AlosFerrer2005}]\label{prop:repr_by_dec_paths_of_decision_forest}
     Let $(F_0,\ge)$ be a decision forest. Let $V$ be the set of its maximal chains and for any $x_0\in F_0$ let $V(x_0)$ be the set of $v\in V$ with $x_0\in v$. Let \[F = \{V(x_0) \mid x_0\in F_0\}.\] Then $F$ defines a decision forest on $V$, and $(F,\supseteq)$ is order-isomorphic to $(F_0,\ge)$.
\end{proposition}

Rephrasing the ``representation by plays'' from \cite{AlosFerrer2005}, we call $F$ the \emph{representation by decision paths} of $(F_0,\ge)$. The definition of a decision forest on a set requires essentially that the operations from the two preceding proposition are, up to isomorphism, inverse to each other. See \cite[Section~4]{AlosFerrer2005} for more details on this.

Thus, there is a rigorous sense in that both of the mentioned descriptions are equivalent. While the graph-theoretical approach is graphically more convenient, the refined partitions approach is naturally aligned with Savage's acts as a model for choice under uncertainty (see the introduction of \cite{AlosFerrer2005} for a discussion). We note that, although phrased differently and with a different aim, all three preceding propositions and the duality concepts constitute one of the essential innovations of \cite{AlosFerrer2005}. 

\subsection{Forests of decision trees}

In the present text, it is crucial to deal with forests and not only with trees, in the sense of Subsection~\ref{subs:conventions} and Remark~\ref{rmk:def_tree_AFK}. Often, the study of forests can be reduced to the study of trees by considering the connected components separately. Is this true for decision forests as well, and in what sense? In other words, is the duality of nodes and outcomes compatible with the forest structure? This question gets fully answered in the following theorem. More precisely, we have:

\begin{thm}\label{thm:decision_forest=forest_of_decision_trees}
	Let $V$ be a set and $F$ be a $V$-poset defining a rooted forest with respect to $\supseteq$. Let $\mc F$ denote the set of connected components of the rooted forest $(F,\supseteq)$. For every $T\in\mc F$, let $V_T$ denote the root of $(T,\supseteq)$.
	
	Then $F$ is a decision forest on $V$ iff 
    \begin{enumerate}
        \item\label{thm:decision_forest=forest_of_decision_trees.partition} $\{V_T\mid T\in\mc F\}$ is a partition of $V$, and
        \item\label{thm:decision_forest=forest_of_decision_trees.T_decision_tree} for every $T\in\mc F$, $(T,\supseteq)$ defines a decision tree on $V_{T}$.
    \end{enumerate}
\end{thm}

To verify the Propositions \ref{prop:f(v)=uparrow v}, \ref{prop:decision_forest_over_set_is_decision_forest}, and \ref{prop:repr_by_dec_paths_of_decision_forest}, on can cite the results in \cite{AlosFerrer2005} in combination with Remark~\ref{rmk:def_tree_AFK}. In view of Theorem~\ref{thm:decision_forest=forest_of_decision_trees}, it suffices, however, to cite the results for the case of rooted trees only. This will be of interest at a later stage when analysing outcome existence and uniqueness for stochastic extensive forms.

\begin{remark}\label{rmk:decision_forests_over_sets}
    Theorem~\ref{thm:decision_forest=forest_of_decision_trees} offers a way of constructing decision forests from collections of decision trees. Let $\mc F_0$ be a non-empty set of decision trees on sets. For any $T_0\in\mc F_0$, let $V_{T_0}$ be its root. Then, let
    \[ V = \bigcup_{T_0\in\mc F_0} V_{T_0} \times \{T_0\} = \{ (v,T_0) \mid T_0 \in \mc F_0,\, v\in T_0\}, \]
    the disjoint union of all roots, and let
    \[ F = \big\{ x \times \{T_0\} \mid T_0 \in \mc F_0,\, x\in T_0\big\}. \]
    Then, $F$ is a decision forest on $V$. The set of connected components is given by
    \[ \mc F = \Big\{ \big\{ x \times \{T_0\} \mid x\in T_0\big\} \bigmid T_0 \in \mc F_0\Big\}. \]
    Moreover, Theorem~\ref{thm:decision_forest=forest_of_decision_trees} states that all decision forests on sets can be represented in this form.
\end{remark}

For examples of decision forests on sets, the reader is referred to the following section, where decision forests whose connected components are indexed over the set of exogenous scenarios are considered.\footnote{Also note the examples of decision trees on sets in \cite{AlosFerrer2005} and the more extensive version in \cite{AlosFerrer2016} in combination with the preceding Remark~\ref{rmk:decision_forests_over_sets}.}

\section{Stochastic decision forests}\label{sec:sdf}

In this section, the central notion of this paper -- stochastic decision forests -- is introduced. The main idea behind this is to weaken the traditional assumptions on exogenous information, which is no more assumed to arise through the dynamic decision making of a nature agent. Rather, an exogenous scenario $\omega$ is ``randomly'' realised within a given measurable space which determines the decision tree underlying the actual decision makers' problem. Therefore, before the main definition, measurable spaces are discussed as a model for exogenous scenarios and sets of scenarios whose probability can be measured. Subsequently, the definition of stochastic decision forests is given and analysed, simple examples are presented, and the class of action path stochastic decision forests is constructed.

\subsection{Exogenous scenario spaces}\label{subs:Ex_sc_sp}

For the remainder of the paper, an \emph{exogenous scenario space} is a measurable space $(\Omega,\ms E)$ such that $\Omega\neq\emptyset$.\footnote{For more information on measure and probability theory, the reader is referred to the introductory texts \cite{Bogachev2007Measure,Kallenberg2021}.} $\Omega$ describes the set of possible \emph{exogenous scenarios} which replace the outcomes possibly generated by the behaviour of a ``nature'' agent. $\ms E$ is the set of \emph{events} alias collections of exogenous scenarios that can be measured by all relevant decision makers. That is, we implictly suppose that decision makers alias agents form beliefs about the probability of events, in the sense of probability theory, and here, $\ms E$ describes the set of collections of exogenous scenarios that agents assign probabilities to.\footnote{These beliefs can be interpreted as representation devices of agents' preferences over the random outcomes of strategic interaction in terms of expected utility -- ``random'' because of the dependence on exogenous scenarios, which in this particular context play the role of states while outcomes correspond to the consequences, put in Savage's words, see, e.g.\ \cite{Savage1972Foundations,Gilboa1989Maxmin,Hara2023Multiple}.}

Why does one suppose $\ms E$ to be a $\sigma$-algebra? As an algebra of sets, $\ms E$ can be seen as a Boolean algebra and can therefore model basic logical operations on events. Moreover, in the probabilistic approach, we are only interested in those functions $\P\colon\ms E\to [0,1]$ compatible with the algebra structure\footnote{The meaning of the phrase ``compatible with the algebra structure'' can be reduced to the statement: a) $\P(\Omega) = 1$, and b) for all disjoint $E_1,E_2\in\ms E$ we have $\P(E_1\cup E_2) = \P(E_1) + \P(E_2)$.} that are $\sigma$-additive, which can be thought of as a continuity property, and it is a fundamental measure-theoretic result, that, for such $\P$, without loss of generality, $\ms E$ can be supposed to be closed under countable unions, that is, to be a $\sigma$-algebra (see \cite[Section~1.5]{Bogachev2007Measure}).

If $\ms E$ can be generated by a countable partition $\ms P$ of $\Omega$, any probability measure on it is entirely described through the probabilities on these partitioning events. So, though all elements of $\ms E$ including unions of partition members are events, one can formally describe the situation by the set $\ms P$ and a countable family of positive numbers $(p_E)_{E\in \ms P}$ adding up to $1$ and indexed by that set. However, in general, $\ms E$ need not be generated by a partition of $\Omega$. For instance, it may be that all singletons are events and have probability zero, as in the case of the Lebesgue measure on the unit interval $[0,1]$. Hence, from a general perspective, scenarios may be neither a relevant nor a sufficient description of exogenous data for decision makers, but it is the ($\sigma$-)algebra $\ms E$ of events that these agents are concerned about. Moreover, as famously illustrated by the Lebesgue measure, in case $\Omega$ is uncountable, there is a non-trivial trade-off between the fineness\footnote{Recall that a $\sigma$-algebra $\ms E'$ on $\Omega$ is \emph{finer} than $\ms E$ iff $\ms E \subseteq \ms E'$.} 
of $\ms E$ and the amount of probability measures on it (see \cite[Section~1.7]{Bogachev2007Measure}). Hence, to describe the agent's relation to exogenous data and admissible beliefs about their probability, it is crucial to specify $\ms E$. This is crucial, not the least because many relevant applications require an uncountable $\Omega$ (with a $\sigma$-algebra $\ms E$ that cannot be generated by a countable partition), be it the uniform distribution on $[0,1]$, be it continuous martingales omnipresent in finance (see, e.g.\ \cite{Delbaen2006Mathematics,Bayer2023Rough}).

\subsection{Stochastic decision forests}
In light of the preceding argument, we model uncertainty about exogenous events by an exogenous scenario space $(\Omega,\ms E)$ and suppose that some $\omega\in\Omega$ is realised, determining the relevant decision tree without being directly communicated to the decision makers. So, the agents do not necessarily have complete information about exogenous events in $\ms E$, which tell in which tree they are while making choices. Note that events in $\ms E$ need not be generated by a partition of $\Omega$. Hence, there must be a structure of similarity among trees that can serve as a consistent basis for exogenous information revelation in an abstract and general probabilistic sense. This distinguishes the present approach from the traditional model of games or decision problems under incomplete information (\cite{Harsanyi1967,Harsanyi1968a,Harsanyi1968b}) or stochastic games (\cite{Shapley1953}, see also \cite[Subsection~2.2.2.5]{AlosFerrer2016}), represented via a randomising nature agent and partition-based information sets with respect to that agent's past action.

\begin{definition}\label{def:sdf}
    A \emph{stochastic decision forest}, in short \emph{\textsc{sdf}}, on an exogenous scenario space $(\Omega,\ms E)$ is a triple $(F,\pi,\X)$ consisting of:
    \begin{enumerate}
        \item\label{def:sdf.df} a decision forest $F$ on some set $W$;
        \item\label{def:sdf.conn_comp} a surjective map $\pi\colon F\to \Omega$ such that the set $\mc F$ of connected components of $(F,\supseteq)$ is given by the fibres of $\pi$, that is,
        \[ \mc F = \{\pi^{-1}(\{\omega\}) \mid \omega\in \Omega\}; \]
        \item\label{def:sdf.X} a set $\X$ such that: 
        \begin{enumerate}
            \item\label{def:sdf.X.section} any element $\x\in\X$ is a section of moves defined on some non-empty event, that is, it is a map $\x\colon D_\x \to X$ satisfying $\pi\circ\x = \id_{D_\x}$ for some $D_\x\in\ms E\setminus\{\emptyset\}$;
            \item\label{def:sdf.X.cov} $\X$ induces a covering of $X$, that is, $\{\x(\omega) \mid \x\in\X,\,\omega\in D_\x\} = X$.
        \end{enumerate}
    \end{enumerate}
    The elements of $\X$ are called \emph{random moves}. For $\omega\in\Omega$, let $T_\omega = \pi^{-1}(\{\omega\})$ and $W_\omega$ be the root of $T_\omega$. For $E\subseteq\Omega$, let $W_E = \bigcup_{\omega\in E} W_\omega$ and $F_E = \bigcup_{\omega\in E} T_\omega$.
\end{definition}

In other words, a stochastic decision forest is a decision forest on a set (Axiom~\ref{def:sdf.df}), whose connected components are indexed by $\Omega$ (Axiom~\ref{def:sdf.conn_comp}), admitting a set of sections of moves (Axiom~\ref{def:sdf.X.section}) called random moves that form a covering (Axiom~\ref{def:sdf.X.cov}) of $X$. As discussed in the remainder of this series of papers, random moves form a flexible basis for a general description of information revelation for different agents, under the innocent hypothesis that there is at least one active agent per move. 

If one visualises a stochastic decision forest in two dimensions, any tree growing along the vertical ``decision path'' axis, the forest's trees being placed along the horizontal ``$\Omega$'' axis, then random moves are -- roughly speaking -- horizontal or diagonal sections covering the set of moves in a decision-theoretically interpretable way. Hence, in a stochastic decision forest, the flow of information can be decomposed into the endogenous movement of agents along the set of random moves caused by decision making and the exogenous information revelation at those random moves. Note also that stochastic decision forests are allowed to vary across scenarios: the trees in $\mc F$ need not be isomorphic. The crowns may become shallower in some scenarios (because some options are no more available there).

It is possible that random moves as such can be ordered in a way consistent with the ordering on $F$. This means that one can relate these points of exogenous information revelation uniformly by means of the words ``before'' and ``after''. This possibility is necessary for conceiving an agent's capacity to condition action on knowledge about exogenous information received ``earlier''. Note that this order consistency of random moves is not about information sets and, in particular, not to be confounded with the strong order on information sets as introduced in \cite{Ritzberger1999Recall}. An example of an order inconsistent stochastic decision forest is presented in the following subsection (namely, Gilboa's interpretation of the absent-minded driver formulated as a stochastic decision forest).
Furthermore, we introduce two non-triviality conditions.
In many contexts it seems reasonable to assume first that any root is a move -- a situation which (provided sufficient measurability) can always be obtained by eliminating those scenarios $\omega$ with singleton $T_\omega$ alias singleton $W_\omega$ --, and second that the set of random moves describes the basis of exogenous information most efficiently, in the sense that it cannot be extended without compromising its structure.

\begin{definition}\label{def:sdf.addon}
    Given a stochastic decision forest $(F,\pi,\X)$ on an exogenous scenario space $(\Omega,\ms E)$, let $\ge_\X$ denote the partial on $\X$ defined by
    \[ \x_1 \ge_\X \x_2 \quad \Longleftrightarrow \quad \Big[ D_{\x_1} \supseteq D_{\x_2}~ \text{ and }~\forall \omega\in D_{\x_2}\colon \x_1(\omega) \supseteq \x_2(\omega)\Big]. \]
    A set $\tilde\X\subseteq\X$ of random moves is said \emph{order consistent} iff for all $\x_1,\x_2\in\tilde\X$:
    \begin{equation*}
        \Big[\exists\omega\in D_{\x_1}\cap D_{\x_2}\colon~ \x_1(\omega) \supseteq \x_2(\omega)\Big] \qquad \Longrightarrow \qquad \x_1 \ge_\X \x_2.
    \end{equation*}

    A stochastic decision forest $(F,\pi,\X)$
    \begin{enumerate}[label=4(\alph*),ref=4(\alph*)]
        \item\label{def:sdf.X.OC} is said \emph{order consistent} iff $\X$ is order consistent;
        \item\label{def.sdf.X.surely_NT} is said \emph{surely non-trivial} iff $(F,\supseteq)$ is everywhere non-trivial;
        \item\label{def:sdf.X.max} that is order consistent, is said \emph{maximal} iff for every set $\bar\X$ such that $(F,\pi,\bar\X)$ is an order consistent stochastic decision forest and that is \emph{refined by $\X$} in that for all $\bar\x\in\bar\X$ there is $P_{\bar\x}\subseteq \X$ with $\bar\x = \bigcup P_{\bar\x}$,\footnote{According to standard set-theoretic conventions, $\bar\x = \bigcup P_{\bar\x}$ means: $\bar\x$ is a map with domain $\bigcup_{\x\in P_{\bar\x}} D_\x$ and for all $\x\in P_{\bar\x}$ and $\omega\in D_\x$, $\bar\x(\omega) = \x(\omega)$.} we have $\bar\X = \X$.
    \end{enumerate}
\end{definition}

As is going to be shown in the sequel, some of these properties are actually stronger than the definition formally indicates. In the case of order consistency, random moves cannot cross each other, and in particular, the root property is compatible with random moves.

\begin{lemma}\label{lemma:sdf.X.roots} 
    Let $(F,\pi,\X)$ be a stochastic decision forest. Let $D = \{\omega\in\Omega \mid W_\omega\in X\}$.
    \begin{enumerate}
        \item\label{lemma:sdf.X.roots.OC}  If $(F,\pi,\X)$ is order consistent, then, for any $\x\in\X$, $\x(\omega)$ is a root in $(F,\supseteq)$ for all or no $\omega\in D_\x$. 
        \item\label{lemma:sdf.X.roots.surely_non-trivial}  $D=\Omega$ iff $(F,\pi,\X)$ is surely non-trivial.
        \item\label{lemma:sdf.X.roots.root_random_move} If $(F,\pi,\X)$ is order consistent and maximal, and $X\neq\emptyset$, then $\x_0\colon D \to X,\,\omega\mapsto W_\omega$ is a random move, i.e.\ $\x_0\in\X$.
    \end{enumerate}
\end{lemma}

Moreover, in the order consistent, surely non-trivial and maximal case, random moves can be seen as moves of a rooted decision tree in their own right. This makes the endogenous movement of agents along $\X$ particularly ordered: as said before, it becomes possible to relate two points of exogenous information revelation via ``before'' and ``after''. Let us make this precise. First, we might have to add terminal nodes, if they exist. For example, think of an ``infinite centipede'' which contains infinitely many maximal chains but only one maximal chain consisting only of moves. Hence, we extend the partial order $\ge_\X$ to a binary relation $\ge_\Tr$ on the disjoint union \[\Tr = \X \cup\Big \{\big\{(\omega,\{w\})\big\}\bigmid (\omega,w)\in\Omega\times W\colon \{w\} \in F,\,\pi(\{w\}) = \omega\Big\}\] by letting, for $\x\in\X$ and $(\omega,w),(\omega',w')\in\Omega\times W$ with $\{w\},\{w'\}\in F$ and $\pi(\{w\})=\omega$, $\pi(\{w'\})=\omega'$:
\begin{itemize}[label=--]
    \item $\x\ge_\Tr \big\{(\omega,\{w\})\big\}$ iff $\omega\in D_\x$ and $w\in\x(\omega)$;
    \item $\big\{(\omega,\{w\})\big\} \ge_\Tr \big\{(\omega',\{w'\})\big\}$ iff $(\omega,w)=(\omega',w')$;
    \item $\big\{(\omega,\{w\})\big\}\ngeq_\Tr \x$.
\end{itemize}
Note that a set $\big\{(\omega,\{w\})\big\}$ as above is nothing else than the set-theoretic function $\y\colon D_\y \to F$ with $D_\y = \{\omega\}$ and $\y(\omega)=\{w\}$. In analogy with random moves and since its image is a terminal node, we call such $\y$ a \emph{random terminal node}. The analogy is substantiated by the subsequent proposition which (seemingly) strengthens Axiom~\ref{def:sdf.X.OC} in Definition~\ref{def:sdf.addon} above.

\begin{proposition}\label{prop:ev_on_Tr_is_iso}
    Let $(F,\pi,\X)$ be an order consistent stochastic decision forest on some exogenous scenario space $(\Omega,\ms E)$ and $\Tr \bullet \Omega = \{(\y,\omega) \in \Tr\times \Omega \mid \omega \in D_\y\}$. Then
    the evaluation map $\ev\colon \Tr\bullet \Omega \to F, (\y,\omega) \mapsto \y(\omega)$ is a bijection such that for all $(\y_1,\omega_1),(\y_2,\omega_2)\in\Tr \bullet \Omega$:
    \[ \Big[\y_1 \ge_\Tr \y_2 \text{ and } \omega_1 = \omega_2\Big] \quad \Longleftrightarrow \quad \y_1(\omega_1) \supseteq \y_2(\omega_2). \]
\end{proposition}

In order-theoretic terms, the evaluation map defines an order isomorphism between $\Tr\bullet\Omega$, equipped with the order induced by the product of $\ge_\Tr$ on $\Tr$ and equality on $\Omega$, and $(F,\supseteq)$. 
Now, as announced, we can make rigorous the sense in that random moves are the moves of a rooted decision tree, under the assumptions presented in Definition~\ref{def:sdf.addon}.

\begin{thm}\label{thm:Xrm_is_dec_tree}
    Let $(F,\pi,\X)$ be an order consistent, surely non-trivial and maximal stochastic decision forest on some exogenous scenario space $(\Omega,\ms E)$. 
    Then, $(\Tr,\ge_\Tr)$ defines a rooted decision tree and $\X$ is the set of its moves.
\end{thm}

Note, however, that $F$ cannot be subsumed under this derived object. The latter as such does not faithfully account for the fact that there is a realised scenario. For instance, only those maximal chains $\mathbf c$ in $(\Tr,\ge_\Tr)$ correspond to outcomes in $W = \bigcup F$ that satisfy $\bigcap_{\y\in\mathbf c} D_\y \neq \emptyset$, which need not be the case. This is developed a bit further in the second part, \cite[Subsection~3.1]{Rapsch2024DecisionB}.

\subsection{Simple examples}\label{subs:simple_sdf}
In the following we present three very simple examples of stochastic decision forests. The first example is illustrated in Figure~\ref{fig:simple_sdf}. 
\begin{figure}
    \centering
    \begin{tikzpicture}[node distance={20mm}, thick, main/.style = {draw, circle}] 
    \node[main] (1) {$\x_0(\omega_1)$}; 
    \node[main] (2) [above left of=1] {$\x_1(\omega_1)$}; 
    \node[main] (3) [above right of=1] {$\x_2(\omega_1)$}; 
    \node[main] (4) [right of=3] {$\x_1(\omega_2)$}; 
    \node[main] (5) [below right of=4] {$\x_0(\omega_2)$};
    \node[main] (6) [above right of=5] {$\x_2(\omega_2)$};
    \node[] (8) [above=0.5cm of 2] {$\{w_{112}\}$};
    \node[] (7) [left=0.2cm of 8] {$\{w_{111}\}$};
    \node[] (10) [above=0.5cm of 3] {$\{w_{122}\}$};
    \node[] (9) [left=0.2cm of 10] {$\{w_{121}\}$};
    \node[] (11) [above=0.5cm of 4] {$\{w_{211}\}$};
    \node[] (12) [right=0.2cm of 11] {$\{w_{212}\}$};
    \node[] (13) [above=0.5cm of 6] {$\{w_{221}\}$};
    \node[] (14) [right=0.2cm of 13] {$\{w_{222}\}$};
    \draw[->] (1) -- (2); 
    \draw[->] (1) -- (3); 
    \draw[->] (5) -- (4); 
    \draw[->] (5) -- (6); 
    \draw[->] (2) -- (7); 
    \draw[->] (2) -- (8); 
    \draw[->] (3) -- (9); 
    \draw[->] (3) -- (10); 
    \draw[->] (4) -- (11); 
    \draw[->] (4) -- (12); 
    \draw[->] (6) -- (13); 
    \draw[->] (6) -- (14); 
    \end{tikzpicture} 
    \caption{A simple stochastic decision forest represented as a directed graph, with $w_{\ell km} = (\omega_\ell,k,m)$, for $(\ell,k,m)\in \{1,2\}^3$. Moves are indicated by circles.}
    \label{fig:simple_sdf}
\end{figure}
It indicates \emph{pars pro toto} how finite stochastic extensive form decision problems can be formalised.

Let $\Omega = \{\omega_1,\omega_2\}$ be some exogenous scenario space with two elements and $\ms E = \mc P \Omega$. Let $W = \Omega \times \{1,2\}^2$ and $\x_0,\x_1,\x_2\colon \Omega \to \mc P(W)$ given by $\x_0(\omega) = \{\omega\} \times \{1,2\}^2$ and $\x_k(\omega) = \{(\omega,k)\}\times \{1,2\}$, $F = \{\x_k(\omega) \mid \omega\in\Omega,~k=0,1,2\} \cup \{ \{w\} \mid w\in W\}$, $\pi\colon F \to\Omega$ be the map sending any node to the first entry of an arbitrary choice among its elements.

\begin{lemma}\label{lemma:simple_sdf1}
    The tuple $(F,\pi,\X)$ defines an order consistent, surely non-trivial, and maximal stochastic decision forest.
\end{lemma}

The corresponding decision tree $(\Tr,\ge_\Tr)$ is illustrated in Figure~\ref{fig:simple_sdf_Tr}.
\begin{figure}
    \centering
    \begin{tikzpicture}[node distance={20mm}, thick, main/.style = {draw, circle}] 
    \node[main] (1) {$\x_0$}; 
    \node[main] (2) [above left=0.5cm and 2.3cm of 1] {$\x_1$}; 
    \node[main] (3) [above right=0.5cm and 2.3cm of 1] {$\x_2$};
    \node[] (9) [above=0.5cm of 2] {$\{w_{212}\}_{\{\omega_2\}}$};
    \node[] (8) [left=0cm of 9] {$\{w_{112}\}_{\{\omega_1\}}$};
    \node[] (7) [left=0cm of 8] {$\{w_{111}\}_{\{\omega_1\}}$};
    \node[] (10) [right=0cm of 9] {$\{w_{211}\}_{\{\omega_2\}}$};
    \node[] (12) [above=0.5cm of 3] {$\{w_{122}\}_{\{\omega_1\}}$};
    \node[] (13) [right=0cm of 12] {$\{w_{221}\}_{\{\omega_2\}}$};
    \node[] (11) [left=0cm of 12] {$\{w_{121}\}_{\{\omega_1\}}$};
    \node[] (14) [right=0cm of 13] {$\{w_{222}\}_{\{\omega_2\}}$};
    \draw[->] (1) -- (2); 
    \draw[->] (1) -- (3); 
    \draw[->] (2) -- (7); 
    \draw[->] (2) -- (8); 
    \draw[->] (2) -- (9); 
    \draw[->] (2) -- (10); 
    \draw[->] (3) -- (11); 
    \draw[->] (3) -- (12); 
    \draw[->] (3) -- (13); 
    \draw[->] (3) -- (14); 
    \end{tikzpicture} 
    \caption{The decision tree $(\Tr,\ge_\Tr)$ for the simple stochastic decision forest, with $w_{\ell km} = (\omega_\ell,k,m)$, for $(\ell,k,m)\in \{1,2\}^3$. (Random) moves are indicated by circles. Elements of $\Tr \setminus \X$, of the form $\{(\omega,\{w\})\}$ and seen as maps $\omega\mapsto \{w\}$, are denoted by $\{w\}_{\{\omega\}}$.}
    \label{fig:simple_sdf_Tr}
\end{figure}
\smallskip

As a variant, identifying the elements $(\omega_1,2,1)$ and $(\omega_1,2,2)$ in $W$ is seen below to provide a stochastic decision forest with a random move that is not defined on all of $\Omega$, as illustrated in Figure~\ref{fig:simple_sdf_variant}. 
\begin{figure}
    \centering
    \begin{tikzpicture}[node distance={20mm}, thick, main/.style = {draw, circle}] 
    \node[main] (1) {$\x'_0(\omega_1)$}; 
    \node[main] (2) [above left=0.6cm and 1cm of 1] {$\x'_1(\omega_1)$}; 
    \node[] (3) [above=2.1cm of 1] {$\{w'_{12}\}$}; 
    \node[main] (5) [right of=1] {$\x'_0(\omega_2)$};
    \node[main] (4) [above=0.2cm of 5] {$\x'_1(\omega_2)$}; 
    \node[main] (6) [above right=0.6cm and 2cm of 5] {$\x'_2(\omega_2)$};
    \node[] (8) [above=0.5cm of 2] {$\{w'_{112}\}$};
    \node[] (7) [left=0.2cm of 8] {$\{w'_{111}\}$};
    \node[] (11) [above=0.5cm of 4] {$\{w'_{211}\}$};
    \node[] (12) [right=0.2cm of 11] {$\{w'_{212}\}$};
    \node[] (13) [above=0.5cm of 6] {$\{w'_{221}\}$};
    \node[] (14) [right=0.2cm of 13] {$\{w'_{222}\}$};
    \draw[->] (1) -- (2); 
    \draw[->] (1) -- (3); 
    \draw[->] (5) -- (4); 
    \draw[->] (5) -- (6);  
    \draw[->] (2) -- (7); 
    \draw[->] (2) -- (8); 
    \draw[->] (4) -- (11); 
    \draw[->] (4) -- (12); 
    \draw[->] (6) -- (13); 
    \draw[->] (6) -- (14); 
    \end{tikzpicture} 
    \caption{A variant of the simple stochastic decision forest in Figure~\ref{fig:simple_sdf} represented as a directed graph, with $w'_{\ell km} = (\omega_\ell,k,m)$, for all triples $(\ell,k,m)\in\{1,2\}^3$ with $(\omega_\ell,k,m)\in W'$, and $w'_{12} = (\omega_1,2)$. Moves are indicated by circles.}
    \label{fig:simple_sdf_variant}
\end{figure}
Put differently, let $W' = W\setminus\{(\omega_1,2,1),(\omega_1,2,2)\} \cup \{(\omega_1,2)\}$ and let $\x'_0 = \x_0$, $\x'_1 = \x_1$, and $\x'_2\colon \rho^{-1}(\{2\}) \to \mc P(W')$ be given by $\x'_2(\omega) = \{(\omega,2)\}\times \{1,2\}$. Let $\X' = \{\x'_0,\x'_1,\x'_2\}$. Let $D_{\x'_0}=\Omega$, $D_{\x'_1} = \Omega$, $D_{\x'_2}=\{\omega_2\}$ and $F'=\{\x'(\omega) \mid \x'\in\X',~\omega\in D_{\x'}\} \cup \{\{w'\}\mid w'\in W'\}$. Let $\pi'\colon F'\to\Omega$ be the map sending any node to the first entry of an arbitrary choice among its elements.

\begin{lemma}\label{lemma:simple_sdf2}
    The tuple $(F',\pi',\X')$ defines an order consistent, surely non-trivial, and maximal stochastic decision forest.
\end{lemma}

The corresponding decision tree $(\Tr',\ge_{\Tr'})$ is illustrated in Figure~\ref{fig:simple_sdf_variant_Tr}.
\begin{figure}
    \centering
    \begin{tikzpicture}[node distance={20mm}, thick, main/.style = {draw, circle}] 
    \node[main] (1) {$\x'_0$}; 
    \node[main] (2) [above left=0.6cm and 1.8cm of 1] {$\x'_1$}; 
    \node[] (3) [above right=1.8cm and 0.5cm of 1] {$\{w'_{12}\}_{\{\omega_1\}}$};  
    \node[main] (6) [above right=0.6cm and 1.8cm of 1] {$\x'_2$};
    \node[] (8) [above left=0.6cm and 0.5cm of 2] {$\{w'_{112}\}_{\{\omega_1\}}$};
    \node[] (7) [left=0cm of 8] {$\{w'_{111}\}_{\{\omega_1\}}$};
    \node[] (11) [right=0cm of 8] {$\{w'_{211}\}_{\{\omega_2\}}$};
    \node[] (12) [right=0cm of 11] {$\{w'_{212}\}_{\{\omega_2\}}$};
    \node[] (13) [above right=0.6cm and 0.5cm of 6] {$\{w'_{221}\}_{\{\omega_2\}}$};
    \node[] (14) [right=0cm of 13] {$\{w'_{222}\}_{\{\omega_2\}}$};
    \draw[->] (1) -- (2); 
    \draw[->] (1) -- (3); 
    \draw[->] (1) -- (6);  
    \draw[->] (2) -- (7); 
    \draw[->] (2) -- (8); 
    \draw[->] (2) -- (11); 
    \draw[->] (2) -- (12); 
    \draw[->] (6) -- (13); 
    \draw[->] (6) -- (14); 
    \end{tikzpicture} 
    \caption{The decision tree $(\Tr',\ge_{\Tr'})$ for the variant of the simple stochastic decision forest, with $w'_{\ell km} = (\omega_\ell,k,m)$, for all triples $(\ell,k,m)\in \{1,2\}^3$ with $(\omega_\ell,k,m) \in W'$, and $w'_{12} = (\omega_1,2)$. (Random) moves are indicated by circles. Elements of $\Tr' \setminus \X'$, of the form $\{(\omega,\{w'\})\}$ and seen as maps $\omega\mapsto \{w'\}$, are denoted by $\{w'\}_{\{\omega\}}$.}
    \label{fig:simple_sdf_variant_Tr}
\end{figure}
\smallskip

The third example is an \textsc{sdf} representation of Gilboa's model of the absent-minded driver phenomenon in \cite{Gilboa1997Comment}. For this, let $(\Omega,\ms E)$ be some exogenous scenario space and $\rho\colon \Omega \to \{1,2\}$ an $\ms E$-$\mc P\{1,2\}$-measurable surjection. Further, let $D,H,M$ be three pairwise distinct symbols meaning ``disastrous region'', ``home'', and ``motel'' as in the original story from \cite{Piccione1997Interpretation}. Let $W = \Omega \times \{D,H,M\}$ and $\x_1,\x_2$ be maps defined on the whole of $\Omega$ given by
\[ \x_k(\omega) = \begin{cases} \{\omega\} \times \{D,H,M\}, &\text{ if } \rho(\omega) = k, \\ \{\omega\} \times \{H,M\}, & \text{ if } \rho(\omega) = 3-k, \end{cases} \qquad k\in\{1,2\}. \]
$F$ is defined as the union of the images of $\x_1$ and $\x_2$ and the set of all singleton sets in $W$. $\pi\colon F \to \Omega$ maps any element of $F$ to the first component of its elements. $(F,\supseteq)$ and $(\Tr,\ge_\Tr)$ are illustrated in Figure \ref{fig:absent_minded_driver_Gilboa_sdf}. Note that $(\Tr,\ge_\Tr)$ is not even a forest, both $\x_1$ and $\x_2$ are maximal elements, and both $\im\x_1$ and $\im\x_2$ contain both a root and a move that is not a root, respectively.

\begin{figure}
    \centering
    \begin{tikzpicture}[node distance={20mm}, thick, main/.style = {draw, circle}] 
    \node[main] (11) {$\x_1(\omega_1)$}; 
    \node[] (1D) [left=0.5 of 11] {$\{w_{1D}\}$};
    \node[main] (12) [above=0.5cm of 11] {$\x_2(\omega_1)$}; 
    \node[] (1H) [left=0.5 of 12] {$\{w_{1H}\}$};
    \node[] (1M) [above=0.5cm of 12] {$\{w_{1M}\}$}; 
    \node[main] (22) [right=0.5cm of 11] {$\x_2(\omega_2)$}; 
    \node[] (2D) [right=0.5 of 22] {$\{w_{1D}\}$};
    \node[main] (21) [above=0.5cm of 22] {$\x_1(\omega_2)$}; 
    \node[] (2H) [right=0.5 of 21] {$\{w_{1H}\}$};
    \node[] (2M) [above=0.5cm of 21] {$\{w_{1M}\}$}; 
    \draw[->] (11) -- (12); 
    \draw[->] (11) -- (1D); 
    \draw[->] (12) -- (1H); 
    \draw[->] (12) -- (1M);  
    \draw[->] (22) -- (21); 
    \draw[->] (22) -- (2D); 
    \draw[->] (21) -- (2H); 
    \draw[->] (21) -- (2M);  
    \end{tikzpicture} 
    \qquad
    \begin{tikzpicture}[node distance={20mm}, thick, main/.style = {draw, circle}] 
        \node[] (0) {};
        \node[main] (1) [above=0.5 of 0] {$\x_1$};
        \node[] (1D) [left=0.5cm of 1] {$\{w_{1D}\}_{\{\omega_1\}}$};
        \node[] (1H) [above left=0.5cm and 0.5cm of 1] {$\{w_{1H}\}_{\{\omega_1\}}$};
        \node[] (1M) [above=0.5cm of 1] {$\{w_{1M}\}_{\{\omega_1\}}$};  
        \node[main] (2) [right=1 of 1] {$\x_2$};
        \node[] (2D) [right=0.5cm of 2] {$\{w_{2D}\}_{\{\omega_2\}}$};
        \node[] (2H) [above right=0.5cm and 0.5cm of 2] {$\{w_{2H}\}_{\{\omega_2\}}$};
        \node[] (2M) [above=0.5cm of 2] {$\{w_{2M}\}_{\{\omega_2\}}$}; 
        \draw[->] (1) -- (1D);
        \draw[->] (1) -- (1H);
        \draw[->] (1) -- (1M);
        \draw[->] (1) -- (2M);
        \draw[->] (1) -- (2H);
        \draw[->] (2) -- (2M);
        \draw[->] (2) -- (1H);
        \draw[->] (2) -- (1M);
        \draw[->] (2) -- (2D);
        \draw[->] (2) -- (2H);
    \end{tikzpicture}
    \caption{The absent-minded driver \textsc{sdf}, following Gilboa: $(F,\supseteq)$ represented as a directed graph, in case $\rho$ is injective, with $\omega_\ell = \rho^{-1}(\ell)$, $w_{\ell S} = (\omega_\ell,S)$, for all $\ell\in \{1,2\}$ and symbols $S$ such that $(\omega,S)\in W$ (left), $(\Tr,\ge_\Tr)$ where elements of $\Tr \setminus \X$, of the form $\{(\omega,\{w\})\}$ and seen as maps $\omega\mapsto \{w\}$, are denoted by $\{w\}_{\{\omega\}}$ (right). Non-minimal elements are indicated by circles, respectively.}
    \label{fig:absent_minded_driver_Gilboa_sdf}
\end{figure}

\begin{lemma}\label{lemma:absent_minded_driver_Gilboa_sdf}
    The tuple $(F,\pi,\X)$ just defined is a not order consistent stochastic decision forest.
\end{lemma}

\subsection{Action path stochastic decision forests}\label{subs:APsdf}

In most pieces of the literature, dynamic games are defined by supposing a notion of time and specifying outcomes as certain paths of action at instants of time. \cite[Subsection~2.2]{AlosFerrer2005} provides a broad overview for this, including classical textbook definitions as in \cite{Fudenberg1991}, infinite bilateral bargaining in discrete time as in \cite{Rubinstein1982Perfect}, repeated games, the long cheap talk game in \cite{Aumann2003}, and a decision-theoretic interpretation of differential games as in \cite{Dockner2000}. We desire to add to this list stochastic control in both discrete and continuous time (see, e.g.\ \cite{Pham2009Continuous,Bertsekas1996Stochastic,Karatzas1998Methods}) and stochastic differential games (see, e.g.\ \cite{Carmona2018}) without restrictions on the noise in question, and the first concrete step in this direction is taken in what follows next.\footnote{As described in the introduction, this is a guiding theme of the whole series and its investigation culminates only in the third and final paper.}

In this subsection, this approach is formulated in one abstract and general framework. This framework is based on a specific structure pertaining to all of these examples, namely \emph{time}. Interestingly, time is not included in the abstract formulation of decision forests, and it serves as a particularly strong similarity structure for trees and even branches of one and the same tree. Of course, every decision path is totally ordered and therefore has its own time axis -- but in the following construction, the ``clocks'' of all decision paths are be synchronised.\footnote{We do not address the question whether all ``relevant'' order consistent and maximal stochastic decision forests can be represented in this framework -- that is, admit a ``clock'' synchronising all their decision paths -- though it does not seem obvious to conceive a counterexample.} The second major point of the following framework is that it will allow for general exogenous stochastic noise, going strictly beyond the ``nature'' agent setting.\smallskip

Let $(\T,\le)$ be a total order admitting a minimum which we denote by $0$. Further, let $I$ be a non-empty finite set (which will represent decision makers alias agents, $(\A^i)_{i\in I}$ be a family of non-empty metric spaces, and $\A = \prod_{i\in I}\A^i$ be their topological product, with canonical projections $p^i\colon\A \to \A^i$, $i\in I$. Of course, the case of singleton $I$ is included in this setting.

Let $(\Omega,\ms E)$ be an exogenous scenario space. Let $W \subseteq \Omega\times\A^\T$ be such that for all $\omega\in\Omega$, there is $f\in\A^\T$ with $(\omega,f)\in W$. An outcome will thus be a pair of an exogenous scenario and a path $f\colon \T \to \A$ of ``action'', and any scenario is required to admit at least one outcome.
For any $t\in\T$ and $\tilde w = (\omega,f)\in\Omega\times\A^\T$, let \[x_t(\tilde w) = x_t(\omega,f) = \{ (\omega',f')\in W \mid \omega' = \omega, \, f'|_{[0,t)_\T} = f|_{[0,t)_\T}\}.\] 
Let $F=\{x_t(w) \mid t\in\T, \,w\in W\} \cup \{\{w\} \mid w\in W\}$. Further, let $\pi\colon F\to \Omega$ be the unique function mapping any $x\in F$ to the first item of one choice of its elements. 

For $(t,f)\in\T\times\A^\T$, let $D_{t,f} = \{\omega \in\Omega\mid |x_t(\omega,f)| \ge 2\}$. This will turn out as the event that $x_t(.,f)$ is a move.
We consider the following assumptions:
\begin{itemize}[label=--]
    \item\hypertarget{Ass:AP.SDF0}{\textbf{Assumption~AP.SDF0.}}~For all $t\in\T$ and $f\in\A^\T$, $D_{t,f}\in\ms E$. 
    \item\hypertarget{Ass:AP.SDF1}{\textbf{Assumption~AP.SDF1.}}~For all $w\in W$ and all $t,u\in\T$ with $t\neq u$ and $x_t(w) = x_u(w)$, we have $x_t(w) = \{w\}$.
    \item\hypertarget{Ass:AP.SDF2}{\textbf{Assumption~AP.SDF2.}}~For all $\omega\in\Omega$ and $\tilde f\in\A^\T$, and for all subsets $\T'\subseteq \T$ satisfying $x_t(\omega,\tilde f) \in F$ for all $t\in \T'$, there is $f\in\A^\T$ with $(\omega,f)\in W$ and $f|_{[0,t)_\T} = \tilde f|_{[0,t)_\T}$ for all $t\in \T'$.
    \item\hypertarget{Ass:AP.SDF3}{\textbf{Assumption~AP.SDF3.}}~For all $t\in\T$ and $f,g\in\A^\T$ such that $D_{t,f},D_{t,g}\neq\emptyset$ and $D_{t,f} \cap D_{t,g} = \emptyset$, there is $u\in [0,t)_\T$ such that $D_{u,f} \cap D_{u,g} \neq\emptyset$ and $f|_{[0,u)_\T} \neq g|_{[0,u)_\T}$.
\end{itemize}
Assumption~\hyperlink{Ass:AP.SDF0}{AP.SDF0} requires that $D_{t,f}$ is indeed an event in $(\Omega,\ms E)$. Assumption~\hyperlink{Ass:AP.SDF1}{AP.SDF1} stipulates that any (presumptive) move has a unique time associated to it. Indeed, we have the following result. For any $x\in F$, let 
\[\T_x = \{t\in\T \mid \exists w\in x\colon x = x_t(w)\}. \]
It is evident that $\T_x$ is a non-empty convex set. We recall that convexity means that for all $t_0,t_1\in\T_x$, we have $[t_0,t_1]_\T \subseteq \T_x$. 

\begin{lemma}\label{lemma:AP_sdf.AssmAP.SDF1}
    Assumption~\hyperlink{Ass:AP.SDF1}{AP.SDF1} is satisfied iff for all $x\in F$ that are not singletons, $\T_x$ is a singleton.
\end{lemma}

Note that the following weak converse statement holds true independently of whether \hyperlink{Ass:AP.SDF1}{AP.SDF1} is assumed or not:

\begin{lemma}\label{lemma:AP_sdf.AssmAP.SDF1_addon}
    All $x\in F$ such that $\T_x = \{t\}$ for some $t\in\T$ that is not maximal in $\T$, are no singletons.
\end{lemma}

We immediately infer:

\begin{corollary}\label{cor:AP_sdf.AssmAP.SDF1}
    If Assumption~\hyperlink{Ass:AP.SDF1}{AP.SDF1} is satisfied and $\T$ has no maximum, then for all $x\in F$ the following statements are equivalent:
    \begin{itemize}[label=--]
        \item $x$ is a singleton.
        \item $\T_x$ is not a singleton.\qed
    \end{itemize}
\end{corollary}

Assumption~\hyperlink{Ass:AP.SDF2}{AP.SDF2} corresponds to what is generally called ``boundedness'' in \cite[Subsection~3.4]{AlosFerrer2005}. In the present context, this is the \emph{conditio sine qua non} demanding that any (presumptive) decision path, alias maximal chain of nodes, corresponds to a path of action $\T\to\A$. 
Assumption~\hyperlink{Ass:AP.SDF3}{AP.SDF3} refers to the maximality in Axiom~\ref{def:sdf.X.max} and ensures that different paths in different scenarios cannot be identified without violating the order consistency. If the assumption does not hold true, then it will become possible to identify the path $f$ on $D_{t,f}$ with the path $g$ on $D_{t,g}$, as for all $u<t$ until that (exclusively) $f$ and $g$ can be distinguished, $D_{u,f}$ and $D_{u,g}$ continue to be disjoint. 

Further, let $\X$ be the set of maps
\[ \x_t(f) \colon D_{t,f} \to F,\, \omega \mapsto \x_t(f)(\omega) = x_t(\omega,f), \]
ranging over all $t\in\T$, $f\in \A^\T$ with $D_{t,f} \neq \emptyset$.

\begin{definition}
    The tuple $(I,\A,\T,W)$ is called \emph{action path stochastic decision forest (\textsc{sdf}) data} on $(\Omega,\ms E)$ iff Assumptions \hyperlink{Ass:AP.SDF0}{AP.SDF$k$}, $k=0,1,2,3$, are satisfied. 
    $(F,\pi,\X)$ is said to be the \emph{action path stochastic decision forest (\textsc{sdf}) candidate} (on $(\Omega,\ms E)$) \emph{induced} by these action path \textsc{sdf} data. 
    
    If $(F,\pi,\X)$ defines an \textsc{sdf} on $(\Omega,\ms E)$, the word ``candidate'' is dropped. $\A = \prod_{i\in I} \A^i$ is called \emph{action space}, $\T$ \emph{time axis}, $I$ the \emph{action index set} and $\A^i$ the \emph{$i$-th action space factor} of the action path \textsc{sdf} data, and \textit{a fortiori}, of the induced action path \textsc{sdf}.\footnote{Note that $\A$, $\T$, $I$ are part of the data only and, in contrast to $W$, not uniquely determined by the \textsc{sdf}, although ``minimal'' representatives might be inferred.}
\end{definition}

\begin{example}
    For $\T = \R_+$, singleton $\Omega$, $W = \Omega \times \A^\T$, the induced action path \textsc{sdf} candidate $(F,\pi,\X)$, more precisely, $F$ yields the decision tree of what is called ``differential game'' in \cite[Subsection~2.2.5]{AlosFerrer2005}. Hence, the action path \textsc{sdf} is a generalisation of this example in several directions: first and foremost, by adding an exogenous scenario space, second, by generalising the time axis, and third, by allowing for much higher flexibility regarding the set of outcomes.
\end{example}

\begin{example}
    For $\T = \{0,\dots,N\}$, for $N\in\Nast$, or $\T = \N$ any corresponding action path \textsc{sdf} candidate can serve as a basis for describing the classical finite or infinite horizon discrete-time stochastic decision problem (see, e.g.\ \cite{Bertsekas1996Stochastic}). The admissible actions can be specified through the choice of $W$. The case $W=\Omega\times\A^\T$ corresponds to the case where at each move, the set of possible joint action is given by $\A$. It is discussed in Example~\ref{ex:APsdf} that these action path \textsc{sdf} data induce an action path \textsc{sdf}, even for general $\T$.

    Note that in contrast to the traditional ``nature'' model of stochastic games and decision problems in discrete time, exogenous information cannot be deduced from the order-theoretic properties of the (random) moves. The ``nature'' player or agent is replaced with a forest of decision trees and a structure of random moves to that one may attach exogenous information in the form of $\sigma$-algebras as explained in the next section. 

    $\T$ can also equal more general well-orders: For $\T = \N \cup \{+\infty\}$ with canonical order, that is, essentially, the second smallest infinite ordinal, the long cheap talk game tree (see \cite{Aumann2003}) -- and general stochastic variants thereof in the sense of \textsc{sdf}s -- can be obtained. This example, in case of singleton $\Omega$, is treated in \cite{AlosFerrer2005}.
\end{example}

An important result of this paper, expressed in the next theorem, is that the preceding construction -- encompassing a large set of decision problems and games in their elementary decision-theoretic structure, including very general stochastic versions of them -- is well-defined and yields an \textsc{sdf}. This provides the basis for formulating a large class of stochastic decision problems in extensive form -- already in this paper it is shown how exogenous information and choices adapted to it, and in the following papers how decision problems can be modelled on that basis.

\begin{thm}\label{thm:AP_sdf}
    Let $(I,\A,\T,W)$ be action path \textsc{sdf} data on an exogenous scenario space $(\Omega,\ms E)$. Then, the induced action path \textsc{sdf} candidate is an order consistent and maximal stochastic decision forest on $(\Omega,\ms E)$. It surely non-trivial iff for all $\omega\in\Omega$, there are $f,f'\in\A^\T$ with $f\neq f'$ and $(\omega,f),(\omega,f')\in W$.
\end{thm}

As will be discussed in the follow-up article, for well-ordered $\T$, well-posed decision problems can be defined on this action path \textsc{sdf}, while for the general case, this \textsc{sdf} can turn out to be ill-posed. For this latter case, the action path \textsc{sdf} represents an important first step though, in that a modification of it, the action-reaction \textsc{sdf}, a generalisation of a concept in \cite{AlosFerrer2015}, will restore well-posedness, as will be discussed in the third paper of the present series.\smallskip

Given an action path \textsc{sdf} as above, let $\mf t\colon X\to \T$ be the map assigning to any move $x\in X$ the unique element $\mf t(x)$ of $\T_x$, see Lemma \ref{lemma:AP_sdf.AssmAP.SDF1}.

\begin{lemma}\label{lemma:mf_t}
    Let $(F,\pi,\X)$ be the action path \textsc{sdf} induced by action path \textsc{sdf} data $(I,\A,\T,W)$ on an exogenous scenario space $(\Omega,\ms E)$. Then $\mf t$ is strictly decreasing, that is, for all $x,y\in X$ with $x\supsetneq y$ we have $\mf t(x) < \mf t(y)$. Moreover, for all $\x\in\X$ and all $\omega,\omega'\in D_{\x}$, $\mf t(\x(\omega)) = \mf t(\x(\omega'))$.
\end{lemma}

By this lemma, $\mf t$ induces a map $\X\to\T$ which we denote also by $\mf t$.\smallskip

We conclude this section with examples of action path \textsc{sdf}. This includes the illustrative examples from Subsection~\ref{subs:simple_sdf} and several typical classes of well-known decision problems. As announced earlier, it remains an open question at this point whether any ``relevant'' order consistent and maximal stochastic decision forest can be represented as an action path \textsc{sdf}, that is, so to speak, whether for any such \textsc{sdf} all branches of all trees can be synchronised.

\begin{lemma}\label{lemma:simple_sdf_as_APsdf}
    The two simple stochastic decision forests from Subsection~\ref{subs:simple_sdf} can be represented as an action path \textsc{sdf} with time $\T = \{0,1\}$.
\end{lemma}

The proof of this lemma in the appendix shows that the representation of the variant $(F',\pi',\X')$ as an action path \textsc{sdf} has to specify a ``dummy'' action at the terminal node $\{w'_{12}\}$, without letting an alternative, that is, an action meaning inaction. This is an artefact of the modelling decision to explicitly include a temporal dimension, as reflected in the action path formulation.\footnote{Assumption~\hyperlink{Ass:AP.SDF1}{AP.SDF1} ensures that moves have a unique time associated to them, but terminal nodes need not. In many cases they do not, as made apparent by the Lemmata \ref{lemma:AP_sdf.AssmAP.SDF1} and \ref{lemma:AP_sdf.AssmAP.SDF1_addon} and, most strikingly, their Corollary~\ref{cor:AP_sdf.AssmAP.SDF1}. But there may be an instant of time that certain branches still ``move'' at, and others do not and actually identify that instant with later points in time, which appears a bit artificial.} 
Stochastic decision forests are based on first principles and do no include such a dimension. Therefore, apart from being more general and flexible, their decision-theoretic interpretation is much clearer. At the same time, they include action path \textsc{sdf}s which model a structure present in many applications and which can therefore simplify the formal representation, at the cost of possibly introducing artificial phenomena like inactive activity. It may be noted that, in action path \textsc{sdf}, time plays a role complementary to that of random moves: it defines similarity of moves across branches, while random moves define similarity across trees. The random moves in action path \textsc{sdf}s are defined such as to be perfectly compatible with time.

\begin{example}\label{ex:APsdf}
    \begin{itemize}[label=--]
        \item Without further assumptions on the data $(I,\A,\T,W)$, it holds true that if $W= \Omega \times \A^\T$, then they are action path \textsc{sdf} data.
        
        \item Suppose that $\A^i = \{0,1\}$ for all $i\in I$. Let $W$ be the set of pairs $(\omega,f)$ where $f\colon\T \to \A$ is componentwise decreasing. Then, the quadruple $(I,\A,\T,W)$ defines action path \textsc{sdf} data. The corresponding action path \textsc{sdf} is a natural (and at least approximate) candidate for describing stochastic decision problems of \emph{timing} (alias \emph{stopping}) (see, e.g.\ the monographs \cite{Shiryaev2007Optimal,Peskir2006,ElKaroui1981} for the mathematical theory, and \cite{Riedel2017,Karatzas1998Methods} for the link to applications in economics and finance).
        
        \item Let us consider an example of a forest that becomes shallower towards its crown in some areas of the exogenous scenario space. We use an example from finance, namely the exercise of an American up-and-out option (\cite[Chapter~26]{Hull2018Options}). Let $P = (P_t)_{t\in\R_+}$ be a continuous stochastic process on $(\Omega,\ms E)$ with strictly positive real values describing the price of a financial security. The initial price is $P_0 = 1$. As long as the price has not reached $2$, the holder of the option can exercise it, but once the price reaches $2$, the option expires irreversibly. This problem can be modelled using the action path stochastic decision forest associated to the following data. 
        
        Let $I$ be a singleton, $\T=\R_+$, and $\A = \{0,1\}$. Take $W$ to be the set of all $(\omega,f)$ where $\omega\in\Omega$ and $f\colon\R_+ \to \{0,1\}$ is decreasing such that, if $f$ takes the value $0$, then $t^\ast_f = \inf\{t\in\R_+ \mid f(t) = 0\}$ satisfies $\max_{t\in[0,t^\ast_f]} P_{t}(\omega) < 2$. Then, $W$ induces an action path \textsc{sdf}, and 
        \[D_{t,f} = \{\omega\in\Omega \mid \max_{u\in[0,t]} P_u(\omega) < 2\},\]
        for all $t\in\R_+$ and decreasing $f\in\{0,1\}^{\R_+}$ with $f(t-) = 1$.
    \end{itemize}
    The proofs of these claims can be found in the corresponding part of the appendix.
\end{example}

\section{Exogenous information}\label{sec:exogenous_information}

As observed in the previous section, exogenous information is not contained in the order-theoretic structure of random moves. The approach put forward in this paper consists in attaching exogenous information to random moves in the form of $\sigma$-algebras, in a way analogous to (and in some sense, more general than) the use of filtrations in probability theory. Therefore, this latter approach is analysed and interpreted in the first subsection, and based on that, in the following subsections, a concept of exogenous information revelation on stochastic decision forests is introduced and explained, as well as illustrated by examples for the stochastic decision forests encountered in the previous section.

\subsection{Filtrations in probability theory}\label{subs:filtr_in_prob_th}

Recall that for a measurable space $(\Omega,\ms E)$, the usual interpretation of $\ms E$ sees its elements as being those subsets of $\Omega$ that are measurable for a given observer or agent. In the context of what we called exogenous scenario spaces, this had essentially the meaning that for all admissible beliefs (alias probability measures) on $\Omega$, probabilities can be computed for these subsets (see Subsection~\ref{subs:Ex_sc_sp}). However, there is a second meaning to ``measurable'', ubiquitous in probability theory and many of its applications.

In probability theory, a filtration is a monotone map $\ms F$ from a non-empty totally ordered set $\T$ modelling time, typically a suborder of the two-sided compactification of $\R$, to the set of sub-$\sigma$-algebras of $\ms E$, thus mapping any $t\in\T$ to a sub-$\sigma$-algebra $\ms F_t \subseteq \ms E$ such that for all $t,u\in\T$ with $t\le u$, $\ms F_t \subseteq \ms F_u$ (see, e.g.\ \cite[Chapter~9]{Kallenberg2021} or \cite[Chapter~3]{Cohen2015Stochastic}). The customary interpretation refers to an agent being equipped with that filtration $\ms F$ such that, for all times $t\in\T$, $\ms F_t$ describes the information the agent has at time $t$, or put differently, $\ms F_t$ is the set of events measurable for that agent at time $t$. Of course, this agent is able to compute probabilities for all $E\in\ms E$, but the elements of $\ms F_t$ are measurable in an even stronger sense. 

Namely, one tacitly supposes that there is one realised scenario $\omega\in\Omega$, drawn ``at random'', so to speak, but information about it is only revealed to the agent progressively via $\ms F$, in the following way: For all instants of time $t\in\T$ and each $E\in\ms F_t$ the agent can discern at time $t$ whether it contains the realised scenario or not. In particular, the agent when equipped with a belief in form of a probability measure on $\ms E$ can derive updated probabilities and expectations by evaluating conditional probabilities and expectations on $E$. Yet, this notion of measurement must be read with some caution. Actually, this measurement capacity of the agent includes logical operations given by the $\sigma$-algebra property, in particular countable logical disjunctions alias unions. Uncountable operations might be excluded however. Moreover, under the belief of the agent, $E$ may have probability zero and hence the evaluation of conditional probabilities and expectations on $E$ may be meaningless (\cite{Kallenberg2021,Cohen2015Stochastic}). This can be illustrated by the following examples.

In the discrete setting, where $\ms E$ is generated by a countable partition of $\Omega$, the interpretation is evident. Then all $\ms F_t$, $t\in\T$, are also generated by countable partitions, and the partition of later instants of time refine the preceding partitions. The members of the partition generating $\ms F_t$ can be thought of as the agent's information sets regarding the ``nature'' agent's past choices at time $t$.\footnote{This will be discussed in detail in the follow-up paper.} 
In that discrete case, thus, filtrations are formally and conceptually a special case of order-theoretic rooted forests on the set $\Omega$, as discussed in Subsection~\ref{sec:def}.

But in general, $\ms F_t$, for some $t\in\T$, may, for instance, contain all the singletons subsets of $\Omega$, without containing all elements of $\ms E$, let alone all subsets of $\Omega$ -- think of $\Omega$ being the unit interval with Borel $\sigma$-algebra $\ms E$ and $\ms F_t$ being the set of subsets $E\subseteq [0,1]$ that are countable or have countable complement. So the agent would actually be able to measure the realised scenario without discerning some $E\in\ms E$, say the interval $[0,\frac 12]$, although it is a union of measurable singletons. The point is that this union is uncountable, and we do not assume agents to be able to perform uncountable logical operations. Hence, the ability to discern the realised scenario does not necessarily imply the ability to discern all events, that is elements of $\ms E$. Moreover, all but countably many (and in the case of the Lebesgue measure on $[0,1]$ even all) singletons must have probability zero, and so the ability to discern the realised scenario may not be relevant from a decision-theoretic perspective: If the conditional expectation of some interesting quantity, e.g.\ a payoff, given $\ms F_t$ is computed, it is completely inconclusive to evaluate it on some $E\in\ms F_t$ whose probability is zero, because the conditional expectation is defined almost surely. 

Once again we note that the full descriptor of exogenous information at time $t$ is the set of ``measurable'' events $\ms F_t$, where ``measurable'' is used in the second, stronger sense, not solely the set of scenarios $\Omega$ (compare the discussion in Subsection~\ref{subs:Ex_sc_sp}). In particular, this discussion suggests that, in general, filtrations cannot be subsumed under the theory of decision forests and trees, which is ``discrete'' in that it can be represented via refined partitions, and thus filtrations provide a means to truly extend the scope of stochastic game and decision theory beyond those types of stochastic noise a dynamically choosing ``nature'' agent can simulate.

\subsection{Exogenous information structures}

Based on the preceding argument about exogenous scenario spaces and filtrations, and bearing in mind the way stochastic decision forests are built on exogenous scenario spaces, the following definition is proposed, which is interpreted and analysed in the sequel.

\begin{definition}\label{def:EIS}
    Let $(F,\pi,\X)$ be a stochastic decision forest on an exogenous scenario space $(\Omega,\ms E)$ and let $\tilde\X\subseteq\X$. An \emph{exogenous information structure on $\tilde\X$} is a family $\ms F = (\ms F_\x)_{\x\in\tilde\X}$ such that for all $\x\in\tilde\X$, $\ms F_\x$ is a sigma-algebra on $D_\x$ with $\ms F_\x\subseteq\ms E$. An exogenous information structure $\ms F$ on $\tilde\X$ is said to admit \emph{recall} iff for all $\x,\x'\in\tilde\X$ with $\x \ge_\X \x'$ and every $E\in \ms F_\x$, we have $E\cap D_{\x'} \in \ms F_{\x'}$.
\end{definition}

Concerning the sense of this definition, we think of any decision maker alias agent as being equipped with some exogenous information structure on the set of his random moves $\tilde\X$. For any of that agent's random moves $\x\in\tilde\X$, $\ms F_\x$ is interpreted as the set of events $E\in\ms E$ relevant at $\x$, meaning that $E\subseteq D_\x$, and that the agent can measure at $\x$, meaning that the agent can discern whether $E$ contains the realised scenario or not. This second property is to be read in the sense discussed in the previous Subsection~\ref{subs:filtr_in_prob_th}. 
The more straightforward condition of recall\footnote{Note that recall of an exogenous information structure is a weak adaptation of the notion of ``weak recall'' in \cite{Ritzberger1999Recall} along the exogenous dimension.} ensures that the agent does not forget exogenous information: if the agent is at $\x'$ and has been at $\x$ already, and was able to measure $E$ at $\x$, then this agent can measure $E$ given the domain of $\x'$ also at $\x'$.
Note that, if for all $\x\in\X$, we have $D_\x = \Omega$, then an exogenous information structure on $\X$ admitting recall is nothing but a monotone decreasing map $\ms F\colon\x\mapsto \ms F_\x$ from $\X$ into the set of sub-sigma-algebras of $\ms E$, ``monotone decreasing'' meaning that for all $\x,\x'\in\X$ with $\x \ge_\X \x'$, we have $\ms F_\x \subseteq \ms F_{\x'}$. 

It should be noted that what we have at hand here is a generalised adaptation of the notion of filtrations from probability theory to decision trees. In both cases, the index set is partially ordered. However, in the former case the order is total, whereas in the latter quite the opposite case is true, since it is a decision tree. Moreover, the underlying forest can ``thin out'' towards the crowns non-uniformly across trees alias exogenous scenarios, so that, in general, domains of random moves must be taken into account.

\subsection{Simple examples}\label{subs:simple_sdf_EIS}
In the following we discuss all exogenous information structures admitting recall for the simple examples from the previous section, Subsection~\ref{subs:simple_sdf}.\smallskip

First, consider the simple stochastic decision forest $(F,\pi,\X)$ from Subsection~\ref{subs:simple_sdf}, illustrated in Figure~\ref{fig:simple_sdf}. Consider the five following families $\ms F = (\ms F_\x)_{\x\in\X}$:
\begin{enumerate}
    \item $\ms F_\x = \{\Omega,\emptyset\}$ for all $\x\in\X$: at all moves, it is unknown which scenario is realised;
    \item $\ms F_{\x_0} = \{\Omega,\emptyset\}$ and one of the following three cases is true:
    \begin{enumerate}
        \item $\ms F_{\x_1} = \ms F_{\x_2} = \mc P \Omega$: only at the second move, it becomes known which scenario is realised, irrespective of which one is the second move;
        \item $\ms F_{\x_1} = \mc P \Omega$, $\ms F_{\x_2} = \{\Omega,\emptyset\}$: $\x_1$ is the only move at that the realised scenario is revealed; an agent with this exogenous information may have interest in choosing (if possible) $\x_1$ rather than $\x_2$ in order to learn, modelling the trade-off \emph{exploration vs.\ exploitation}; that way, problems with partial information and adaptive control can be modelled;
        \item $\ms F_{\x_2} = \mc P \Omega$, $\ms F_{\x_1} = \{\Omega,\emptyset\}$: analogous to the preceding situation;
    \end{enumerate}
    \item $\ms F_\x = \mc P \Omega$ for all $\x\in\X$: at all moves, the realised scenario is known.
\end{enumerate}

\begin{lemma}\label{lemma:simple_sdf1_EIS}
    For the simple stochastic decision forest $(F,\pi,\X)$ from Subsection~\ref{subs:simple_sdf} there are exactly five exogenous information structures on $\X$ admitting recall, and these are given by the families $\ms F$ considered just above.
\end{lemma}

Now, consider the variant $(F',\pi',\X')$ from Subsection~\ref{subs:simple_sdf}, as illustrated in Figure~\ref{fig:simple_sdf_variant}. Consider the three following families $(\ms F'_{\x'})_{\x'\in\X'}$.
In all three cases, let $\ms F'_{\x'_2} = \{\{\omega_2\},\emptyset\}$. Moreover, separate the following three cases.
\begin{enumerate}
    \item $\ms F'_{\x'_0} = \ms F'_{\x'_1} = \{\Omega,\emptyset\}$: again, there could be an exploitation vs.\ exploration trade-off (compare the cases 2(c) above);
    \item $\ms F'_{\x'_0} = \{\Omega,\emptyset\}$, $\ms F'_{\x'_1} =\mc P(\Omega)$: this is similar to case 2(a) above;
    \item $\ms F'_{\x'_0} = \ms F'_{\x'_1} = \mc P(\Omega)$: the realised scenario is known at any move (similar to case 3 above).
\end{enumerate}

\begin{lemma}\label{lemma:simple_sdf2_EIS}
    For the stochastic decision forest $(F',\pi',\X')$ from Subsection~\ref{subs:simple_sdf} there are exactly three exogenous information structures on $\X'$ admitting recall, and these are given by the families $\ms F'$ considered just above.
\end{lemma}

We conclude this subsection with a note on the absent-minded driver story as modelled by Gilboa in \cite{Gilboa1997Comment}. Let $I = \{1,2\}$ and $\xi^1,\xi^2$ real-valued random variables on $(\Omega,\ms E)$ such that there is a probability measure $\P$ on it making $(\rho,\xi^1,\xi^2)$ independent. On the basis of the \textsc{sdf} $(F,\pi,\X)$ proposed in Subsection~\ref{subs:simple_sdf}, let $\ms F^i_{\x_i} = \sigma(\xi^i)$ for both $i\in I$, that is, the agent active at $\x_i$ has no non-trivial information about $\rho$ if both agents have posterior beliefs derived from $\P$. These give trivially rise to exogenous information structures on $\X$ admitting recall.

\subsection{Action path stochastic decision forests}\label{subs:APsdf_EIS}

In this subsection, we discuss examples of exogenous information structures for action path \textsc{sdf} data $(I,\A,\T,W)$ and the induced action path \textsc{sdf} $(F,\pi,\X)$ as in Subsection~\ref{subs:APsdf}, defined on an exogenous scenario space $(\Omega,\ms E)$. First, we observe a necessary condition on these in the case of recall: When evaluated along maximal chains of random moves in $(\Tr,\ge_\Tr)$, we obtain a filtration in the sense of classical probability theory. More precisely:

\begin{lemma}\label{lemma:APsdf_EIS_induces_filtration}
    Let $(F,\pi,\X)$ be the action path \textsc{sdf} induced by action path \textsc{sdf} data $(I,\A,\T,W)$ on an exogenous scenario space $(\Omega,\ms E)$. 
    Let $\ms F$ be an exogenous information structure on $\X$ admitting recall, and let $f\in\A^\T$. Let $\T_f = \{t\in\T \mid D_{t,f} \neq \emptyset\}$. 
    
    Then, for all $t,u\in\T_f$ with $t<u$ and $E\in\ms F_{\x_t(f)}$, $E\cap D_{u,f} \in \ms F_{\x_u(f)}$. In particular, if $D_{t,f} = \Omega$ for all $t\in\T_f$, then $(\ms F_{\x_t(f)})_{t\in\T_f}$ defines a filtration with time index set $\T_f$.
\end{lemma}

In the already discussed special case that $\A$ has at least two elements and $W=\Omega \times \A^\T$, $(\ms F_{\x_t(f)})_{t\in\T}$ is a filtration.\smallskip

Now we turn to sufficient conditions, or more precisely, to the construction of exogenous information structures for the action path \textsc{sdf}. Let $\ms G = (\ms G_t)_{t\in\T}$ be a filtration on $(\Omega,\ms E)$. For example, $\ms G$ could be generated by some stochastic process, possibly augmented with nullsets with respect to some set of probability measures on $(\Omega,\ms E)$.

\begin{enumerate}
    \item The basic setting in stochastic control of exogenous noise revealed over time independently of the agents' behaviour (see, e.g.\ \cite{Pham2009Continuous}) corresponds to letting, for all $\x\in\X$, \[\ms F_{\x} = \ms G_{\mf t(\x)} \mid_{D_{\x}}, \] and $\ms F = (\ms F_\x)_{\x\in\X}$. This exogenous information structure admits recall.
    \item We can also consider a more general case that allows for exogenous information depending on previous behaviour -- which can serve as a basis for describing problems with partial information involving stochastic filtering, the trade-off exploration vs.\ exploitation, and adaptive control (see, e.g.\ \cite{Bain2009Fundamentals,Cohen2015Stochastic,Cohen2023Optimal}). Again, let $\ms G = (\ms G_t)_{t\in\T}$ be a filtration on $(\Omega,\ms E)$. Let $(Y_\x)_{\x\in\X}$ be a family of random variables with values in some metric space $\B$. Let, for any $\x\in\X$:
    \[ \ms F_\x =\Big( \sigma(Y_{\x'} \mid \x' \ge_\X \x) \vee \ms G_{\mf t(\x)} \Big)  \bigmid_{D_{\x}}. \]
    Again, let $\ms F = (\ms F_\x)_{\x\in\X}$, yielding an exogenous information structure admitting recall.
    \item Taking $\ms F_\x = \sigma(Y_\x)$, for all $\x\in\X$, can and does often provide an exogenous information structure not admitting recall.
\end{enumerate}

In particular, these cases can be used to model ``open loop'' and ``closed loop'' controls, respectively, covering the definitions in \cite[Chapter~2]{Carmona2018}, for example. However, as will be discussed in the Subsection~\ref{sec:adapted_choices} about adapted choices, the crux lies in that local decisions not only depend on the filtration-like object representing exogenous information at the current random move, but also on the current random move itself. Hence, when comparing a counterfactual random move to a reference random move revealing the same exogenous information, the same strategy may demand a different choice, because the random move is different.

\begin{example}\label{ex:APsdf_EIS}
    We give a typical example of such a family $(Y_\x)_{\x\in\X}$ in the case $\B = \R$, $\A = \R$, and $W = \Omega \times C(\R)$. Consider an auxiliary family of real-valued random variables $(\tilde Y_\x)_{\x\in\X}$ satisfying the stochastic differential equation
    \[ \tilde Y_{\x_t(f)} = \int_0^t \tilde b(f(u),\tilde Y_{\x_u(f)})~\d \tilde Z_{u}, \]
    for all $t\in\T$ and $f\in C(\R)$, and bounded continuous $\tilde b\colon \R\times \R \to \R$ and a suitable $\ms G$-adapted stochastic integrator $\tilde Z$, say Brownian motion\footnote{... with respect to a probability measure on $(\Omega,\ms E)$}. Then let $(Y_\x)_{\x\in\X}$ solve the stochastic differential equation
    \[ Y_{\x_t(f)} = \int_0^t b(\tilde Y_{\x_u(f)},Y_{\x_u(f)})~\d u + \int_0^t \sigma(Y_{\x_u(f)})~\d Z_{u} \]
    for all $t\in\T$ and $f\in C(\R)$, and bounded continuous $b\colon \R\times \R \to \R$ and $\sigma \colon \R \to \R$, and a suitable $\ms G$-adapted stochastic integrator $Z$, say Brownian motion. 

    $\tilde Y$ can be interpreted as a noisy and not directly observable signal controlled by an agent through $f$, while $Y$ can be seen as an observation depending on $\tilde Y$. This is a typical setting in control theory involving stochastic filtering (see, for instance, \cite{Bain2009Fundamentals,Cohen2023Optimal}, \cite[Part~V]{Cohen2015Stochastic}).
\end{example}

\begin{thm}\label{thm:AP_sdf_EIS}
    Let $(F,\pi,\X)$ be the action path \textsc{sdf} induced by action path \textsc{sdf} data $(I,\A,\T,W)$ on an exogenous scenario space $(\Omega,\ms E)$. Consider the family $\ms F$ from above, in its general version as in point (2). Then, $\ms F$ defines an exogenous information structure on $\X$ admitting recall.
\end{thm}

\section{Adapted choices}\label{sec:adapted_choices}

In this section, a concept of choices is introduced that aims at reconciling the classical decision-theoretic model of choice under uncertainty and the probabilistic setting of adapted processes. First, these two concepts, their differences and intersections are discussed. Then, the concept of adapted choices on stochastic decision forests is introduced and explained. The section will be completed with examples for the stochastic decision forests and exogenous information structures introduced in the two preceding sections.

\subsection{Choice under uncertainty and adapted processes}

In the refined partitions approach classical in sequential decision theory and developed in high generality in \cite{AlosFerrer2005}, a choice consists in selecting a member of a given partition of the set of outcomes that are still possible according to the information the given agent has at the current move. The partition members can be seen as ``local consequences'', pertaining to the current information. Upon relabelling moves as ``local states'', this is a direct application of Savage's concept of acts (\cite{Savage1972Foundations}) which assign a (local) consequence to each (local) state. If the overall problem is well-posed this happens in such a way that the set of possible outcomes (which can be seen as global consequences) is progressively reduced to a singleton.\footnote{This point is discussed in the follow-up paper.}

In probabilistic models based on filtrations which describe exogenous information on an abstract measurable space $(\Omega,\ms E)$, choices are modelled differently. If $\ms F = (\ms F_t)_{t\in\T}$ is a filtration on $(\Omega,\ms E)$ with some totally ordered time index set $\T$, then the choice of an agent $i$ equipped with that information and capable of action described by a metric space $\A^i$ at a time $t$ is typically modelled by an $\ms F_t$-measurable function $g_t\colon\Omega\to \A^i$, modulo regularity conditions in $t$ (leading to predictable processes, for instance), see, for instance, \cite{Pham2009Continuous,Cohen2015Stochastic}. $\ms F_t$-measurability of $g_t$ means that for any measurable set of actions $A_t^i \subseteq \A^i$, the event that $i$ chooses an action in $A_t^i$ is measurable to agent $i$ at time $t$, that is, $(g_t)^{-1}(A_t^i)\in \ms F_t$. Actually, $A_t^i$ can be understood as a set-valued reference choice and $(g_t)^{-1}(A_t^i)$ is just the event that both choices, $g_t$ and $A^i_t$, are compatible.

In the discrete case, this approach can be subsumed under the refined partitions framework going back to \cite{Neumann1944} and presented in high generality in \cite{AlosFerrer2016}. Indeed, if $\ms E$ is generated by a countable partition, then so is $\ms F_t$ for each $t\in\T$, and $\ms F_t$-measurability is another way of expressing the requirement that $g_t$ is constant on each member of the partition generating $\ms F_t$. Recall from Subsection~\ref{subs:filtr_in_prob_th} that each realised partition member is thought of as the agent's information set regarding the ``nature'' agent's past choices. In that sense, $g_t$ corresponds to a map from moves at time $t$ to action such that at moves from the same information set the same action is selected. A family of such maps $g_t$, ranging over all times $t$, also called adapted process in the language of stochastic processes, then defines a complete contingent plan of action for this agent $i$ provided $i$ has only exogenous information.

On general measurable spaces, however, the $\ms F_t$-measurability cannot be rephrased like this in terms of partitions. Moreover, agents may have endogenous information about the own or other agents' past behaviour and condition their own behaviour on this information, that is, they may condition on what they know about their current random move. This suggests that one may adapt the more general measure-theoretic concept of measurable functions to stochastic decision forests as regards exogenous information about the horizontal $\Omega$ axis, while maintaining the refined partitions logic along the vertical tree axis in order to rigorously explain the interactive decision making in extensive form. The aim of this section is to introduce such a concept, bringing together the measure-theoretic probabilistic and the refined partitions based tree-like approaches to information.\footnote{In this section we do however not make formally precise which axioms the set of choices an agent is equipped with should satisfy in order to consistently define a stochastic extensive form following the refined partitions paradigm and how exactly information sets can be modelled, nor do we discuss how and when this general stochastic approach could be represented using a ``nature'' agent: this is the aim of the second part of the present series, see \cite{Rapsch2024DecisionB}.}

\subsection{Choices in stochastic decision forests}

A basic principle in extensive form theory is that at any ``move'' it is ``known'' to decision makers whether a given ``choice'' is available to them or not. One of the important facts the refined partitions approach formalised in \cite{AlosFerrer2005} elucidates, is that the availability of a choice at a given move can be completely described in terms of the underlying set-theoretic structure: A choice is available at a move iff the latter is an immediate predecessor of the former. More precisely, if $(F,\pi,\X)$ is a stochastic decision forest on an exogenous scenario space $(\Omega,\ms E)$, $W=\bigcup F$ and $c\subseteq W$ is some subset (for instance, a union of nodes representing a choice), then, with
\[ \downarrow c = \{ x\in F \mid c \supseteq x \}, \]
let, as in the classical setting of \cite{AlosFerrer2005}, the set of \emph{immediate predecessors of $c$} be defined by:
\[ P(c) = \{x\in F \mid \exists y \in \downarrow c\colon \uparrow x = \uparrow y \setminus \downarrow c \}. \]

\begin{lemma}\label{lemma:P(c)_compatible_with_conn_comp}
    For each subset $E\subseteq\Omega$, and each subset $c\subseteq W$, 
    \[ P(c \cap W_E) = P(c) \cap F_E. \]
\end{lemma}

As in \cite{AlosFerrer2005}, a \emph{choice} is a non-empty union $c$ of nodes. For $x\in X$, $c$ is said to be \emph{available at $x$} iff $x\in P(c)$, as in \cite{AlosFerrer2005}. The sets $P(c)$ are a model for information sets in \cite{AlosFerrer2005} and \cite{AlosFerrer2008,AlosFerrer2011Comment} and will turn out to have, in a weaker sense, a similar role in the setting of stochastic decision forests, as discussed in the second paper, see \cite[Subsection~1.3]{Rapsch2024DecisionB}.

While in a discrete setting the description of availability of choices in terms of predecessors makes it relatively easy to interpret and thus adhere to the above-mentioned basic principle, this becomes less evident in the present general measure-theoretic context, because of the tension between the discreteness of choices and the general measure-theoretic description of exogenous information in terms of systems of ``measurable'' sets rather than using partitions.
We therefore explicitly assume a locally defined structure of reference choices like the $A_t^i$ above, that can be measured by agents at this random move. Choices made by agents must be adapted with respect to their individual exogenous information structure and that structure of reference choices. Hence, on the one hand, choices can be ``discrete'' in the sense of constituting partitions and thus eliminating a sufficient amount of alternatives so that progressive choosing results in a unique outcome. And on the other hand, choices can be adapted in the sense that the availability of its restriction to any reference choice is a measurable event with respect to the current exogenous information.

\begin{definition}
    Let $(F,\pi,\X)$ be a stochastic decision forest on an exogenous scenario space $(\Omega,\ms E)$, let $\tilde \X\subseteq\X$, and let $\ms F$ be an exogenous information structure on $\tilde \X$.
    \begin{enumerate}
        \item A choice $c$ is said
        \begin{enumerate}
            \item \emph{non-redundant} iff for any $\omega\in\Omega$ with $P(c) \cap T_\omega = \emptyset$, we have $c\cap W_\omega = \emptyset$;
            \item \emph{$\tilde \X$-complete} iff for every random move $\x\in \tilde\X$, $\x^{-1}(P(c))$ is either empty or equal to $D_\x$;
            \item \emph{complete} iff it is $\X$-complete.
        \end{enumerate}
        \item For any random move $\x\in\X$, a choice $c$ is said \emph{available at $\x$} iff $\x^{-1}(P(c)) = D_\x$.
        \item A \emph{reference choice structure on $\tilde\X$} is a family $\ms C = (\ms C_\x)_{\x\in\tilde\X}$ of sets $\ms C_\x$ of non-redundant and $\tilde\X$-complete choices available at $\x$, $\x\in\tilde\X$.
        \item Let $\ms C$ be a reference choice structure on $\tilde\X$. An \emph{$\ms F$-$\ms C$-adapted choice} is a non-redundant and $\tilde\X$-complete choice $c$ such that for all $\x\in\tilde\X$ that $c$ is available at and all $c'\in\ms C_\x$:
        \[ \x^{-1}(P(c \cap c')) = \{\omega\in D_\x \mid \x(\omega) \in P(c\cap c') \} \in \ms F_\x. \]
    \end{enumerate}
\end{definition}

From the an agent's perspective, both for exogenous information revelation and adapted choices the relevant order (or even tree-like) structure is the partial order (or even decision tree) $(\Tr,\ge_\Tr)$, and more precisely some subset $\tilde\X$ describing the moves of that agent. In that respect, choices are made on that tree with respect to the exogenous information revealed along it. Regarding outcomes and outcome generation, however, the $\Omega$ dimension and thus the forest $F$ are crucial. $\X$ builds the link between both, and exogenous information as well as adapted choices are defined with respect to $\X$ and such as to be compatible with each other. It is more general than usual continuous-time stochastic control and differential games formulations (as in \cite{Pham2009Continuous,Cohen2015Stochastic,Carmona2018}) because of the dependence on $\X$, rather than on the more rigid notion of time. This point, among others, will be clarified in the following examples.


\subsection{Simple examples}\label{subs:simple_sdf_AC}

For the examples from Subsection~\ref{subs:simple_sdf}, both the basic version and its variation, we provide a reference choice structure, and a list of adapted choices, one for each exogenous information structure from Example~\ref{subs:simple_sdf_EIS}.\smallskip

First, consider the basic version $(F,\pi,\X)$. Let $M$ be the set of maps $\Omega\to \{1,2\}$. For $k\in \{1,2\}$ and $f,g\in M$, let
\begin{align*}
    c_{f\bullet} =&~ \{(\omega,k',m')\in W \mid k' = f(\omega)\}, \\
    c_{k g} =&~  \{(\omega,k',m')\in W \mid k' = k, ~ m' = g(\omega)\}, \\
    c_{\bullet g} =&~\{(\omega,k',m')\in W \mid  m' = g(\omega)\}.
\end{align*}
Define $c_{k\bullet}$, $c_{\bullet m}$, and $c_{km}$ by identifying $k,m\in\{1,2\}$ with the constant maps on $\Omega$ with values $k$ and $m$, respectively.  
Let $\ms C_{\x_0} = \{c_{1 \bullet},c_{2 \bullet}\}$, $\ms C_{\x_1} = \ms C_{\x_2} = \{c_{\bullet1},c_{\bullet2}\}$. Note the partitioned structure of these sets, reflecting the discreteness of the situation. 

\begin{lemma}\label{lemma:simple_sdf1_LCS}
    $\ms C = (\ms C_\x)_{\x\in\X}$ defines a reference choice structure on $\X$.
\end{lemma} 

Next, consider the following table. It reads as follows: Each line specifies a set of subsets of $W$ for each of the five exogenous information structures (\textsc{eis}) listed in Subsection~\ref{subs:simple_sdf_EIS}, first part; these subsets are classified according to whether they will correspond to choices at the beginning of the ``first period'' (at time $0$) or of the ''second period'' (at time $1$), if perceived as action path \textsc{sdf} according to Lemma \ref{lemma:simple_sdf_as_APsdf}:
\begin{center}
    \begin{tabular}{r| c c}
     \textsc{eis}& 1st period & 2nd period  \\
     \hline
      1. & $c_{k \bullet}$ : $k\in\{1,2\}$ &$c_{k m}, c_{\bullet m}$ : $k,m\in\{1,2\}$ \\
      2.(a) &$c_{k \bullet}$ : $k\in\{1,2\}$ & $c_{k g}, c_{\bullet g}$ : $k\in\{1,2\},\, g\in M$ \\
      2.(b) & $c_{k \bullet}$ : $k\in\{1,2\}$ & $c_{1g}, c_{2m}, c_{\bullet m}$ : $m\in\{1,2\},\, g\in M$ \\
      2.(c) & $c_{k \bullet}$ : $k\in\{1,2\}$ & $c_{1m}, c_{2g}, c_{\bullet m}$ : $m\in\{1,2\},\, g\in M$ \\
      3. &  $c_{f \bullet}$ : $f\in M$ & $c_{k g}, c_{\bullet g}$ : $k\in\{1,2\},\, g\in M$ \\
    \end{tabular}
\end{center}

\begin{lemma}\label{lemma:simple_sdf1_AC}
    For each exogenous information structure $\ms F$, the subsets of $W$ given in the corresponding line of the preceding table are $\ms F$-$\ms C$-adapted choices on $(F,\pi,\X)$.
\end{lemma}

Second, consider the variant $(F',\pi',\X')$ of the simple stochastic decision forest. Let, again, $M$ be the set of maps $\Omega\to \{1,2\}$. For $k\in \{1,2\}$ and $f,g\in M$, let
\begin{align*}
    c'_{f\bullet} =&~ \{(\omega,k',m'), (\omega,k')\in W \mid k' = f(\omega)\}, \\
    c'_{k g} =&~ \{(\omega,k',m')\in W \mid k' = k, \, m' = g(\omega)\}, \\
    c'_{\bullet g} =&~ \{(\omega,k',m') \in W \mid  m' = g(\omega)\}.
\end{align*}
Define $c'_{k\bullet}$, $c'_{\bullet m}$, and $c'_{km}$ by identifying $k,m\in\{1,2\}$ with the constant maps on $\Omega$ with values $k$ and $m$, respectively.  
Let $\ms C'_{\x'_0} = \{ c'_{1 \bullet},c'_{2 \bullet}\}$, $\ms C'_{\x'_1} = \ms C'_{\x'_2} = \{c'_{\bullet1},c'_{\bullet2}\}$. 

\begin{lemma}\label{lemma:simple_sdf2_LCS}
    $\ms C' = (\ms C'_{\x'})_{\x'\in\X'}$ defines a reference choice structure on $\X'$.
\end{lemma}

Next, consider the following table. It reads as above.
\begin{center}
    \begin{tabular}{r| c c}
     \textsc{eis}& 1st period & 2nd period  \\
     \hline
      1. & $c'_{k \bullet}$ : $k\in\{1,2\}$ & $c'_{k m}, c'_{\bullet m}$ : $k,m\in\{1,2\}$ \\
      2. & $c'_{k \bullet}$ : $k\in\{1,2\}$ & $c'_{k g}, c'_{\bullet g}$ : $k\in\{1,2\},\, g\in M$ \\
      3. & $c'_{f \bullet}$ : $f\in M$ & $c'_{k g}, c'_{\bullet g}$ : $k\in\{1,2\},\, g\in M$ \\
    \end{tabular}
\end{center}   

\begin{lemma}\label{lemma:simple_sdf2_AC}
    For each exogenous information structure $\ms F'$, the subsets of $W'$ given in the corresponding line of the preceding table are $\ms F'$-$\ms C'$-adapted choices on $(F',\pi',\X')$.
\end{lemma}

Note that the adaptedness of choices can be rephrased, namely by requiring that $f$ and $g$ be measurable with respect to the $\sigma$-algebra of exogenous information at the random move the choice is available at, respectively. This more convenient language is used in the context of action paths in the next subsection which contains the previous two examples following up on Lemma \ref{lemma:simple_sdf_as_APsdf}.

Also note that a choice like $c_{kg}$ is conditional on the knowledge that $\x_k$ has been chosen at the root of $\X$, while $c_{\bullet g}$ is independent of the initial choice. This is reflected by the respective sets of immediate predecessors, see Lemma \ref{lemma:simple_sdf1_choices}. A similar remark is true for the variant, see Lemma \ref{lemma:simple_sdf2_choices}. Hence, as in \cite[Section~5]{AlosFerrer2005}, choices can reflect the endogenous information, that is, the information about the position in the decision tree $(\Tr,\ge_\Tr)$, agents have. While $c_{kg}$ is a choice of perfect (endogenous) information, $c_{\bullet g}$ is not. The discussion on this theme will be continued in the second paper, see \cite{Rapsch2024DecisionB}.\smallskip

Concerning the \textsc{sdf} model $(F,\pi,\X)$ of the absent-minded driver story following Gilboa, with the exogenous information structure $\ms F$ as discussed in Subsection~\ref{subs:simple_sdf_EIS}, let, for $i\in I = \{1,2\}$:
\[ \op{Ex}_i = \underbrace{[\rho^{-1}(i)\times\{D\}] \cup [\rho^{-1}(3-i)\times\{H\}]\}}_{=\text{``exit''}}, \qquad \op{Ct}_i = \underbrace{[\rho^{-1}(i)\times\{H,M\}] \cup [\rho^{-1}(3-i) \times \{M\}]}_{\text{=``continue''}},\]
and $\ms C_{\x_i}^i = \{ \op{Ex}_i,\op{Ct}_i\}$.
Futhermore, for any $E\in\ms F^i_{\x_i}$, let
\[ c_i(E) = (W_E \cap \op{Ex}_i)\cup(W_{E^\complement} \cap \op{Ct}_i), \]
the choice of ``agent'' to exit in the event $E$ and to continue in the opposite event $E^\complement$. $E$ might be thought about as an event independent of $\rho$, allowing for individual ``randomisation''. Let $C^i = \{c_i(E) \mid E\in\ms F^i_{\x_i}\}$. That is, at both random moves $\x_i$, $i\in I=\{1,2\}$, the active agent $i$ has two basic choices: ``exit'' and ``continue'', between that $i$ can randomise depending on $i$'s exogenous information, that is, in an $\ms F_{\x_i}^i$-measurable way. It is easily seen that $\ms C^i$ defines a reference choice structure on $\{\x_i\}$ and that $C^i$ is a set of $\ms F^i$-$\ms C^i$-adapted choices, for both $i\in I$.

\subsection{Action path stochastic decision forests}\label{subs:APsdf_AC}

Finally, we consider the action path \textsc{sdf} $(F,\pi,\X)$ from Subsection~\ref{subs:APsdf}, induced by action path \textsc{sdf} data $(I,\A,\T,W)$ on an exogenous scenario space $(\Omega,\ms E)$. \smallskip

Let $t\in\T$. For any set $A_{<t}\subseteq \A^{[0,t)_\T}$ and any family $A_t = (A_{t,\omega})_{\omega\in\Omega}\in\mc P(\A)^\Omega$ of subsets of $\A$, let
\[ c(A_{<t},A_t) = \{(\omega,f)\in W \mid f|_{[0,t)_\T} \in A_{<t},\, f(t) \in A_{t,\omega} \}. \]
For $t\in\T$, let $\ms C_t$ be the set of all $c(A_{<t},A_t)$ ranging over all $A_{<t}\subseteq\A^{[0,t)_\T}$ and all families $A_t = (A_{t,\omega})_{\omega\in\Omega}\in \mc P(\A)^\Omega$ of subsets of $\A$ satisfying the following assumptions:
\begin{itemize}[label=--]
    \item \hypertarget{Ass:AP.C0}{\textbf{Assumption~AP.C0.}}~$c(A_{<t},A_t)\neq\emptyset$.
    \item \hypertarget{Ass:AP.C1}{\textbf{Assumption~AP.C1.}}~For all $w\in c(A_{<t},A_t)$, there is $w'\in x_t(w)\setminus c(A_{<t},A_t)$. 
    \item \hypertarget{Ass:AP.C2}{\textbf{Assumption~AP.C2.}}~For all $f\in\A^\T$ with $f|_{[0,t)_\T} \in A_{<t}$, we have \[x_t(\omega,f)\cap c(A_{<t},A_t) \neq \emptyset\] for all or for no $\omega\in D_{t,f}$.
\end{itemize}
So we consider choices that correspond to actions at a predefined time $t$. Again, there is some sort of duality here: action paths are the result of progressive choosing; choices are collections of action paths, essentially. More precisely, $c(A_{<t},A_t)\in\ms C_t$ describes the choice of an action in $A_{t,\omega}$ in scenario $\omega$ and at time $t$, given the history is contained in $A_{<t}$. We assume that such an action is really possible for at least some scenario (Assumption~\hyperlink{Ass:AP.C0}{AP.C0}), that it really constitutes a choice in that there is an alternative (Assumption~\hyperlink{Ass:AP.C1}{AP.C1}), and that it is complete in the sense that it only trivially intersects with random moves (Assumption~\hyperlink{Ass:AP.C2}{AP.C2}).

The principal down-sets and sets of predecessors of such choices take the expected form, as affirmed by the following lemmata. In particular, such a choice $c(A_{<t},A_t)$ is available exactly at all those moves $x$, whose time $\mf t(x)$ is $t$, and that contain an outcome compatible with the choice.

\begin{lemma}\label{lemma:AP_downarrow_c}
    For all $t\in\T$ and $c \in\ms C_t$, we have:
    \[ \downarrow c = \{ x_u(w) \mid w \in c,\,u\in\T\colon t<u\} \cup \{\{w\} \mid w\in c\}.  \]
\end{lemma}

\begin{lemma}\label{lemma:AP_P(c)}
    For all $t\in\T$ and $c \in\ms C_t$, we have:
    \[ P(c) = \{ x_t(w) \mid w \in c\}.  \]
\end{lemma}

As a result, we obtain a large class of non-redundant and complete choices for action path \textsc{sdf}s.

\begin{lemma}\label{lemma:C_t_non-redundant_complete}
    Let $t\in\T$ and $c\in \ms C_t$. Then, $c$ defines a non-redundant and complete choice. Moreover, for all $\x\in\X$ that $c$ is available at, there is $(\omega,f)\in c$ such that $\omega\in D_{t,f} = D_\x$ and $\x = \x_t(f)$.
\end{lemma}

We fix an action index $i\in I$ (in perspective, modelling an agent), a set of random moves $\tilde\X^i$ (of that agent), and an exogenous information structure $\ms F^i = (\ms F^i_\x)_{\x\in\tilde\X^i}$ (this agent is equipped with). We denote the canonical projection $\A\to\A^i$ by $p^i$.
For all $t\in\T$ and $\x\in\tilde\X^i$ with $\mf t(\x) = t$, let $\ms C^i_{\x}$ be \emph{a} set of sets $c(A_{<t},A_t)$ as above such that
\begin{enumerate}
    \item\label{def:msC.1} $A_{<t}\subseteq\A^{[0,t)_\T}$;
    \item\label{def:msC.2} $A_t = (A_{t,\omega})_{\omega\in\Omega}$ such that there is $A_t^i\in \ms B(\A^i)$ satisfying, for all $\omega\in \Omega$, \[A_{t,\omega} = \begin{cases} (p^i)^{-1}(A^i_t), &\quad \omega\in D_\x, \\ \emptyset, &\quad \omega\notin D_\x; \end{cases}\]
    \item\label{def:msC.3} $c(A_{<t},A_t)\in\ms C_t$; and
    \item\label{def:msC.4} for all $\omega\in D_{\x}$, $\x(\omega) \cap c(A_{<t},A_t) \neq\emptyset$.
\end{enumerate}
These properties are referenced as ($\ms C^i_\x$.\ref{def:msC.1}) etc. Hence, a choice $c(A_{<t},A_t)\in\ms C^i_\x$ allows for choosing a measurable set of individual actions for ``agent'' $i$ at the random move $\x$ given the endogenous past $A_{<t}$.

\begin{proposition}\label{prop:APsdf_LCS}
    $\ms C^i = (\ms C^i_\x)_{\x\in\tilde\X^i}$ defines a reference choice structure on $\tilde\X^i$.
\end{proposition}

Let $t\in\T$, $A_{<t}\subseteq \A^{[0,t)_\T}$, $D\in\ms E$, and $g\colon D\to\A^i$. Let $A_t^{i,g} = (A_{t,\omega}^{i,g})_{\omega\in\Omega}$ be given by
\[ A_{t,\omega}^{i,g} = \begin{cases} \{a\in \A \mid p^i(a) = g(\omega) \}, &\quad\omega\in D, \\ \emptyset, &\quad \omega\notin D. \end{cases} \]
Let $c(A_{<t},i,g) = c(A_{<t},A_t^{i,g})$. Provided it is an element of $\ms C_t$, this models the choice of ``agent'' $i$, given an endogenous history in $A_{<t}$, to take the random action $g$, that is, the action $g(\omega)$ in scenario $\omega\in D$, at time $t$.
The perfect (endogenous) information case corresponds to $A_{<t} = \{f|_{[0,t)_\T}\}$ for some $f\in\A^\T$ such that $D_{t,f}\neq\emptyset$.

\begin{example}
    For $\T=\R_+$, singleton $\Omega$, singleton $A_{<t} = \{f|_{[0,t)_\T}\}$ for $t\in\T$ and $f\in\A^\T$, and $W = \Omega \times \A^\T \cong \A^\T$, $c(A_{<t},i,g)$ corresponds to the so-called ``differential game'' choice in \cite[Section~2]{AlosFerrer2011Comment} (see also \cite[Example~4.14]{AlosFerrer2016}), where agent $i$ chooses the value of $g$ given the history $f|_{[0,t)_\T}$.
\end{example}

\begin{example}\label{ex:adapted_choice_vs_adapted_process}
    If $\T = \{0,1,\dots,N\}$ for some $N\in\N$ or $\T = \N$, then these $c(A_{<t},i,g)$ reflect the usual specification of choices in terms of adapted processes from decision problems in discrete time (see, e.g.\ \cite{Bertsekas1996Stochastic}) if $g$ is supposed to be $\ms F_\x^i$-adapted at all $\x\in\tilde\X^i$ that $c(A_{<t},i,g)$ is available at. For the case of continuous-time, $\T=\R_+$ (see, e.g.\, \cite{Pham2009Continuous}), a similar remark can be made (note, however, that additional regularity conditions along time -- progressive measurability, optionality, predictability, for example -- are imposed for complete contingent plans of action, which is discussed in the third part).
\end{example}

The next and final theorem of this paper is concerned with the following question, which can be motivated by the preceding Example~\ref{ex:adapted_choice_vs_adapted_process}. Provided $c = c(A_{<t},i,g) \in \ms C_t$, what is the link between the adaptedness of $c$ and the measurability of the function $g|_{D_\x}$ with respect to $\ms F^i_\x$, for all $\x\in\tilde\X^i$ that $c$ is available at?\footnote{As stated in the theorem below, the availability of $c$ at $\x$ implies $D_\x \subseteq D$.} Indeed, in stochastic control and game theory, action path specifications are often implicitly used and choices are defined by assigning an action $g(\omega)$ to any scenario $\omega$, for $g$ measurable with respect to the value of the underlying filtration at the current time (see, e.g.\ \cite[Chapter~8]{Bertsekas1996Stochastic}, \cite{Pham2009Continuous,Cohen2015Stochastic,Carmona2018}). Hence, if we can answer the question above by providing a tight link, then this customary modelling paradigm in stochastic control and game theory can be explained as a derivative of adapted choices in stochastic decision forests and can therefore be interpreted in terms of traditional decision theory.

The theorem below affirms that the $\ms F^i_\x$-measurability of $g|_{D_\x}$ for relevant $\x\in\tilde\X^i$ is sufficient for the adaptedness of $c$. Moreover, it is also necessary, provided the reference choice structure sufficiently reflects the Borel $\sigma$-algebra on $\A^i$. 
For making this precise, let us call a set $\mc M$ \emph{stable under non-trivial intersections} iff for all $A,B\in\mc M$ with $A\cap B\neq\emptyset$, we have $A\cap B\in\mc M$. This property is equivalent to saying that $\mc M \cup \{\emptyset\}$ is stable under intersections. For example, any partition has this property. 
Now, given some $t\in\T$, $A_{<t}\subseteq \A^{[0,t)_\T}$, $i\in I$, $D\in\ms E$, and $g\colon D \to \A^i$ such that $\tilde c = c(A_{<t},i,g)\in\ms C_t$, we consider the following assumption:
\begin{itemize}[label=--]
    \item \hypertarget{Ass:AP.C3}{\textbf{Assumption~AP.C3}} on the triple $(A_{<t},i,g)$ and on $\ms C^i$.~For all $\x\in\tilde\X^i$ that $\tilde c$ is available at, there is
    a generator $\ms G(\A^i)$ of the Borel $\sigma$-algebra of $\A^i$, stable under non-trivial intersections, such that for all $G\in\ms G(\A^i)$, upon letting $A^{i,G}_t = (A^{i,G}_{t,\omega})_{\omega\in\Omega}$ be given by $A^{i,G}_{t,\omega} = (p^i)^{-1}(G)$ for $\omega\in D_\x$ and $A^{i,G}_{t,\omega} = \emptyset$ for $\omega\notin D_\x$, we have \[c(A_{<t},A^{i,G}_t)\in \ms C^i_{\x}.\]
\end{itemize}

\begin{thm}\label{thm:APsdf_AC}
    Let $(F,\pi,\X)$ be the action path \textsc{sdf} induced by action path \textsc{sdf} data $(I,\A,\T,W)$ on an exogenous scenario space $(\Omega,\ms E)$. Further, let $i\in I$ be an action index, $\tilde\X^i\subseteq\X$ be a set of random moves, $\ms F^i$ be an exogenous information structure on $\tilde\X^i$ and $\ms C^i$ be a reference choice structure on $\tilde\X^i$ as above satisfying Axioms ($\ms C^i_\x$.$k$), $k=1,\dots,4$. 
    Let $t\in\T$, $A_{<t}\subseteq \A^{[0,t)_\T}$, $D\in\ms E$, and $g\colon D\to\A^i$ such that $c(A_{<t},i,g)\in\ms C_t$. 
    
    Then, we have:
    \begin{enumerate}
        \item\label{thm:APsdf_AC.non_red_and_compl} $c(A_{<t},i,g)$ is a non-redundant and $\tilde\X^i$-complete choice. 
        \item\label{thm:APsdf_AC.Dx_subset_D} For all $\x\in\tilde\X^i$ that $c(A_{<t},i,g)$ is available at, we have $D_\x \subseteq D$.
        \item\label{thm:APsdf_AC.Fx_mb_=>_adapted} If for all $\x\in\tilde\X^i$ that $c(A_{<t},i,g)$ is available at $g|_{D_\x}$ is $\ms F^i_{\x}$-measurable, then $c(A_{<t},i,g)$ is $\ms F^i$-$\ms C^i$-adapted.
        \item\label{thm:APsdf_AC.adapted_=>_Fx_mb} If Assumption~\hyperlink{Ass:AP.C3}{AP.C3} is satisfied for $(A_{<t},i,g)$ and $\ms C^i$, and if moreover $c(A_{<t},i,g)$ is $\ms F^i$-$\ms C^i$-adapted, then $g|_{D_\x}$ is $\ms F^i_{\x}$-measurable for all $\x\in\tilde\X^i$ that $c(A_{<t},i,g)$ is available at.
    \end{enumerate}
\end{thm}

\begin{example}\label{ex:APsdf_AC}
    Recall that we consider the action path \textsc{sdf} $(F,\pi,\X)$ induced by action path \textsc{sdf} data $(I,\A,\T,W)$ on an exogenous scenario space $(\Omega,\ms E)$. 
    \begin{itemize}[label=--]
        \item If we consider the simple \textsc{sdf}s from the preceding subsection in action path formulation, according to Lemma \ref{lemma:simple_sdf_as_APsdf}, then the preceding construction yields exactly the adapted choices from the preceding subsection.
        
        \item Suppose that $W=\Omega\times\A^\T$. Fix some $i\in I$ and suppose that $\A^i$ has at least two elements. Let $t\in\T$ and $A_{<t}\subseteq \A^{[0,t)_\T}$.
        Then, for any map $g\colon\Omega\to\A^i$ we have $c(A_{<t},i,g) \in \ms C_t$, that is, it satisfies Assumptions \hyperlink{Ass:AP.C0}{AP.C0}, \hyperlink{Ass:AP.C1}{AP.C1}, and \hyperlink{Ass:AP.C2}{AP.C2}. Moreover, if for all $\x\in\tilde\X^i$ with $\mf t(\x) = t$, $\ms C^i_\x$ contains all $c(A_{<t},A_t)$ ranging over all $A_t$ satisfying ($\ms C_\x^i$.$k$), $k=2,3,4$, then Assumption~\hyperlink{Ass:AP.C3}{AP.C3} is satisfied for $(A_{<t},i,g)$ and $\ms C^i$.
        
        \item In case of the timing problem (see Example~\ref{ex:APsdf}, and, e.g.\ \cite{Shiryaev2007Optimal,Karatzas1998Methods}), we have $\A^i = \{0,1\}$ for all $i\in I$. Fix $i\in I$. Let $t\in\T$ and $A_{<t}\subseteq \A^{[0,t)_\T}$ be a non-empty set of componentwise decreasing paths $f_t$ such that $p^i\circ f_t = 1_{[0,t)_\T}$. 
        Then, for all $g\colon\Omega\to\{0,1\}$, we have $c(A_{<t},i,g)\in\ms C_t$, that is, it satisfies Assumptions \hyperlink{Ass:AP.C0}{AP.C0}, \hyperlink{Ass:AP.C1}{AP.C1}, \hyperlink{Ass:AP.C2}{AP.C2}. Moreover, if for all $\x\in\tilde\X^i$ with $\mf t(\x) = t$, $\ms C^i_\x$ contains all $c(A_{<t},A_t)$ ranging over all $A_t$ satisfying ($\ms C_\x^i$.$k$), $k=2,3,4$, then Assumption~\hyperlink{Ass:AP.C3}{AP.C3} is satisfied for $(A_{<t},i,g)$ and $\ms C^i$.
        
        \item An analogous statement holds true in the case of the up-and-out option control problem (see Example~\ref{ex:APsdf}, and \cite{Hull2018Options}). Here $I$ is a singleton, $\T = \R_+$, and $\A^i = \A = \{0,1\}$ for the unique $i\in I$. Let $t\in\R_+$ such that, with  
        \[ D = \{ \omega\in \Omega \mid \max_{u\in[0,t]} P_u(\omega) < 2\}, \]
        we have $D\neq\emptyset$. Further, let $g\colon D\to\{0,1\}$ be a map and let $A_{<t} = \{1\}^{[0,t)}$. 
        Then, we have $c(A_{<t},i,g)\in\ms C_t$, that is, it satisfies Assumptions \hyperlink{Ass:AP.C0}{AP.C0}, \hyperlink{Ass:AP.C1}{AP.C1}, \hyperlink{Ass:AP.C2}{AP.C2}. Moreover, if for all $\x\in\tilde\X^i$ with $\mf t(\x) = t$, $\ms C^i_\x$ contains all $c(A_{<t},A_t)$ ranging over all $A_t$ satisfying ($\ms C_\x^i$.$k$), $k=2,3,4$, Assumption~\hyperlink{Ass:AP.C3}{AP.C3} is satisfied for $(A_{<t},i,g)$ and $\ms C^i$.
    \end{itemize}
    The proofs of these claims can be found in the corresponding part of the appendix.
\end{example}

It is important to note that a choice $c = c(A_{<t},i,g) \in \ms C_t$ as above, available at a random move $\x$, with $A_{<t}\subseteq \A^{[0,t)_\T}$ and $\ms F^i_\x$-measurable $g\colon D \to \A^i$, $D\in\ms E$, can condition on two distinct things: first, on the endogenous information $A_{<t}$, and second, on the exogenous information via the $\ms F^i_\x$-measurability of $g$. For instance, one can the imagine the case in that there are $\x'\in\tilde\X^i \setminus \{\x\}$ with $\ms F^i_\x = \ms F^i_{\x'}$ and $\mf t(\x) = t = \mf t(\x')$, $\ms F^i_{\x}$-measurable $g'\colon D \to \A^i$ with $g\neq g'$ and $A'_{<t}\subseteq \A^{[0,t)_\T}$ such that $c(A'_{<t},i,g')$ is available at $\x'$ and an agent's complete contingent plan of action specifies $c(A_{<t},i,g)$ at $\x$, but $c(A'_{<t},i,g')$ at $\x'$.
In the context of Example~\ref{ex:APsdf_EIS}, with $\x = \x_t(f)$ and $\x' = \x_t(f')$ for suitable $f,f'\in \A^\T$, this would mean that, even if $Y_{\x_u(f)}$, $u\le t$, and $Y_{\x_u(f')}$, $u\le t$, generate the same $\sigma$-algebra, i.e.\ $\ms F^i_{\x_t(f)} = \ms F^i_{\x_t(f')}$, an agent need not make the very choice the agent makes at $\x_t(f)$ at the counterfactual random move $\x_t(f')$ as well.

Let us close on a note about this specific example for that we consider an agent only having choices of the form $c = c(A_{<t},i,g)$ as in the preceding paragraph. In case $A_{<t} = \A^{[0,t)_\T}$ for all choices available to this agent, and moreover $\ms F^i_{\x_t(f)} = \ms G_t$ for all $f\in\A^\T$ and $t\in D_{t,f}$, and a fixed filtration $\ms G = (\ms G_t)_{t\in\T}$ (case 1 in Subsection~\ref{subs:APsdf_EIS}), then the complete contingent plans of action this agent can build, model what is described by the term ``open loop'' in the control- and game-theoretic literature. If, conversely, only choices with singleton $A_{<t}$ are available to the agent and $\ms F^i_{\x_t(f)}$ depends on a stochastic state controlled by $f$, for $f\in\A^\T$ and $t\in D_{t,f}$, in the sense of case 2 of Subsection~\ref{subs:APsdf_EIS}, possibly in combination with Example~\ref{ex:APsdf_EIS}, then we obtain what is subsumed under the term ``closed loop''  (see, e.g.\ \cite{Fudenberg1988Open,Fudenberg1991}).

Note that there are mixed regimes. ``Large'' alias uninformative $A_{<t}$ and $\ms F^i$ like in case 2 of Subsection~\ref{subs:APsdf_EIS} may coexist as well, which is often also called ``closed loop'' in the control and differential games literature (see, e.g.\ the corresponding discussion in \cite[p.\ 72--76]{Carmona2018}), but actually the complete contingent plans of action the agent can make only close the loop with respect to exogenous information: the agent can only react to the controlled noise, but not explicitly to the employed controls themselves which can make a crucial difference when evaluating counterfactuals, especially in case several agents interact. 
On the other hand, ``small'' alias informative $A_{<t}$ may coexist with case 1 of Subsection~\ref{subs:APsdf_EIS}, that is, ``$\ms F^i_{\x_t(f)} = \ms G_t$'', which is near, or in the extreme case of singleton $A_{<t}$ equal to closed loops with respect to endogenous information, but is equivalent, on the exogenous information side, to the ignorance of any controlled stochastic state. In any case, we conclude, that the terminology of ``closed'' and ``open'' loops must be used with respect to both endogenous and exogenous information. It requires some caution in the mixed regimes just discussed, especially if the underlying stochastic decision forest, exogenous information structures and adapted choices are not properly specified.

\section*{Conclusion}
\addcontentsline{toc}{section}{Conclusion}

It is possible to implement general stochastic processes as background noise on refined partitions-based decision forests without encountering outcome generation problems for a ``nature'' agent, while allowing for a rigorous decision-theoretic interpretation of the relationship between endogenous and exogenous information and choices. This represents an improvement over the existing state of the art in both refined partitions-based decision and game theory (e.g.\ in \cite{AlosFerrer2016,AlosFerrer2015}) and stochastic control and differential games theory (e.g.\ \cite{Pham2009Continuous,Karatzas1998Methods,Carmona2018,Cohen2015Stochastic}). This has been achieved by abandoning the assumption of a ``nature'' agent and instead constructing a theory of stochastic decision forests. These decision forests satisfy duality between outcomes and nodes and can be represented as forests of decision trees. They allow for a notion of similarity across moves on different trees, given by random moves, which, under weak hypotheses on the underlying order-theoretic structure, form the moves of a decision tree in its own right. They serve as the basis for both a structure of exogenous information revelation similar to filtrations in probability theory and a concept of adapted choices compatible with exogenous information structures. Adapted choices capture measurability assumptions on choices typically made in the literature and also provide insight into the decision-theoretic meaning of open and closed loop controls in the stochastic setting. 

In addition, a general model of action path stochastic decision forests has been constructed, based on a small number of easily verifiable conditions. This model can serve as a unified decision-theoretic foundation for a large class of stochastic decision problems, including those in discrete time. Moreover, as will be shown in the third paper, it constitutes a step toward approximately explaining the extensive form characteristics of stochastic decision problems in continuous time.\smallskip

In order to develop a theory of stochastic extensive forms that encompasses the classical theory in \cite{AlosFerrer2016}, it is still necessary to formulate consistency criteria on stochastic decision forests equipped with a set of agents, each agent provided with an exogenous information structure, a reference choice structure, and a set of adapted choices. This would allow us to define strategies as Savage acts, as prepared in this paper, and to address the question of whether and, if so, how outcome existence and uniqueness for strategy profiles -- and in that sense, essentially, ``well-posedness'' -- can be characterised for stochastic extensive forms. Based on this, we can address the question, implicit in the present paper, of which stochastic extensive forms can and which cannot be faithfully represented as extensive forms with an additional ``nature'' agent that, by executing a strategy, generates exogenous information as endogenous information. Furthermore, we note, in a spirit related to \cite{Riedel2017}, that in order to allow for subgame-perfect or perfect Bayesian equilibrium analysis, a decision-theoretically explainable and general model of subgames, histories, and information sets in stochastic extensive form decision problems needs to be established. Moreover, such an analysis requires an abstract concept of preferences, e.g.\ via beliefs and utility functions, following entirely from basic decision-theoretic principles. Last but not least, such a theory is expected to enable the construction and study of stochastic extensive forms based on action path data, providing a large class of game-theoretic models and, especially, the basis of an approximation theory for stochastic differential games. This remains to be shown as well. Together, all of this constitutes the programme for the second paper \cite{Rapsch2024DecisionB} of the present series.

\section*{Acknowledgements}
\addcontentsline{toc}{section}{Acknowledgements}
First and foremost, the author is particularly indebted to his doctoral supervisor, Christoph Knochenhauer, who read and commented on several earlier versions of this paper, encouraged him to pursue this project, and actively supported him in presenting its contents at workshops. Special thanks are due to Frank Riedel for the helpful discussion we had in Berlin in 2022, especially for providing useful insights regarding the game-theoretic literature. The author also wants to thank Elizabeth Baldwin, Miguel Ballester, and Samuel N.\ Cohen for inspiring discussions on decision and control theory during his stay in Oxford in 2023. The author is also grateful to the attendants of the workshops and seminars where he had the opportunity to present and discuss earlier versions of this particular project, in Oxford, Kiel, Berlin, Palaiseau, including Daniel Andrei, Karolina Bassa, Peter E.\ Caines, Fanny Cartellier, Sebastian Ertel, Ali Lazrak, Kristoffer Lindensjö, Christopher Lorenz, Berenice Anne Neumann, Manos Perdikakis, Jan-Henrik Steg, Peter Tankov, and Jacco Thijssen, for their questions and comments. 
Partial funding by Deutsche Forschungsgemeinschaft (DFG) through, first, the \href{https://gepris.dfg.de/gepris/projekt/410208580}{IRTG 2544} Stochastic Analysis in Interaction, Project ID: 410208580, and, second, under Germany's Excellence Strategy, the \href{https://gepris.dfg.de/gepris/projekt/390685689}{EXC 2046/1} The Berlin Mathematics Research Center MATH+, Project ID: 390685689, is gratefully acknowledged. 

%% file: contents/appendix.tex
\section{Proofs}

\subsection{Section~\ref{sec:def}}

We first prove Lemma~\ref{lemma:partion_of_forest} and Theorem~\ref{thm:decision_forest=forest_of_decision_trees}, without using the Propositions \ref{prop:f(v)=uparrow v}, \ref{prop:decision_forest_over_set_is_decision_forest}, and \ref{prop:repr_by_dec_paths_of_decision_forest}.

\begin{proof}
[Proof of Lemma~\ref{lemma:partion_of_forest}] 
    \footnote{We recall that this lemma can actually be seen as an explicitly order-theoretic reformulation of a basic result from graph theory (see the discussion in \cite[Section~I.1]{Bollobas2013Modern}). As the claim that it is a reformulation requires proof, and also for the reader's convenience, a proof is given nonetheless.}
    For $x,y\in F$, let $x\sim y$ iff there is $z\in F$ such that $z\ge x$ and $z\ge y$. $\sim$ is clearly reflexive and symmetric. It is also transitive, because any principal up-set in $(F,\ge)$ is a chain and $\ge$ is transitive. The equivalence classes of $\sim$ define trees by construction of $\sim$.
	If there are $x,y\in F$ such that $x\ge y$, then (by taking $z=x$) we obtain that $x\sim y$. Hence, the equivalence classes of $\sim$ define a partition with the claimed property.
	
	Let $\mc F$ be an arbitrary partition with the claimed property. Then, for $x,y\in F$, $x$ and $y$ belong to the same partition member iff $x\sim y$. Indeed, if they belong to the same partition member $T$, which by assumption is a tree, then $x\sim y$ by definition of a tree. If conversely $x\sim y$, then there is $z\in F$ with $z\ge x$ and $z\ge y$. Hence, there are $T, T'\in\mc F$ such that $x,z\in T$ and $y,z\in T'$. Thus $T\cap T' \neq \emptyset$, whence $T=T'$. Hence, $\mc F$ is uniquely determined.

    We now prove the statement of the second sentence. Suppose the forest $(F,\ge)$ to be rooted, and let $T\in\mc F$. By definition of $\mc F$, $T$ is a non-empty tree. Let $x\in T$. Then the principal up-set $\uparrow x = \{y\in F \mid y \ge x\}$ in $F$ contains a maximal element $y$. By definition of $T$, $y\in\uparrow x$ implies $y\in T$. $\uparrow x$ is a chain, because $(F,\ge)$ is a forest. Hence, $y$ is even a maximum of $\uparrow x$ with respect to the induced order. Let $z\in T$ be an arbitrary element of the tree $T$. Then $T\cap\uparrow x \cap \uparrow z \neq \emptyset$. Take an element $u$ in this intersection. Then $y\ge u \ge z$, whence $y\ge z$. Thus $y$ is a maximum of $(T,\ge)$. In particular, $y$ is maximal in $(T,\ge)$. Therefore, $(T,\ge)$ is a rooted tree according to our definition and has a maximum.
\end{proof}

We prepare the proof of Theorem~\ref{thm:decision_forest=forest_of_decision_trees} with two lemmata.

\begin{lemma}\label{lemma:set_forest}
	Let $V$ be a set and $F$ be a $V$-poset such that $(F,\supseteq)$ is a rooted forest. Let $\mc F$ be the set of its connected components and for any $T\in\mc F$, let $V_T$ be the root of $(T,\supseteq)$. Then, for any $T\in\mc F$, we have:
	\begin{enumerate}
		\item\label{lemma:set_forest.T=downarrow_V_T} $T = \{x\in F \mid V_T \supseteq x\}$;
		\item\label{lemma:set_forest.W(V_T)=w_subseteq_T} $W(V_T) = \{w\in W \mid w\subseteq T\}$;
		\item\label{lemma:set_forest.W(x)} for any $x\in T$, $W(x)$ is equal to the set of maximal chains $w$ in $(T,\supseteq)$ with $x\in w$; in particular, $W(V_T)$ is equal to the set of maximal chains in $(T,\supseteq)$.
	\end{enumerate}
\end{lemma}

\begin{proof}
	(Ad 1):~ If $x\in T$, then $V_T\supseteq x$, because $V_T$ is the root of $T$. Conversely, let $x\in F$ such that $V_T\supseteq x$. In combination with $V_T\in T$, we obtain $x\in T$, because $T\in\mc F$.\smallskip
	
	(Ad 2):~ Let $w\in W(V_T)$ and $x\in w$, then $V_T\supseteq x$ or $x\supseteq V_T$, whence $x\in T$, by definition of $V_T$ and $\mc F$. Thus $w\subseteq T$. Conversely, let $w\in W$ such that $w\subseteq T$. By definition of $V_T$, any $x\in w$ satisfies $V_T \supseteq x$. Because of maximality of the chain $w$, we must have $V_T\in w$.\smallskip
	
	(Ad 3):~ Let $x\in T$.
	
	Let $w\in W(x)$. Then $w\subseteq T$ since, for any $y\in w$ we have $x\supseteq y$ or $y\supseteq x$, whence $y\in T$. The chain $w$ is also maximal in $(T,\supseteq)$ because $T$ is a subset of $F$ and $w$ is maximal in $(F,\supseteq)$ by assumption.
	
	Now let $w$ be a maximal chain in $(T,\supseteq)$ with $x\in w$. By definition of $V_T$, we have $V_T\supseteq y$ for all $y\in w$, and maximality of the chain $w$ implies that $V_T\in w$. Let $y\in F$ be such that $w\cup\{y\}$ is a chain in $(F,\supseteq)$. Hence, $V_T\supseteq y$ or $y\supseteq V_T$, in particular we get $y\in T$. As $w$ is maximal in $(T,\supseteq)$, $y\in w$. Hence $w$ is already maximal in $(F,\supseteq)$.
	
	The second part follows from the first, since all maximal chains $w$ in $(T,\supseteq)$ satisfy $V_T\in w$. Indeed, as a maximal chain in the rooted tree $(T,\supseteq)$, $w$ is non-empty, and any element $x\in w\subseteq T$ satisfies $V_T\supseteq x$ by definition of $V_T$. By maximality of the chain $w$, we infer $V_T\in w$.
\end{proof}

\begin{lemma}\label{lemma:game_forest}
	Let $F$ be a decision forest on a set $V$ with $f$ as in Definition~\ref{def:decision_forest}. Let $\mc F$ be the set of $(F,\supseteq)$'s connected components. Then, for any $v\in V$ and any $T\in\mc F$, we have $f(v)\subseteq T$ iff $v\in V_T$.
\end{lemma}

\begin{proof}
	Let $v\in V$ and $T\in\mc F$. By Proposition~\ref{prop:f(v)=uparrow v}, $v\in V_T$ iff $V_T\in f(v)$. By Lemma~\ref{lemma:set_forest}, Part~\ref{lemma:set_forest.W(V_T)=w_subseteq_T}, this is equivalent to $f(v) \subseteq T$.
\end{proof}

\begin{proof}[Proof of Theorem~\ref{thm:decision_forest=forest_of_decision_trees}]
	(Ad necessity):~ Suppose $F$ to be a decision forest on $V$. Let $W$ be the set of maximal chains in $(F,\supseteq)$, and $f\colon V \rightarrow W$ a bijection such that, for every $x\in F$, $(\Pot f)(x) = W(x)$, just as in Definition~\ref{def:decision_forest}.
	
	We claim that for all $v\in V$ there is a unique $T\in\mc F$ such that $f(v)\subseteq T$. Property \ref{thm:decision_forest=forest_of_decision_trees.partition} then follows from Lemma~\ref{lemma:game_forest}.
	
	For the proof, let $v\in V$. As s maximal chain in the rooted forest $(F,\supseteq)$, $f(v)$ is non-empty. Hence, there is $x\in f(v)$.
	As $\mc F$ is a partition of $F$, there is a unique $T\in\mc F$ with $x\in T$. For arbitrary $y\in f(v)$, we have $x\subseteq y$ or $y\subseteq x$, since $f(v)$ is a chain. Hence, $x\sim y$, or in other words $y\in T$. Thus $T$ is the unique element of $\mc F$ with $f(v)\subseteq T$.
	
	For Property \ref{thm:decision_forest=forest_of_decision_trees.T_decision_tree}, let $T\in\mc F$. Then $(T,\supseteq)$ defines a rooted tree with root $V_T$, by definition of $\mc F$. Moreover, for all $v\in V$, we have $f(v)\in W(V_T)$ iff $v\in V_T$, by Proposition~\ref{prop:f(v)=uparrow v}. Hence, $f_T:=f|_{V_T}$ yields a bijection $V_T \rightarrow W(V_T)$. 
	
	By Lemma~\ref{lemma:set_forest}, Part~\ref{lemma:set_forest.W(x)}, second sentence, $W(V_T)$ equals the set of maximal chains in $(T,\supseteq)$. By definition of $V_T$, all $x\in T$ satisfy $V_T\supseteq x$ so that the definition of $f$ yields:
    \[ (\mc P f|_T)(x) = (\mc P f)(x) = W(x). \]
    By Lemma~\ref{lemma:set_forest}, Part~\ref{lemma:set_forest.W(x)}, first sentence, $W(x)$ consists exactly of all maximal chains in $(T,\supseteq)$ containing $x$. Hence, $T$ is its own representation by decision paths.\smallskip
	
	(Ad sufficiency):~ Let $F$ satisfy Properties \ref{thm:decision_forest=forest_of_decision_trees.partition} and \ref{thm:decision_forest=forest_of_decision_trees.T_decision_tree} from the theorem. It remains to show that $F$ is its own representation by decision paths, i.e.\ Part~\ref{def:decision_forest:repr_by_dec_paths} in Definition~\ref{def:decision_forest}.
	
	By Property \ref{thm:decision_forest=forest_of_decision_trees.T_decision_tree} and Lemma~\ref{lemma:set_forest}, Part~\ref{lemma:set_forest.W(x)}, for each $T\in\mc F$, there is a bijection $f_T\colon V_T\rightarrow W(V_T)$ such that for all $x\in T$, $(\mc P f_T)(x)$ equals the set of decision paths in $(T,\supseteq)$ containing $x$. By Part~\ref{lemma:set_forest.W(x)} of Lemma~\ref{lemma:set_forest}, this set is equal to $W(x)$. By Property \ref{thm:decision_forest=forest_of_decision_trees.partition}, we obtain a map $f\colon V \rightarrow W$ satisfying, for any $T\in\mc F$ and $v\in V_T$, $f(v) = f_T(v)$. $f$ is injective by Lemma~\ref{lemma:set_forest}, Part~\ref{lemma:set_forest.W(V_T)=w_subseteq_T}, since $\mc F$ is a partition. $f$ is also surjective, hence bijective, since $F\neq\emptyset$ and so every maximal chain $w\in W$ contains an element $x$ for which there is $T\in\mc F$ with $x\in T$. Hence, $w\in W(x)$ which is a subset of $W(V_T)$ by Lemma~\ref{lemma:set_forest}, Part~\ref{lemma:set_forest.W(x)}. Hence, there is $v\in V_T$ with $f(v) = f_T(v) = w$.
	
	Let $x\in F$. Then there is $T\in\mc F$ with $x\in T$. Then $V_T \supseteq x$. Hence, $(\Pot f)(x) = (\Pot f_T)(x)$. By definition of $f_T$ and Part~\ref{lemma:set_forest.W(x)} of Lemma~\ref{lemma:set_forest}, we have  
	\[ (\Pot f_T)(x) = \{ w\in W(V_T) \mid x\in w \}. \]
	Lemma~\ref{lemma:set_forest}, Part~\ref{lemma:set_forest.W(x)}, implies that the latter set equals $W(x)$. Thus, $(\mc P f)(x) = W(x)$.
\end{proof}

\subsection{Section~\ref{sec:sdf}}

\begin{proof}[Proof of Lemma~\ref{lemma:sdf.X.roots}]
    (Ad Part~\ref{lemma:sdf.X.roots.OC}):~ Suppose that $(F,\pi,\X)$ is order consistent. Let $\x\in\X$. Suppose there is $\omega\in D_\x$ such that $\x(\omega) = W_\omega$. Let $\omega'\in D_\x$. It suffices to show that $\x(\omega') = W_{\omega'}$.
    
    By the covering property in Axiom~\ref{def:sdf.X.cov}, there is $\x'\in\X$ such that $\x'(\omega') = W_{\omega'}$. Then $\x'(\omega') \supseteq \x(\omega')$. Hence, $\x' \ge_\X \x$ by Axiom~\ref{def:sdf.X.OC}. Thus, by definition of $\ge_\X$, $\omega\in D_{\x'}$ and $W_\omega \supseteq \x'(\omega) \supseteq \x(\omega) = W_\omega$, whence $\x'(\omega) = \x(\omega)$ which in turn implies $\x' = \x$ by order consistency, Axiom~\ref{def:sdf.X.OC}. Thus, $\x(\omega') = \x'(\omega') = W_{\omega'}$.\smallskip

    (Ad Part~\ref{lemma:sdf.X.roots.surely_non-trivial}):~ $D=\Omega$ iff for all $\omega\in\Omega$, $W_\omega\in X$. In view of Axiom~\ref{def:sdf.conn_comp}, the set of roots of $(F,\supseteq)$ is given by precisely all $W_\omega$, $\omega\in\Omega$. The claimed equivalence follows.\smallskip

    (Ad Part~\ref{lemma:sdf.X.roots.root_random_move}):~ Suppose that $(F,\pi,\X)$ is order consistent and maximal, and such that $X\neq\emptyset$. It suffices to show the claim $(\ast)$ that there is exactly one $\x\in\X$ whose image contains a root of $(F,\supseteq)$. Indeed, then, by Axiom~\ref{def:sdf.X.cov}, $D\subseteq D_\x$ and $\x(\omega) = W_\omega = \x_0(\omega)$ for all $\omega\in D$, and by Part~\ref{lemma:sdf.X.roots.OC} of this lemma, proven just before, $D=D_\x$, whence $\x_0 = \x\in\X$.

    Regarding the proof of $(\ast)$, note that by hypothesis there is $x\in X$. As $W_{\pi(x)}\supseteq x$, $W_{\pi(x)}\in X$. By Axiom~\ref{def:sdf.X.cov}, there is $\x\in\X$ with $x\in\im\x$. Hence, if $(\ast)$ did not hold true, there would have to be at least two distinct random moves $\x_1,\x_2\in\X$ whose images contain roots of $F$. By Part~\ref{lemma:sdf.X.roots.OC}, the images of both $\x_1$ and $\x_2$ then would contain roots only. As $\x_1\neq\x_2$, by Axiom~\ref{def:sdf.X.OC}, we would have $D_{\x_1}\cap D_{\x_2} = \emptyset$. Then, let $\bar\x = \x_1 \cup \x_2$,\footnote{That is, $\bar\x\colon D_{\x_1} \cup D_{\x_2}\to X$ such that $\bar\x(\omega) = \x_k(\omega)$ for $\omega\in D_{\x_k}$, for both $k=0,1$.} 
    and $\bar\X = (\X \setminus \{\x_1,\x_2\}) \cup \{\bar\x\}$. 
    
    Clearly, $D_{\bar\x} = D_{\x_1} \cup D_{\x_2}$ would be a non-empty element of $\ms E$ and $\bar\x$ would be a section of $\pi$ on $D_{\bar\x}$. It is evident that $\bar\X$ would be order consistent and induce a covering of $X$. By construction, $\X$ would refine $\bar\X$. Thus, maximality (see Axiom~\ref{def:sdf.X.max}) would imply the false statement $\bar\X = \X$. Hence, $(\ast)$ obtains and the proof is complete.
\end{proof}

\begin{proof}[Proof of Proposition~\ref{prop:ev_on_Tr_is_iso}]
    (Ad ``order embedding''):~ Let $(\y_1,\omega_1),(\y_2,\omega_2)\in\Tr\bullet\Omega$. If both $\y_1,\y_2$ are elements of $\X$, the implication ``$\Rightarrow$'' is a consequence of the definition of $\ge_\X$, and the other implication ``$\Leftarrow$'' can be shown as follows: If $\y_1(\omega_1) \supseteq \y_2(\omega_2)$, then both moves belong to the same connected component of the rooted forest $(F,\supseteq)$ (whose existence Axiom~\ref{def:sdf.df} ensures). By Axioms~\ref{def:sdf.conn_comp} and \ref{def:sdf.X.section}, $\omega_1 = \omega_2$. Then, Axiom~\ref{def:sdf.X.OC} implies that $\y_1 \ge_\X \y_2$, whence $\y_1\ge_\Tr \y_2$. 
    
    Else, for some $k=1,2$, $\y_k$ is a random terminal node so that $D_{\y_k} = \{\omega_k\}$ and $\y_k(\omega_k) = \{w_k\}$ for some $w_k\in W$.    
    If $k=1$, then first, $\y_1 \ge_\Tr \y_2$ is equivalent to $\y_1 = \y_2$. Second, $\y_1(\omega_1) \supseteq \y_2(\omega_2)$ is equivalent to $\y_1(\omega_1) = \y_2(\omega_2)$, because $\y_1(\omega)$ is terminal in $F$. But then $\y_2$ must be a random terminal node as well, whence $\y_1 = \y_2$ and $\omega_1=\omega_2$. Conversely, if these two equalities hold true, then $\y_1(\omega_1) = \y_2(\omega_2)$ follows.
    
    If $k=2$, then $\y_1(\omega_1) \supseteq \y_2(\omega_2)$ is equivalent to the conjunction of these three statements: $\omega_1=\omega_2$, $\omega_2\in D_{\y_1}$, and $w_2\in \y_1(\omega_1)$. By definition of $\ge_\Tr$, this is equivalent to $\y_1 \ge_\Tr \y_2$ and $\omega_1 = \omega_2$.\smallskip

    (Ad ``bijection''):~ Let $(\y_1,\omega_1),(\y_2,\omega_2)\in\Tr\bullet\Omega$ such that $\y_1(\omega_1) = \y_2(\omega_2)$. As the evaluation map is an order embedding, which we have proven just before, we obtain $\y_1 = \y_2$ and $\omega_1 = \omega_2$. Regarding surjectivity, let $y\in F$. If $y$ is terminal, then there is $w\in W$ with $y = \{w\}$ because $F$ is a decision forest on $W$.\footnote{Although this is discussed in \cite{AlosFerrer2005}, we give an argument here for the reader's convenience. Let $f$ be the map from Definition~\ref{def:decision_forest}. If $y$ is terminal, then $(\mc P f)(y) = W(y)$ contains exactly one element, namely $\uparrow y$. As $f$ is injective, $y$ must be a singleton.}
    Hence, $y = \y(\omega)$ for $\omega = \pi(y)$ and the random terminal node $\y = \{(\omega,\{w\})\}$. If $y$ is a move, then, by the definition of $\X$, Property \ref{def:sdf.X.cov}, there is $\x\in\X$ and $\omega\in D_\x$ such that $\x(\omega) = y$.
\end{proof}

\begin{proof}[Proof of Theorem~\ref{thm:Xrm_is_dec_tree}]
    (Ad ``poset''):~ It follows directly from the definitions that $\ge_\Tr$ defines a partial order on $\Tr$.\smallskip
    
    (Ad ``forest''):~ Let $\y_0\in\Tr$ and $\y_1,\y_2\in\uparrow \y_0$. As $D_{\y_0}\neq\emptyset$ by Definition~\ref{def:sdf} and the definition of random terminal nodes, there is $\omega\in D_{\y_0}$, and we have $\omega\in D_{\y_k}$ and $\y_k(\omega) \supseteq \y_0(\omega)$ for both $k=1,2$. As $(F,\supseteq)$ is a forest, $\uparrow \y_0(\omega)$ is a chain, thus $\y_1(\omega) \supseteq \y_2(\omega)$ or $\y_2(\omega) \supseteq \y_1(\omega)$. Hence, by Proposition~\ref{prop:ev_on_Tr_is_iso}, we get $\y_1\ge_\Tr \y_2$ or $\y_2\ge_\Tr \y_1$. We conclude that $(\Tr,\ge_\Tr)$ is a forest. \smallskip

    (Ad ``decision forest''):~ Let $\y_1,\y_2\in\Tr$ be such that $\y_1\neq \y_2$. If they are not comparable by $\ge_\Tr$, that is neither $\y_1 >_\Tr \y_2$ nor $\y_2>_\Tr \y_1$, then any maximal chain $\w$ in $(\Tr,\ge_\Tr)$ with $\y_1\in\w$ satisfies $\y_2\notin\w$. Using the axiom of choice in the form of the Hausdorff maximality principle, there exists indeed such a $\w$. 
    
    For symmetry reasons, it remains to consider the case in that we have $\y_1 >_\Tr \y_2$. Then, by Proposition~\ref{prop:ev_on_Tr_is_iso}, we have for all $\omega^\prime\in D_{\y_2}$, $\y_1(\omega^\prime) \supsetneq \y_2(\omega^\prime)$. As $D_{\y_2} \neq \emptyset$, we can select $\omega\in D_{\y_2}$. We can also select $w\in \y_1(\omega) \setminus \y_2(\omega)$. In particular, $\y_2(\omega) \notin \uparrow \{w\}$. As the latter set is a maximal chain in $(F,\supseteq)$, there is $y_3\in \uparrow \{w\}$ such that $y_3$ and $\y_2(\omega)$ cannot be compared by $\supseteq$. As $\y_1(\omega)$ is also an element of the chain $\uparrow \{w\}$, this implies $\y_1(\omega) \supsetneq y_3$. By Proposition~\ref{prop:ev_on_Tr_is_iso}, there is $\y_3\in\Tr$ such that $\omega\in D_{\y_3}$ and $y_3 = \y_3(\omega)$. This lemma also implies that $\y_2$ and $\y_3$ are not comparable via $\ge_\Tr$. Hence, as we have proven just above, there is a maximal chain $\w$ in $(\Tr,\ge_\Tr)$ with $\y_2\notin\w$ and $\y_3\in\w$. By Proposition~\ref{prop:ev_on_Tr_is_iso}, we obtain $\y_1 \ge_\Tr \y_3$, and using the fact that $(\Tr,\ge_\Tr)$ is a forest, as already proven, we infer that $\y_1\in\w$. Hence there is a maximal chain in $(\Tr,\ge_\Tr)$ separating $\y_1$ and $\y_2$.\smallskip

    (Ad ``rooted tree''):~ In view of the preceding results, it suffices to show that $(\Tr,\ge_\Tr)$ has a maximum. As $(F,\supseteq)$ is non-empty (as a rooted forest) and non-trivial (by sure non-triviality and non-emptiness), $X\neq\emptyset$. By Lemma~\ref{lemma:sdf.X.roots}, Parts \ref{lemma:sdf.X.roots.surely_non-trivial} and \ref{lemma:sdf.X.roots.root_random_move}, the map $\x_0\colon\Omega\to X,\,\omega\to W_\omega$ is a random move, i.e.\ $\x_0\in\X$. Then, we have $\x_0 \ge_\Tr \y$ for all $\y\in\Tr$. Indeed, for all $\y\in\Tr$ and $\omega\in D_\y$, $W_\omega$ is a root of the tree $(T_\omega,\supseteq)$, whence $W_\omega\supseteq \y(\omega)$.\smallskip

    (Ad $\X$ = set of moves):~ Let $\x\in\X$. Then there is $\omega\in D_\x$ and, by Axiom~\ref{def:sdf.X.section}, $\x(\omega)\in X$. Hence, there is $y\in F$ with $\x(\omega) \supsetneq y$. By Proposition~\ref{prop:ev_on_Tr_is_iso}, there is unique $\y\in\Tr$ with $\omega\in D_\y$ and $y = \y(\omega)$, and whence, by the same proposition, $\x >_\Tr \y$. Hence, $\x$ is a move in $(\Tr,\ge_\Tr)$. 

    On the other hand, the elements of $\Tr \setminus \X$, i.e.\ the random terminal nodes, are terminal nodes in $(\Tr,\ge_\Tr)$ directly by definition of $\ge_\Tr$. Hence, $\X$ is the set of moves in $(\Tr,\ge_\Tr)$.
\end{proof}

We continue with the verification of the presented examples of stochastic decision forests.

\begin{proof}[Proof of Lemma~\ref{lemma:simple_sdf1}]
    $(\Omega,\ms E)$ is a measurable space. $F$ is clearly a $W$-poset and the map assigning to any $w\in W$ the chain $\uparrow \{w\}$ is easily seen to be a bijection between $W$ and the set of maximal chains in $(F,\supseteq)$. Moreover, any chain in $(F,\supseteq)$ contains a maximal element, hence, by \cite[Theorem~3]{AlosFerrer2005}, expressed in the language of the present paper, $F$ a decision forest on $W$ (Axiom~\ref{def:sdf.df}).
   
    The connected components of $(F,\supseteq)$ are easily identified as 
    \[ \{\x_k(\omega) \mid k=0,1,2\} \cup \{\{(\omega,k,m)\} \mid k,m\in\{1,2\}\}, \]
    ranging over $\omega\in\Omega$. $\pi$ is well-defined because all nodes are non-empty subsets of $W$ and its definition does not depend on the choice of the element whose first entry is evaluated. 
    The connected component spelled out above equals $\pi^{-1}(\{\omega\})$ indeed. Hence, Axiom~\ref{def:sdf.conn_comp} is satisfied.

    For all $k=0,1,2$, $\x_k$ is defined on $\Omega$ which is an event. Clearly, $\x_k$ is a section of $\pi$. The union of the images of $\X$'s elements is the set of nodes with at least two elements. As $\{\{w\} \mid w\in W\} \subseteq F$, this union equals $X$. Thus Axioms~\ref{def:sdf.X.section} and \ref{def:sdf.X.cov} are satisfied.

    Furthermore, the relations between the images of all presumed random moves, put aside equalities, amount to
    \[ \forall \omega\in\Omega\forall k=1,2\colon \quad \x_0(\omega) \supseteq \x_k(\omega). \]
    Hence, put aside equalities, $\ge_\X$ is given by $\x_0 \ge_\X \x_k$, for $k=1,2$. Moreover, the roots of $F$ are given by the image of $\x_0$. It follows that $\X$ satisfies Axioms~\ref{def:sdf.X.OC} and \ref{def.sdf.X.surely_NT}.

    As the domains of all $\x\in\X$ are equal to the sample space $\Omega$, the maximality property of Axiom~\ref{def:sdf.X.max} is satisfied. 
\end{proof}

\begin{proof}[Proof of Lemma~\ref{lemma:simple_sdf2}]
    The proof is almost completely analogous to the preceding one, up to small modifications. We only comment on those. Thus, we omit the first two axioms.

    For all $k=0,1$, $\x_k$ is defined on $\Omega$ as in the previous lemma. For $k=2$, it has to be noted that $\{\omega_2\}\in\ms E$. Clearly, $\x'$ is a section of $\pi$ for all $\x'\in\X'$. The union of the images of the elements of $\X'$ is the set of nodes with at least two elements. Again, as $\{\{w'\} \mid w'\in W\} \subseteq F'$, this union equals $X'$. It follows that $\X'$ satisfies Axioms~\ref{def:sdf.X.section} and \ref{def:sdf.X.cov}.

    Furthermore, the relations between the images of all presumed random moves, put aside equalities, amount to
    \[ \forall \omega\in\Omega\colon \quad \x'_0(\omega) \supseteq \x'_1(\omega), \qquad \x'_0(\omega_2) \supseteq \x'_2(\omega_2).\]
    Hence, put aside equalities, $\ge_{\X'}$ is given by $\x'_0 \ge_{\X'} \x'_1$ and $\x'_0 \ge_{\X'} \x'_2$. The roots of $F'$ are given by the image of $\x'_0$. It follows that $\X'$ satisfied Axioms~\ref{def:sdf.X.OC} and \ref{def.sdf.X.surely_NT}.

    The maximality property of Axiom~\ref{def:sdf.X.max} requires only a little bit more justification this time. If $\bar\X'$ is such that $(F,\pi,\bar\X')$ is an \textsc{sdf} and $\X'$ refines $\bar\X'$, then any $\bar\x'\in\bar\X'$ with values in $\im\x'_1 \cup \im\x'_0$, we have $\bar\x'\in\{\x'_1,\x'_0\}$ because the domain of $\x'_1$ and $\x'_0$ is $\Omega$. Hence, there can and must be exactly one other element of $\bar\X'$, namely $\x'_2$. Hence, $\bar\X' = \X'$.
\end{proof}

\begin{proof}[Proof of Lemma~\ref{lemma:absent_minded_driver_Gilboa_sdf}]
    $(F,\pi,\X)$ is easily shown to be an \textsc{sdf}, similarly as above.

    $\rho$ is surjective by assumption. Let $\omega_k\in\rho^{-1}(\{k\})$, for $k=1,2$. Then for both of these $k$, $\x_k(\omega_k) \supsetneq \x_{3-k}(\omega_k)$, but, as a consequence, $\x_k\ngeq_\X \x_{3-k}$. Hence, $(F,\pi,\X)$ is not order consistent.
\end{proof}

\begin{proof}[Proof of Lemma~\ref{lemma:AP_sdf.AssmAP.SDF1}]
    (Ad ``$\Rightarrow$''):~ Suppose \hyperlink{Ass:AP.SDF1}{AP.SDF1} to hold and let $x\in F$ be not a singleton. Then, there are $(t,w)\in \T\times W$ with $x = x_t(w)$. Let $(t',w')\in\T\times W$ such that $x=x_{t'}(w')$. Hence, $w\in x_{t'}(w')$ and thus $x=x_{t'}(w)$. By Assumption~\hyperlink{Ass:AP.SDF1}{AP.SDF1}, $t=t'$. Hence, $\T_x = \{t\}$.\smallskip

    (Ad ``$\Leftarrow$'')~ Suppose the right-hand criterion to be satisfied. Let $w\in W$ and $t,u\in\T$ with $t\neq u$ and $x_t(w) = x_u(w)$. Hence, $\T_{x_t(w)}$ is not a singleton, and by our hypothesis $x_t(w)$ must be a singleton. As $w\in x_t(w)$, we get $x_t(w) = \{w\}$.
\end{proof}

\begin{proof}[Proof of Lemma~\ref{lemma:AP_sdf.AssmAP.SDF1_addon}]
    Let $x\in F$ such that $\T_x = \{t\}$ for some non-maximal $t\in\T$, and let $w\in x$. Then there is $u\in\T$ with $t<u$, whence $x = x_t(w) \supsetneq x_u(w)$. As $w\in x_u(w)$, $x$ contains at least two elements.
\end{proof}

\begin{proof}[Proof of Theorem~\ref{thm:AP_sdf}]
    (Ad Axiom~\ref{def:sdf.df}):~ As $F\subseteq \mc P (W)$, $F$ is a $W$-poset. We show that it is a ``bounded'' and ``irreducible'' ``$W$-set tree'', in the language of \cite{AlosFerrer2005}, and that any principal up-set in $(F,\supseteq)$ has a maximal element. From this, using \cite[Theorem~3]{AlosFerrer2005}, we obtain, expressed in the language of the present paper, that $F$ is decision forest on $W$.

    A ``$W$-set tree'' is a $W$-poset satisfying ``trivial intersection'' and ``separability'' (see \cite[Definition 3]{AlosFerrer2005}).    
    Regarding ``trivial intersection'' (defined in \cite[(4), p.\ 768]{AlosFerrer2005}), let $x,y\in F$ be such that $x\cap y \neq\emptyset$. Then there is $w\in W$ such that $w\in x \cap y$. Hence, $x$ and $y$ are, respectively, of the form $\{w\}$ or $x_t(w)$, for some $t\in\T$. In each of the four possible cases, we get $x\supseteq y$ or $y\supseteq x$. Hence, the ``trivial intersection'' property is satisfied.
    Regarding ``separability'' (defined in \cite[(8), p.\ 773]{AlosFerrer2005}), let $x,y\in F$ such that $x\supsetneq y$. Then there is $w\in y$, whence $w\in x$. $x$ has at least two elements, thus $x = x_t(w)$ for some $t\in\T$. Let $w'\in x\setminus y$. Then $x \supseteq \{w'\}$ and $y \cap \{w'\} = \emptyset$. Hence, the ``separability'' property is satisfied, and $F$ is a $W$-set tree.

    Regarding ``irreducibility'', let $w_0,w_1\in W$ with $w_0\neq w_1$. Then $\{w_k\}\in F$ satisfies $w_k\in\{w_k\} \setminus \{w_{1-k}\}$ for both $k=0,1$. Whence irreducibility. 
    
    Regarding ``boundedness'', let $c$ be a chain in $(F,\supseteq)$. It suffices to consider the case where $c\neq \emptyset$. First, consider the case that there is $w\in W$ with $\{w\} \in c$. In that case, any $x\in c$ is non-empty and hence the chain property of $c$ implies $w\in x$. Second, if there is no $w\in W$ with $\{w\} \in c$, then there must be a unique $\omega\in\Omega$ such that 
    \[ \T' = \{t'\in\T \mid \exists f_{t'} \in\A^\T\colon x_{t'}(\omega,f_{t'})\in c\} \]
    is non-empty. Uniqueness holds true because if $t_1,t_2\in\T$ and $(\omega_1,f_1), (\omega_2,f_2)\in W$ satisfy $x_{t_1}(\omega_1,f_1) \supseteq x_{t_2}(\omega_2,f_2)$, then, by definition, $\omega_1 = \omega_2$. 
    
    Let $\tilde \T = \{t\in\T \mid \exists t'\in\T'\colon t < t'\}$, a convex subset of $\T$. As $c$ is a chain and by definition of $\T'$, for all $t\in\tilde \T$, there is a unique $a(t)\in\A$ such that for all $t'\in\T'$ with $t < t'$ and $f_{t'}\in \A^\T$ satisfying $x_{t'}(\omega,f_{t'})\in c$ we have $f_{t'}(t) = a(t)$. Let $\tilde f\in\A^\T$ be such that $\tilde f(t) = a(t)$ for all $t\in \tilde \T$. Hence, all $t'\in\T'$ satisfy $x_{t'}(\omega,\tilde f)\in c\subseteq F$. By Assumption~\hyperlink{Ass:AP.SDF2}{AP.SDF2}, there is $f\in\A^\T$ such that $(\omega,f)\in W$ and $x_{t'}(\omega,f) = x_{t'}(\omega,\tilde f)\in c$ for all $t'\in\T'$.

    Hence, as $c$ contains no singleton, we get
    \[ (\ast) \qquad c = \{ x_{t'}(\omega,f) \mid t'\in\T'\}. \]
    Hence, for all $x\in c$, we have $(\omega,f)\in x$. The proof of ``boundedness'' is complete.

    It remains to be shown that any principal up-set contains a maximal element. The crucial point here is that for any $x\in F$, there is $\omega\in\Omega$ such that $x\subseteq (\{\omega\} \times \A^\T)\cap W$. Moreover, for any $w\in x$, $(\{\omega\} \times \A^\T)\cap W = x_0(w)$. Thus $x_0(w)\in \uparrow x$, and if $y\in\uparrow x_0(w)$, then $y = x_0(w)$. Thus, $x_0(w) \in\uparrow x$ is a maximal element for $(F,\supseteq)$.\smallskip

    (Ad Axiom~\ref{def:sdf.conn_comp}):~ Recall from the preceding argument that any $x\in F$ is contained in $(\{\omega\} \times \A^\T)\cap W$ for some $\omega\in\Omega$. Hence, $\pi(x)$ is well-defined and equal to this uniquely determined $\omega$. Let $x,y\in F$. Then $\pi(x) = \pi(y)$ iff there is $\omega\in\Omega$ such that $x,y \subseteq (\{\omega\} \times \A^\T)\cap W$. As $(\{\omega\} \times \A^\T)\cap W$ is a root of $(F,\supseteq)$, as seen just above, this is equivalent to the statement that $x,y$ belong to the same connected component.\smallskip

    (Ad Axiom~\ref{def:sdf.X.section}):~ By definition of the sets $D_{t,f}$, $(t,f)\in\T\times\A^\T$, and in view of the fact that $F$ contains all singletons in $\mc P(W)$, all elements of $\X$ map into the set $X$ of moves of the decision forest $F$ on $W$. Moreover, by Assumption~\hyperlink{Ass:AP.SDF0}{AP.SDF0}, $D_{t,f}\in\ms E$. Furthermore, any $\x\in\X$ is a section of $\pi$. For this represent $\x$ as $\x_t(f)$ for $(t,f)\in\T\times\A^\T$ such that $D_{t,f}\neq\emptyset$, and let $\omega\in D_{t,f}$. Then, 
    \[ \pi(\x_t(f)(\omega)) = \pi(x_t(\omega,f)) = \omega, \]
    because $\omega$ is the first component of $(\omega,f)\in x_t(\omega,f)$.\smallskip

    (Ad Axiom~\ref{def:sdf.X.cov}):~ Let $x\in X$. Then $x$ has at least two elements, because $F$ is a decision forest on $W$. Hence, there are $(\omega,f)\in W$ and $t\in\T$ such that $x = x_t(\omega,f)$. Hence, $\omega\in D_{t,f}$, and $x = \x_t(f)(\omega)$.\smallskip

    (Ad Axiom~\ref{def:sdf.X.OC}):~ Let $\x_1,\x_2\in\X$ such that there is $\omega\in D_{\x_1}\cap D_{\x_2}$ with $\x_1(\omega) \supseteq \x_2(\omega)$. Hence, there are $t_1,t_2\in\T$ and $f\in\A^\T$ such that $\omega \in D_{t_k,f}$ and $\x_k = \x_{t_k}(\omega,f)$, for both $k=1,2$. If we had $t_1 > t_2$, then, by construction of the nodes and by Assumption~\hyperlink{Ass:AP.SDF1}{AP.SDF1}, $\x_1(\omega) = x_{t_1}(\omega,f) \subsetneq x_{t_2}(\omega,f) = \x_2(\omega)$, because they contain at least two elements as moves -- a contradiction. Hence, $t_1\le t_2$.

    Let $\omega'\in D_{t_2,f}$. By definition, $x_{t_2}(\omega',f)$ is a move. As $t_1 \le t_2$, $x_{t_1}(\omega',f) \supseteq x_{t_2}(\omega',f)$, whence $\omega'\in D_{t_1,f}$. In particular, $\x_1(f)(\omega') \supseteq \x_2(f)(\omega')$. Thus, $\x_1 \ge_\X \x_2$.\smallskip

    (Ad Axiom~\ref{def:sdf.X.max}):~ Let $\bar\X$ be a set such that $(F,\pi,\bar\X)$ is an order consistent \textsc{sdf} and that is refined by $\X$. By definition of the latter, for any $\bar\x\in\bar\X$ there is $P_{\bar\x}\subseteq\X$ such that $\bar\x = \bigcup P_{\bar\x}$. 
    As $\X$ and $\bar\X$ induce partitions of $X$, by Proposition~\ref{prop:ev_on_Tr_is_iso}, and as $D_\x\neq \emptyset$ for all $\x\in\X$ by construction of $\X$, we infer that there is a unique map $b\colon\X \to \bar\X, \x\mapsto b(\x)$ such that $D_\x\subseteq D_{b(\x)}$ and $b(\x)|_{D_\x} = \x$, for all $\x\in\X$. 
    For each pair $(t,f)\in\T\times\A^\T$ with $D_{t,f}\neq\emptyset$, denote $\bar \x_t(f) = b({\x_t(f)})$ the image of $\x_t(f)$ under this map.

    Our aim is to show that for all $\bar\x\in\bar\X$, we have $P_{\bar\x} = \{\bar\x\}$. If this were not the case, there would be $\bar\x\in\bar\X$ such that $P_{\bar\x}$ has at least two elements. Hence, there would be $t,u\in\T$ and $f,g\in\A^\T$ with $D_{t,f},D_{u,g}\neq\emptyset$, $\x_t(f) \neq \x_u(g)$, and $\x_t(f),\x_u(g)\in P_{\bar\x}$. This would imply 
    \[ (\dagger)\qquad D_{t,f} \cap D_{u,g} = \emptyset. \] 
    Indeed, otherwise there would be $\omega \in D_{t,f}\cap D_{u,g}$ for that we would obtain $x_t(\omega,f) = \bar\x(\omega) = x_u(\omega,g)$. Hence, for $v = t\wedge u$, we would have $f|_{[0,v)_\T} = g|_{[0,v)_\T}$ and there would be $w\in W$ with $x_t(\omega,f) = x_t(w) = x_u(w) = x_u(\omega,g)$. As these are moves, Assumption~\hyperlink{Ass:AP.SDF1}{AP.SDF1} would imply $t=u=v$, whence $\x_t(f) = \x_u(g)$ which is absurd. Hence, we would necessarily have $D_{t,f} \cap D_{u,g} = \emptyset$.
    
    We can assume, without loss of generality, that $t\le u$. 
    There would be two cases. First, $D_{t,f} \cap D_{t,g} \neq \emptyset$. Then, $f|_{[0,t)_\T} \neq g|_{[0,t)_\T}$, because else $D_{t,f} = D_{t,g}$ whence $D_{u,g} \subseteq D_{t,g} = D_{t,f}$ which contradicts $(\dagger)$. The second case would be $D_{t,f} \cap D_{t,g} = \emptyset$. Then, by Assumption~\hyperlink{Ass:AP.SDF3}{AP.SDF3}, there would be $v\in [0,t)_\T$ such that
    \[ (\circ) \qquad D_{v,f} \cap D_{v,g} \neq \emptyset, \quad \text{and} \quad f|_{[0,v)_\T} \neq g|_{[0,v)_\T}. \]
    We conclude that in both cases, there would be $v\in [0,t]_\T$ satisfying $(\circ)$.

    Let $\omega'\in D_{v,f} \cap D_{v,g}$. Then, $x_v(\omega',f)$ and $x_v(\omega',g)$ would not be comparable in $(F,\supseteq)$ since this would imply $f|_{[0,v)_\T} = g|_{[0,v)_\T}$. Hence, by the definition of $\ge_{\bar\X}$, the same would be true for $\bar\x_v(f)$ and $\bar\x_v(g)$ in $(\bar\X,\ge_{\bar\X})$. But, on the other hand, $\x_v(f) \ge_\X \x_t(f)$ and $\x_v(g) \ge_\X \x_u(g)$, whence by definition of $\ge_\X$ and Axiom~\ref{def:sdf.X.OC} applied to $\bar\X$, we get $\bar\x_v(f),\bar\x_v(g) \in \uparrow \bar \x_t(f)$. As $(\bar\X,\ge_{\bar\X})$ is a forest (which follows from Proposition~\ref{prop:ev_on_Tr_is_iso} as in the proof of Theorem~\ref{thm:Xrm_is_dec_tree}), $\bar\x_v(f)$ and $\bar\x_v(g)$ would have to be comparable, a contradiction. Hence, the initial assumption was false, and we infer that $P_{\bar\x} = \{\bar\x\}$ for all $\bar\x\in\bar\X$. 
    
    This in turn implies that $\x = b(\x)$ for all $\x\in\X$. Hence, we have both $\bar\X\subseteq \X$ and $\X\subseteq\bar\X$, whence $\X = \bar \X$.\smallskip

    (Ad statement about sure non-triviality):~ For all $w\in W$, we have $x_0(w) \supseteq \{w\}$ and $\{w\}\in F$. Hence, for all $\omega\in\Omega$ and $f\in\A^\T$, $x_0(\omega,f)\in X$ iff $\omega\in D_{0,f}$. It is moreover obvious from the definitions that all $\omega\in\Omega$ and $f\in\A^\T$ satisfy $x_0(\omega,f) = W_\omega$.
    Hence, for any $f\in\A^\T$, 
    \[ \{\omega\in\Omega \mid W_\omega\in X\} = D_{0,f} = \big \{\omega\in\Omega \mid \exists f',f''\in\A^\T\colon [f'\neq f'' \text{ and } (\omega,f'),(\omega,f'')\in W]\big \}, \] 
    from which the claimed equivalence directly follows in view of $(\ast)$.   
\end{proof}

\begin{proof}[Proof of Lemma~\ref{lemma:mf_t}]
    Let $x,y\in X$ with $x \supsetneq y$. Then, by definition of $\mf t$ and the construction of nodes in action path \textsc{sdf}, there is $w\in W$ such that $x=x_{\mf t(x)}(w)$ and $y=x_{\mf t(y)}(w)$. Indeed, $y$ is non-empty, and any $w\in y$ satisfies the preceding equalities. If we had $\mf t(x)\ge \mf t(y)$, then $y = x_{\mf t(y)}(w) \supseteq x_{\mf t(x)}(w) = x$, in contradiction to the choice of $x$ and $y$. Hence, by totality of the order on $\T$, $\mf t(x)<\mf t(y)$. 

    Let $\x\in\X$. Then there are $t\in\T$ and $f\in\A^\T$ such that $D_{t,f} \neq \emptyset$ with $\x = \x_t(f)$. Hence, for all $\omega\in D_\x = D_{t,f}$:
    \[\mf t(\x(\omega)) = \mf t(x_t(\omega,f)) = t, \]
    which does not depend on $\omega$.
\end{proof}

\begin{proof}[Proof of Lemma~\ref{lemma:simple_sdf_as_APsdf}]
    Take $\A = \{1,2\}$ and $\T = \{0,1\}$. Recall that $\Omega = \{\omega_1,\omega_2\}$, with $\omega_1 \neq \omega_2$, and $\ms E = \mc P(\Omega)$. 
    
    For the first example, we may set $W = \Omega \times \A^\T$. Any $(\omega,f)\in W$ can be identified with the triple $(\omega,f(0),f(1))$, so that $W$ corresponds to the set denoted by $W$ in the example in Subsection~\ref{subs:simple_sdf}. Under this identification, it is easily verified that $D_{t,f} = \Omega$ for all $f\in\A^\T$ and $t\in\T$, $\x_0 = \x_0(f)$ for all $f\in\A^\T$, and $\x_k = \x_1(f)$ for all $f\in\A^\T$ with $f(0) = k$, for both $k=1,2$. Assumptions \hyperlink{Ass:AP.SDF0}{AP.SDF0} to \hyperlink{Ass:AP.SDF3}{AP.SDF3} are readily verified. Under this identification, the action path construction yields exactly the objects $F$, $\pi$, and $\X$ from the example in Subsection~\ref{subs:simple_sdf}. 

    For the second example, let $W'$ as in Subsection~\ref{subs:simple_sdf}. Interpret every triple $(\omega,k,m)\in W'$ as the pair $(\omega,f)$ where $f\in\A^\T$ is given by $f(0) = k$ and $f(1) = m$. Interpret the pair $(\omega_1,2)$ as the map $f\colon\T\to\{0,1,2\}$ with $f(0) = 2$ and $f(1) = 0$. Then, $W \subseteq {\A'}^\T$ with $\A' = \A \cup \{0\}$, where $0$ is a placeholder for inaction. Furthermore, $D_{t,f} = \Omega$ for all $(t,f)\in\T\times {\A'}^\T$ with $t=0$, or $t=1$ and $f(0) = 1$; and $D_{1,f} = \{\omega_2\}$ for $f\colon\T\to \A'$ with $f(0) = 2$; for all other pairs $(t,f) \in \T\times{\A'}^\T$ we have $D_{t,f} = \emptyset$. Moreover, $\x'_0 = \x_0(f)$ for all $f\in{\A'}^\T$; and $\x'_1 = \x_1(f)$ for all $f\colon\T\to\A'$ with $f(0) = 1$; and $\x'_2 = \x_1(f)$ for $f\colon\T\to \A'$ with $f(0) = 2$. Assumptions \hyperlink{Ass:AP.SDF0}{AP.SDF0} to \hyperlink{Ass:AP.SDF3}{AP.SDF3} are readily verified. Under these identifications, the action path construction yields exactly the objects $F'$, $\pi'$, and $\X'$ from the example in Subsection~\ref{subs:simple_sdf}. 
\end{proof}

\begin{proof}[Proof of the claims in Example~\ref{ex:APsdf}]
    (The case $W = \Omega \times \A^\T$):~ We show that for $W = \Omega\times\A^\T$, $W$ induces an action path \textsc{sdf}. First, we note that for any $\omega\in\Omega$, there is $f\in\A^\T$ such that $(\omega,f)\in W$. 
    
    Clearly, for all $(t,f)\in\T\times\A^\T$, $D_{t,f} = \Omega$, whence \hyperlink{Ass:AP.SDF0}{AP.SDF0}. 
    
    Regarding \hyperlink{Ass:AP.SDF1}{AP.SDF1}, let $w=(\omega,f)\in W$ and $t,u\in\T$ with $t\neq u$ such that $x_t(w) = x_u(w)$. Without loss of generality, assume $t<u$. If $\A$ is a singleton, then $\A^\T$ is a singleton as well. Hence, $x_t(w) = \{w\}$. But $\A$ must be a singleton, since else there would be $g\in\A^\T$ with $(\omega,g) \in x_t(w)$ and $g(t) \neq f(t)$. Thus, $(\omega,g)\in x_t(w)\setminus x_u(w)$ -- a contradiction. 

    Regarding \hyperlink{Ass:AP.SDF2}{AP.SDF2}, let $\omega\in\Omega$, $\tilde f\in\A^\T$, and $\T'\subseteq\T$ satisfying $x_t(\omega,\tilde f)\in F$ for all $t\in\T'$. Then we already have $(\omega,\tilde f)\in W$.

    \hyperlink{Ass:AP.SDF3}{AP.SDF3} is clearly satisfied because $D_{t,f} = \Omega$ for all $(t,f)\in\T\times\A^\T$.\smallskip

    (The timing problem):~ Clearly, for any $\omega\in\Omega$ there is $f\in\A^\T$ with $(\omega,f)\in W$, e.g.\ the constant map taking only the value $1$ at all times, in all components. 

    Regarding Assumption~\hyperlink{Ass:AP.SDF0}{AP.SDF0}, for all $(t,f)\in\A^\T$, $D_{t,f} = \emptyset$ if $f$ is not decreasing on $[0,t)_\T$ or if $f(t-) = 0$. Else, $D_{t,f} = \Omega$. Hence, \hyperlink{Ass:AP.SDF0}{AP.SDF0} is satisfied.

    Regarding Assumption~\hyperlink{Ass:AP.SDF1}{AP.SDF1}, let $w=(\omega,f)\in W$ and $t,u\in\T$ with $t<u$ such that $x_t(w) = x_u(w)$. Hence, for all decreasing $g\colon \T\to\A$ with $g|_{[0,t)_\T} = f|_{[0,t)_\T}$ we must have $f(v) = g(v)$ for all $v\in [t,u)_\T$. Thus, we must have $f(t-) = 0$. Hence, $x_t(w) = \{(\omega,f)\}$.

    Regarding Assumption~\hyperlink{Ass:AP.SDF2}{AP.SDF2}, let $\omega\in\Omega$, $\tilde f\in \A^\T$, and $\T'\subseteq\T$ such that $x_t(\omega,\tilde f)\in F$ for all $t\in \T'$. Hence, $\tilde f$ is decreasing on $\bigcup_{t\in\T'} [0,t)_\T$. Hence, there is decreasing $f\colon \T\to \A$ such that $f|_{[0,t)_\T} = \tilde f|_{[0,t)_\T}$ for all $t\in\T'$. By construction, $(\omega,f)\in W$.

    Regarding Assumption~\hyperlink{Ass:AP.SDF3}{AP.SDF3}, it is sufficient to note that $D_{t,f}$ equals $\emptyset$ or $\Omega$, for all $(t,f)\in\T\times\A^\T$ which has been established earlier.\smallskip

    (Ad American up-and-out option):~ Let $\A = \{0,1\}$ and $\T=\R_+$. Clearly, for any $\omega\in\Omega$, there is $f\colon\R_+\to\A$ with $(\omega,f)\in W$, namely, the constant path with value $1$.

    Regarding Assumption~\hyperlink{Ass:AP.SDF0}{AP.SDF0}, let $(t,f)\in\T\times\A^{\R_+}$ and $\omega\in\Omega$. Then, we have $\omega\in D_{t,f}$ iff $x_t(\omega,f)$ has at least two elements. This is equivalent to the fact that $f$ is decreasing on $[0,t)$, $f(t-) = 1$, and $P_u(\omega) < 2$ for all $u\in[0,t]$. Whence the claimed representation of $D_{t,f}$ and the fact that $D_{t,f}\in\ms E$, because $P$ has continuous paths so that $\max_{u\in[0,t]} P_u$ is $\ms E$-measurable since $P_u$ is so for all real (and in particular rational) $u\ge 0$.

    Regarding Assumption~\hyperlink{Ass:AP.SDF1}{AP.SDF1}, let $w=(\omega,f)\in W$ and $t,u\in\R_+$ with $t<u$ such that $x_t(w) = x_u(w)$. 
    Hence, for all decreasing $g\colon \R_+\to\{0,1\}$ with $g|_{[0,t)} = f|_{[0,t)}$ and $(\omega,g)\in W$, we must have 
    \[ (\ast) \qquad f(v) = g(v), \quad \text{ for all }~v\in [t,u). \] 
    
    If $f$ is constant to $1$, this implies the existence of $t_0\in[0,t]$ such that $P_{t_0}(\omega) \ge 2$. If such $t_0$ did not exist, the map $g\colon \R_+ \to \{0,1\}$ given by $g(t') = 1\{t' < t\}$ would violate condition $(\ast)$ above. Hence, $x_t(w) = \{w\}$. 
    If $f$ takes the value $0$ and $t_f^\ast < t$, then clearly $x_t(w) = \{w\}$. 
    The remaining case, namely that $f$ takes the value $0$ and $t_f^\ast \ge t$, cannot arise. Indeed, if it did, for both $a\in\{0,1\}$, the map $g_a\colon\R_+ \to \{0,1\}$ given by $g_a(t') = 1$ for $t'<t$, $g_a(t) = a$, and $g_a(t') = 0$ for $t'> t$ would satisfy $t_{g_a}^\ast = t$. Moreover, the fact that $(\omega,f)\in W$ would imply
    \[ \max_{t'\in[0,t]} P_{t'}(\omega) \le \max_{t'\in[0,t_f^\ast]} P_{t'}(\omega) < 2, \]
    hence $(\omega,g_a)\in W$. In particular, $(\ast)$ would imply that $a = g_a(t) = f(t)$, for both $a\in\{0,1\}$ which is impossible. Hence, in any possible case we have $x_t(w) = \{w\}$, thus Assumption~\hyperlink{Ass:AP.SDF1}{AP.SDF1} is satisfied.

    Regarding Assumption~\hyperlink{Ass:AP.SDF2}{AP.SDF2}, let $\omega\in\Omega$, $\tilde f\in \A^{\R_+}$, and $\T'\subseteq\R_+$ such that $x_t(\omega,\tilde f)\in F$ for all $t\in \T'$. Hence, $\tilde f$ is decreasing on $[0,\sup \T')$ (where, in the context of the order on $\R_+$, we have $\sup \emptyset = 0$), and if $\tilde f|_{[0,\sup \T')}$ attains $0$, then $P_u(\omega) < 2$ for all $u\in[0,t_{\tilde f}^\ast]$. Let $f\colon \R_+ \to\{0,1\}$ be given by $f|_{[0,\sup \T')} = \tilde f|_{[0,\sup \T')}$ and, for all $u\in [\sup\T',\infty)$, $f(u) = f(\sup\T'-)$ if $\sup\T' > 0$ and $f(u) = 1$ if $\sup\T' = 0$. By construction, $f|_{[0,t)} = \tilde f|_{[0,t)}$ for all $t\in\T'$. Further, $f$ is decreasing. Moreover, if $f$ attains the value $0$, then it does so on $[0,\sup\T')$ and so does $\tilde f$, and $t_f^\ast = t_{\tilde f}^\ast$. Hence, $(\omega,f)\in W$. We conclude that Assumption~\hyperlink{Ass:AP.SDF2}{AP.SDF2} is satisfied.

    Regarding Assumption~\hyperlink{Ass:AP.SDF3}{AP.SDF3}, let $t\in\R_+$ and $f,g\in\A^\T$ such that $D_{t,f},D_{t,g} \neq\emptyset$. Then, as shown earlier in this proof, $f(t-) = g(t-) = 1$ and
    \[ D_{t,f} = \{\omega\in\Omega \mid \max_{u\in[0,t]} P_u(\omega) < 2\} = D_{t,g}. \]
    Hence, $D_{t,f} \cap D_{t,g} \neq \emptyset$. Thus, Assumption~\hyperlink{Ass:AP.SDF3}{AP.SDF3} is trivially satisfied.
\end{proof}

\subsection{Section~\ref{sec:exogenous_information}}

\begin{proof}[Proof of Lemma~\ref{lemma:simple_sdf1_EIS}]
    There are exactly two ($\sigma$-)algebras on $\Omega$: the discrete and the trivial one, that is, $\mc P(\Omega)$ and $\{\Omega,\emptyset\}$. In the present situation, we have $D_\x = \Omega$ for all $\x\in\X$. Hence, by definition, the set of exogenous information structures $\ms F$ on $(F,\pi,\X)$ admitting recall is given by all families $\ms F = (\ms F_\x)_{\x\in\X}$ of $\sigma$-algebras on $\Omega$ such that $\ms F_{\x_0} \subseteq \ms F_{\x_1} \cap \ms F_{\x_2}$. The claim follows easily from this.
\end{proof}

\begin{proof}[Proof of Lemma~\ref{lemma:simple_sdf2_EIS}]
    There is exactly one ($\sigma$-)algebra on $\{\omega_2\}$, namely $\{\{\omega_2\},\emptyset\}$. On $\Omega$, there are exactly two $\sigma$-algebras, namely $\mc P(\Omega)$ and $\{\Omega,\emptyset\}$, as in the preceding proof. Hence, by definition, the set of exogenous information structures $\ms F'$ on $(F',\pi',\X')$ admitting recall is given by all families $\ms F' = (\ms F'_{\x'})_{\x'\in\X'}$ of $\sigma$-algebras on $\Omega$ such that $\ms F_{\x'_2} =\{\{\omega_2\},\emptyset\}$ and $\ms F_{\x'_0} \subseteq \ms F_{\x'_1}$. The claim follows easily from this.
\end{proof}

\begin{proof}[Proof of Lemma~\ref{lemma:APsdf_EIS_induces_filtration}]
    Suppose that the exogenous information structure $\ms F$ on $\X$ admits recall and let $f\in\A^\T$. Let $t,u\in\T_f$ with $t<u$ and $E\in\ms F_{\x_t(f)}$. Then, $\x_t(f) \ge_\X \x_u(f)$, whence $E\cap D_{u,f} \in \ms F_{\x_u(f)}$ by Definition~\ref{def:EIS}.

    The second claim follows from the first because under its hypothesis $\ms F_{\x_t(f)}$ is a $\sigma$-algebra on $D_{t,f} = \Omega$ for all $t\in\T_f$.
\end{proof}

\begin{proof}[Proof of Theorem~\ref{thm:AP_sdf_EIS}]
    Let $\x\in\X$. Then, by construction, $\ms F_\x$ is a $\sigma$-algebra on $D_\x$ contained in $\ms E$, because $\ms G_{\mf t(x)}$ is a sub-$\sigma$-algebra of $\ms E$ and all $Y_{\x'}$, ranging over $\x'\in\X$, are $\ms E$-measurable.

    Furthermore, let $\x_1,\x_2\in\X$ such that $\x_1\ge_\X\x_2$, and let $E\in \ms F_{\x_1}$. Hence, there is $E'\in \sigma(Y_{\x'} \mid \x'\ge_\X \x_1) \vee \ms G_{\mf t(\x_1)}$ such that $E = E' \cap D_{\x_1}$. Note that
    \[ E'\in \sigma(Y_{\x'} \mid \x'\ge_\X \x_1) \vee \ms G_{\mf t(\x_1)} \subseteq \sigma(Y_{\x'} \mid \x'\ge_\X \x_2) \vee \ms G_{\mf t(\x_2)}. \]
    As $D_{\x_1} \supseteq D_{\x_2}$, we have $ E\cap D_{\x_2} = E' \cap D_{\x_2} $. Hence, $E \cap D_{\x_2} \in \ms F_{\x_2}$.
\end{proof}

\subsection{Section~\ref{sec:adapted_choices}}

\begin{proof}[Proof of Lemma~\ref{lemma:P(c)_compatible_with_conn_comp}]
    As $c\cap W_E = \bigcup_{\omega\in E} (c\cap W_\omega)$ and $\{T_\omega\mid\omega\in\Omega\}$ is a partition of $F$, we infer from Lemma~\ref{lemma:set_forest}, Part~\ref{lemma:set_forest.T=downarrow_V_T}, that the following introductory statement holds true:
    \[ \downarrow (c\cap W_E) = \{x\in F \mid c \cap W_E \supseteq x\} = \{x\in F_E \mid c \supseteq x\} = (\downarrow c) \cap F_E. \]
    
    From this, the claim follows easily: If $x\in P(c\cap W_E)$, then there is $y\in (\downarrow c)\cap F_E$ such that
    \[ (\dagger) \qquad \uparrow x = \uparrow y \setminus ((\downarrow c) \cap F_E). \]
    Hence, $x \in \uparrow y$. There is unique $\omega\in E$ such that $y\in (\downarrow c)\cap T_\omega$. As $T_\omega$ is a connected component of $(F,\supseteq)$ and $y\in T_\omega$, we have $\uparrow y\subseteq T_\omega$, whence $x\in T_\omega$ and
    \[ (\ast) \qquad \uparrow x = \uparrow y \setminus \downarrow c. \]
    If conversely $x\in F_E$ is such that there is $y\in \downarrow c$ satisfying $(\ast)$, then $x\in \uparrow y$, hence $x,y$ belong to the same connected component, whence $y\in F_E$ and $\uparrow y \subseteq F_E$. We infer that $(\dagger)$ holds true, i.e.\ $x\in P(c\cap W_E)$ by the introductory statement.
\end{proof}

\begin{lemma}\label{lemma:simple_sdf1_choices}
    Consider the basic version $(F,\pi,\X)$ of the simple \textsc{sdf} on the exogenous scenario space $(\Omega,\ms E)$ as in Subsection~\ref{subs:simple_sdf_AC}. Let $M$ be the set of maps $\Omega\to \{1,2\}$. 
    Then the following subsets of $W$ define non-redundant and complete choices:
    \begin{itemize}[label=--]
        \item $c_{f\bullet}$, where $f\in M$;
        \item $c_{kg}$, where $k=1,2$ and $g\in M$;
        \item $c_{\bullet g}$, where $g\in M$.
    \end{itemize}
    The corresponding sets of immediate predecessors are given by $P(c_{f\bullet}) = \im\x_0$;  $P(c_{kg}) = \im \x_k$; $P(c_{\bullet g}) = \im \x_1 \cup \im \x_2$.
\end{lemma}

\begin{proof}
    Any set $c$ of the form above is non-empty, and as $F$ contains all singletons in $\mc P(W)$, $c$ is a choice.
    The sets of immediate predecessors are easily shown to be of the claimed form, using Lemma~\ref{lemma:P(c)_compatible_with_conn_comp}. From this, we directly infer the non-redundancy and completeness of all the considered choices.
\end{proof}

\begin{proof}[Proof of Lemma~\ref{lemma:simple_sdf1_LCS}]
    Let $C = \bigcup_{\x\in\X} \ms C_\x$. By Lemma~\ref{lemma:simple_sdf1_choices}, all elements of $C$ are non-redundant and complete choices. The statement about the set of immediate predecessors of first $c_{k \bullet}$ and second $c_{\bullet m}$, $k,m\in\{1,2\}$, stated in Lemma~\ref{lemma:simple_sdf1_choices} shows that the first type of choice is available at $\x_0$, while the second is available at $\x_1$ and $\x_2$.  
\end{proof}

\begin{proof}[Proof of Lemma~\ref{lemma:simple_sdf1_AC}]
    In view of Lemma~\ref{lemma:simple_sdf1_choices}, we only have to verify that for any line and corresponding exogenous information structure (\textsc{eis}) $\ms F$, and any subset $c\subseteq W$ in that line, all $\x\in\X$ that $c$ is available at and $c_\x\in\ms C_\x$, we have
    \[ \x^{-1}(P(c\cap c_\x)) \in \ms F_\x. \]
    For this, one easily verifies that all $k,k',m'\in\{1,2\}$ and $f,g\in M$ satisfy:
    \begin{align*}
        c_{f\bullet} \cap c_{k'\bullet} =&~ c_{k'\bullet} \cap W_{\{f = k'\}}; \\
        c_{kg} \cap c_{\bullet m'} =&~ c_{k m'} \cap W_{\{g = m'\}}; \\
        c_{\bullet g} \cap c_{\bullet m'} =&~  c_{\bullet m'} \cap W_{\{g = m'\}}.
    \end{align*}
    As shown in Lemma~\ref{lemma:simple_sdf1_choices}, we have
    \[ P(c_{k'\bullet} ) = \im \x_0, \quad P(c_{km'} ) = \im \x_k, \quad P(c_{\bullet m'} ) = \im \x_1 \cup \im \x_2.\]
    Then, applying Lemma~\ref{lemma:P(c)_compatible_with_conn_comp} and using the definition of $\ms F$, in each of the five cases respectively, completes the proof.
\end{proof}

\begin{lemma}\label{lemma:simple_sdf2_choices}
    Consider the variant $(F',\pi',\X')$ of the simple \textsc{sdf} on the exogenous scenario space $(\Omega,\ms E)$ as in Subsection~\ref{subs:simple_sdf_AC}. Let $M$ be the set of maps $\Omega\to \{1,2\}$. 
    Then the following subsets of $W'$ define non-redundant and complete choices:
    \begin{itemize}[label=--]
        \item $c'_{f\bullet}$, where $f\in M$;
        \item $c'_{kg}$, where $k=1,2$ and $g\in M$;
        \item $c'_{\bullet g}$, where $g\in M$.
    \end{itemize}
    The corresponding sets of immediate predecessors are given by $P(c'_{f\bullet}) = \im\x'_0$; $P(c'_{kg}) = \im\x'_k$; $P(c'_{\bullet g}) = \im\x'_1\cup \im\x'_2$.
\end{lemma}

\begin{proof}
    Any set $c'$ of the form above is non-empty, and as $F'$ contains all singletons in $\mc P(W')$, any such $c'$ is a choice.
    The sets of immediate predecessors are easily shown to be of the claimed form, using Lemma~\ref{lemma:P(c)_compatible_with_conn_comp}. From this, we directly infer the non-redundancy and completeness of all the considered choices.
\end{proof}

\begin{proof}[Proof of Lemma~\ref{lemma:simple_sdf2_LCS}]
    Let $C' = \bigcup_{\x'\in\X'} \ms C'_{\x'}$. By Lemma~\ref{lemma:simple_sdf2_choices}, all elements of $C'$ are non-redundant and complete choices. The statement about the set of immediate predecessors of first $c'_{k \bullet}$ and second $c'_{\bullet m}$, $k,m\in\{1,2\}$, stated in Lemma~\ref{lemma:simple_sdf2_choices} shows that the first type of choice is available at $\x'_0$, while the second is available at $\x'_1$ and $\x'_2$.  
\end{proof}

\begin{proof}[Proof of Lemma~\ref{lemma:simple_sdf2_AC}]
    In view of Lemma~\ref{lemma:simple_sdf2_choices}, we only have to verify that for any line and corresponding exogenous information structure (\textsc{eis}) $\ms F'$, and any subset $c'\subseteq W$ in that line, all $\x'\in\X'$ that $c'$ is available at and $c'_{\x'}\in\ms C'_{\x'}$, we have
    \[ \x'^{-1}(P(c'\cap c'_{\x'})) \in \ms F'_{\x'}. \]
    For this, one easily verifies that all $k,k',m'\in\{1,2\}$ and $f,g\in M$ satisfy:
    \begin{align*}
        c'_{f\bullet} \cap c'_{k'\bullet} =&~ c'_{k'\bullet} \cap W_{\{f = k'\}}; \\
        c'_{kg} \cap c'_{\bullet m'} =&~ c'_{k m'} \cap W_{\{g = m'\}}; \\
        c'_{\bullet g} \cap c'_{\bullet m'} =&~  c'_{\bullet m'} \cap W_{\{g = m'\}}.
    \end{align*}
    As shown in Lemma~\ref{lemma:simple_sdf2_choices}, we have
    \[ P(c'_{k'\bullet} ) = \im \x'_0, \quad P(c'_{km'} ) = \im \x'_k, \quad P(c'_{\bullet m'} ) = \im \x'_1 \cup \im \x'_2.\]
    Then applying Lemma~\ref{lemma:P(c)_compatible_with_conn_comp} and using the definition of $\ms F'$, in each of the three cases respectively, completes the proof.
\end{proof}

\begin{proof}[Proof of Lemma~\ref{lemma:AP_downarrow_c}]
    Let $t\in\T$ and $c\in\ms C_t$. Let $x\in F$. 
    
    If $x = \{w\}$ for some $w\in W$, then clearly $x\in \downarrow c$ iff $w\in c$.
    It remains to consider the case $x = x_u(w)$ for some $w=(\omega,f)\in W$ and $u\in\T$. 
    
    If $t<u$ and $w\in c$, then any $w' = (\omega',f')\in x$ satisfies $\omega' = \omega$ and $f'|_{[0,u)_\T} = f|_{[0,u)_\T}$, in particular $f'|_{[0,t]_\T} = f|_{[0,t]_\T}$, hence $w'\in c$. We conclude for this case that $x\subseteq c$. 

    If $t<u$ and $w\in c$ do not both hold true, then either $w\notin c$, or $w\in c$ and $u\le t$. If $w\notin c$, then $w\in x \setminus c$, hence $x\nsubseteq c$. If $w\in c$ and $u\le t$, then, by Assumption~\hyperlink{Ass:AP.C1}{AP.C1}, there is $w' \in x_t(w) \setminus c \subseteq x_u(w) \setminus c = x\setminus c$. Whence $x\nsubseteq c$.
\end{proof}

\begin{lemma}\label{lemma:AP_uparrow_x}
    Consider action path \textsc{sdf} data $(I,\A,\T,W)$ on an exogenous scenario space $(\Omega,\ms E)$ and the induced action path \textsc{sdf} $(F,\pi,\X)$. Let $t\in\T$ and $w\in W$. Then,
    \[ \uparrow x_t(w) = \{ x_u(w) \mid u\in \T\colon u\le t\}. \]
\end{lemma}

\begin{proof}
    The inclusion $\supseteq$ is clear from the definition of the nodes of the action path \textsc{sdf}. For the converse inclusion, let $x\in\uparrow x_t(w)$. If $x$ is a singleton, then $x = x_t(w)$. If $x$ is no singleton, then it is a move and there are $u\in \T$ and $w'\in W$ with $x =  x_u(w')$. As $w\in x_t(w) \subseteq x$, we even have $x = x_u(w)$. By means of Lemma~\ref{lemma:mf_t}, we infer $u = \mf t(x) \le \mf t(x_t(w)) = t$.
\end{proof}

\begin{proof}[Proof of Lemma~\ref{lemma:AP_P(c)}]
    Let $t\in\T$ and $c\in\ms C_t$. Let $x\in F$. By definition, $x\in P(c)$ is equivalent to the existence of $y\in \downarrow c$ satisfying
    \[ (\ast)\qquad \uparrow x = \uparrow y \setminus \downarrow c. \]

    If $x = x_t(w)$ for $w\in c$, then, by Assumption~\hyperlink{Ass:AP.C1}{AP.C1}, there is $w'\in x_t(w) \setminus c$. Thus, $x$ has at least two elements. In view of Lemmata \ref{lemma:AP_downarrow_c} and \ref{lemma:AP_uparrow_x} and Assumption~\hyperlink{Ass:AP.SDF1}{AP.SDF1}, $(\ast)$ is satisfied for $y = \{w\}$.

    If, conversely, $(\ast)$ is satisfied for some $y\in \downarrow c$, then $x$ is a move, whence the existence of $w_0\in W$ and $t_0\in\T$ with $x = x_{t_0}(w_0)$. Moreover, there is $w\in y$, and as $y\in \downarrow c$, we obtain $w\in c$. By $(\ast)$, we know that $x\in\uparrow y$, hence also $w\in x$. Thus, $x = x_{t_0}(w)$. From representation $(\ast)$ and Lemma~\ref{lemma:AP_downarrow_c}, we immediately get $t_0 \le t$. Moreover, by Assumption~\hyperlink{Ass:AP.C1}{AP.C1} on $c\in\ms C_t$, $x_t(w)$ has at least two elements. Indeed, as $w\in c$, there is $w'\in x_t(w)\setminus c$. Hence, by Lemma~\ref{lemma:AP_sdf.AssmAP.SDF1}, $\T_{\x_t(w)} = \{t\}$. Hence, $x_t(w) \in \uparrow y \setminus \downarrow c = \uparrow x$, again by representation $(\ast)$, the fact that $y\in\downarrow c$ and Lemma~\ref{lemma:AP_downarrow_c}. Thus, by Lemma~\ref{lemma:mf_t}, $t = \mf t(x_t(w)) \le \mf t(x) = t_0$. Hence, $t_0 = t$ and $x = x_t(w)$ with $w\in c$.\smallskip
\end{proof}

\begin{proof}[Proof of Lemma~\ref{lemma:C_t_non-redundant_complete}]
    By Assumption~\hyperlink{Ass:AP.C0}{AP.C0}, $c$ is a non-empty union of singletons in $\mc P(W)$, which are elements of $F$ by construction. Hence, it is a choice.

    Concerning non-redundancy, let $\omega\in\Omega$ and suppose there is $w\in c \cap W_\omega$. Then, $x_t(w)\in P(c) \cap T_\omega$, by Lemma~\ref{lemma:AP_P(c)}. Hence, by contraposition, if $P(c) \cap T_\omega = \emptyset$, then $c\cap W_\omega = \emptyset$ as well.

    Concerning completeness, let $\x\in\X$ such that there is $\omega\in D_\x$ with $\x(\omega) \in P(c)$. By Lemma~\ref{lemma:AP_P(c)}, there is $f\in\A^\T$ such that $(\omega,f)\in c$ and $\x(\omega) = x_t(\omega,f)$. First, we infer that $D_{t,f}\neq\emptyset$, whence $D_\x = D_{t,f}$ and $\x = \x_t(f)$, by Proposition~\ref{prop:ev_on_Tr_is_iso}. Second, we infer that $x_t(\omega,f)\cap c \neq\emptyset$. By Assumption~\hyperlink{Ass:AP.C2}{AP.C2}, for any $\omega'\in D_{t,f}$ there is $w' \in x_t(\omega',f)\cap c$. In particular, $\x(\omega') = x_t(\omega',f) = x_t(w')$. Hence, $\x(\omega') \in P(c)$, by Lemma~\ref{lemma:AP_P(c)}. We conclude that $\x^{-1}(P(c)) = D_\x$.\smallskip

    Regarding the proof of the second sentence, let $\x\in\X$ such that $c$ is available at $\x$. There is $\omega\in D_\x$. Then, $\x(\omega)\in P(c)$. As in the proof of completeness above, we infer the existence of $f\in\A^\T$ such that $(\omega,f)\in c$, $\omega\in D_{t,f} = D_\x$ and $\x = \x_t(f)$.
\end{proof}

\begin{proof}[Proof of Proposition~\ref{prop:APsdf_LCS}]
    Let $\x\in\tilde\X^i$, $t=\mf t(\x)$, and $c\in \ms C_\x^i$. In particular, $c\in\ms C_t$, and, by Lemma~\ref{lemma:C_t_non-redundant_complete}, $c$ is a non-redundant and complete, in particular $\tilde\X^i$-complete choice. Let $\omega\in D_\x$. Then, by definition of $\ms C_\x^i$, there is $w\in\x(\omega)\cap c$. Hence, $\x(\omega) = x_{\mf t(\x)}(w) = x_t(w)$, and, by Lemma~\ref{lemma:AP_P(c)}, $\x(\omega)\in P(c)$. We conclude that $c$ is available at $\x$.
\end{proof}

\begin{proof}[Proof of Theorem~\ref{thm:APsdf_AC}]
    Let $c = c(A_{<t},i,g)$.
    \smallskip

    (Ad \ref{thm:APsdf_AC.non_red_and_compl}):~ By Lemma~\ref{lemma:C_t_non-redundant_complete}, $c$ is a non-redundant and complete, in particular $\tilde\X^i$-complete choice which proves the first claim. \smallskip

    (Ad \ref{thm:APsdf_AC.Dx_subset_D}):~ 
    Let $\x\in\tilde\X^i$ be such that $c$ is available at $\x$. Let $\omega\in D_\x$. Then $\x(\omega) \in P(c)$, hence, by Lemma~\ref{lemma:AP_P(c)}, there is $w\in c$ with $\x(\omega) = x_t(w)$. There is $f\in\A^\T$ such that $w=(\omega,f)$, and as $w\in c$, we have $f(t) \in A^{i,g}_{t,\omega}$. Hence, $\omega\in D$. We conclude that $D_\x\subseteq D$.\smallskip
    
    (Helpful statements for $c'\in\ms C^i_\x$ and $\x$ that $c$ is available at):~ Let $\x\in\tilde\X^i$ be such that $c$ is available at $\x$. By definition of $\X$, there is $f_0\in\A^\T$ with $D_{t,f_0} \neq\emptyset$ such that $\x = \x_t(f_0)$. 
    
    Let $c'\in\ms C^i_\x$. We compute the set $P(c\cap c')$ and its preimage under $\x$. By definition of $\ms C^i_\x$, there are $A'_{<t}\subseteq\A^{[0,t)_\T}$ and $A_t^{\prime i}\in\ms B(\A^i)$ such that, with $A'_t = (A'_{t,\omega})_{\omega\in\Omega}$ and $A'_{t,\omega} = (p^i)^{-1}(A_t^{\prime i})$ for all $\omega\in D_\x$, and $A'_{t,\omega}=\emptyset$ for all $\omega\notin D_\x$, we have $c' = c(A'_{<t},A'_t)$, $c'\in\ms C_t$, and, for all $\omega\in D_\x$, $\x(\omega) \cap c' \neq\emptyset$.
    
    Let $c_0 = c(A_{<t}\cap A'_{<t},i,g)$. By definition of $c$ and $c'$, we have
    \begin{align*} 
        &~c\cap c' \\
        =&~ \{(\omega,f)\in W_{D_\x} \mid f|_{[0,t)_\T} \in A_{<t} \cap A'_{<t},~ p^i \circ f(t) = g(\omega)\in A_t^{\prime i}\}\\
        =&~ c(A_{<t}\cap A'_{<t},i,g) \cap \{(\omega,f) \in W_{D_\x} \mid g(\omega) \in A_t^{\prime i} \} \\
        =&~ c_0 \cap W_{g|_{D_\x}^{-1}(A_t^{\prime i})}.
    \end{align*}
    Using Lemma~\ref{lemma:P(c)_compatible_with_conn_comp}, we infer
    \[ (\ast)\qquad P(c\cap c') = P(c_0) \cap F_{g|_{D_\x}^{-1}(A_t^{\prime i})}. \]

    Next, we show that
    \[ (\circ)\qquad \forall \omega\in D_\x\exists f\in \A^\T\colon \quad (\omega,f) \in c_0,~ \x(\omega) = x_t(\omega,f).  \]
    For the proof of $(\circ)$, let $\omega\in D_\x$. As $\x$ is available both at $c$ and $c'$, we have $\x(\omega) = x_t(\omega,f_0)\in P(c)\cap P(c')$. Hence, by Lemma~\ref{lemma:AP_P(c)} applied to $c,c'\in\ms C_t$, there are $f,f'\in\A^\T$ with $(\omega,f)\in c$ and $(\omega,f')\in c'$ such that 
    \[ f'|_{[0,t)_\T} = f_0|_{[0,t)_\T} = f|_{[0,t)_\T}. \]
    Thus, $f|_{[0,t)_\T} \in A_{<t}\cap A'_{<t}$. By definition of $f$ and $f_0$, we must have $p^i \circ f(t) = g(\omega)$. Hence, $(\omega,f)\in c_0$. By definition of $f$, we also have $\x(\omega) = x_t(\omega,f)$.

    We infer that $c_0 \in \ms C_t$.
    Indeed, Assumption~\hyperlink{Ass:AP.C0}{AP.C0} is satisfied by $(\circ)$. Concerning \hyperlink{Ass:AP.C1}{AP.C1}, note that $c_0\subseteq c$. Let $w\in c_0$. Thus $w\in c$. Hence, by \hyperlink{Ass:AP.C1}{AP.C1} applied to $c$, there is $w'\in x_t(w) \setminus c \subseteq x_t(w) \setminus c_0$. Regarding \hyperlink{Ass:AP.C2}{AP.C2}, let $f\in\A^\T$ with $f|_{[0,t)_\T} \in A_{<t}\cap A'_{<t}$ such that there is $\omega\in D_{t,f}$ satisfying
    \[ x_t(\omega,f) \cap c_0 \neq \emptyset. \]
    Let $\omega'\in D_{t,f}$. As $c_0 \subseteq c$, we infer that $x_t(\omega,f) \cap c \neq\emptyset$. Hence, by \hyperlink{Ass:AP.C2}{AP.C2} applied to $c$, there exists $w'\in x_t(\omega',f) \cap c$. Let $f'\in\A^\T$ such that $w' = (\omega',f')$. Then, by definition of $w'$ and $f'$, we get $f'|_{[0,t)_\T} = f|_{[0,t)_\T}\in A_{<t}\cap A'_{<t}$, $\omega'\in D$, and $p^i(f'(t)) = g(\omega')$. Hence, $(\omega',f') \in x_t(\omega',f) \cap c_0$. We conclude that $c_0\in\ms C_t$.

    Hence, by Lemma~\ref{lemma:C_t_non-redundant_complete}, $c_0$ defines a non-redundant and complete choice. Moreover, by $(\circ)$ and Lemma~\ref{lemma:AP_P(c)}, $c_0$ is available at $\x$. 
    In particular, we have $\x^{-1}(P(c_0)) = D_\x$. By $(\ast)$, we get:
    \begin{equation*}
    (\dagger)\qquad
    \begin{aligned} 
        \x^{-1}(P(c\cap c')) = \x^{-1}(P(c_0)) \cap \x^{-1}(F_{g|_{D_\x}^{-1}(A_t^{\prime i})}) 
        = D_{\x} \cap g|_{D_\x}^{-1}(A_t^{\prime i}) = g|_{D_\x}^{-1}(A_t^{\prime i}). 
    \end{aligned}
    \end{equation*}
    
    (Ad \ref{thm:APsdf_AC.Fx_mb_=>_adapted}):~ Let $\x\in\tilde\X^i$ be such that $c$ is available at $\x$. If $g|_{D_\x}$ is $\ms F^i_\x$-measurable, then, by $(\dagger)$, $\x^{-1}(P(c\cap c'))\in\ms F^i_\x$ for all $c'\in\ms C^i_\x$. \smallskip
    
    (Ad \ref{thm:APsdf_AC.adapted_=>_Fx_mb}):~ Suppose conversely that $\x^{-1}(P(c\cap c'))\in\ms F^i_\x$ for all $\x\in\tilde\X^i$ that $c$ is available at and all $c'\in\ms C^i_\x$, and that Assumption~\hyperlink{Ass:AP.C3}{AP.C3} is satisfied for $(A_{<t},i,g)$ and $\ms C^i$. 
    
    Let $\x\in\tilde\X^i$ be such that $c$ is available at it. By Assumption~\hyperlink{Ass:AP.C3}{AP.C3}, there is a generator $\ms G(\A^i)$ of $\ms B(\A^i)$, stable under non-trivial intersections, such that for all $G\in\ms G(\A^i)$, we have $c(A_{<t},A^{i,G}_t)\in\ms C^i_\x$. Let $G\in\ms G(\A^i)$ and $c' = c(A_{<t},A^{i,G}_t)$.
    
    Represent $c'$ as in the beginning of the ``Helpful statements'' part above, with $A_t^{\prime i} = G$ and $A'_{<t} = A_{<t}$. Let $c_0$ be defined as in the ``Helpful statements'' part.
    
    Then, we can use these helpful statements to infer that $c_0$ is a non-redundant and complete choice available at $\x$.
    Hence, $\x^{-1}(P(c_0)) = D_\x$, and thus $(\dagger)$ and the hypothesis imply that
    \[ (g|_{D_\x})^{-1}(G) = \x^{-1}(P(c\cap c')) \in \ms F^i_\x. \]
   
    We conclude that $(g|_{D_\x})^{-1}(G)\in\ms F^i_\x$ for all $G\in\ms G(\A^i)$. Trivially, we have $(g|_{D_\x})^{-1}(\emptyset) = \emptyset \in\ms F^i_\x$. As $\ms G(\A^i)\cup \{\emptyset\}$ is a generator of $\ms B(\A^i)$ stable under intersections, in view of the $\pi$-$\lambda$ theorem, $g|_{D_\x}$ is $\ms F^i_\x$-measurable.
\end{proof}

\begin{proof}[Proofs for Example~\ref{ex:APsdf_AC}]
    (Ad $W=\Omega\times\A^\T$):~ We show the following more general statement: If $t\in\T$ and $A_{<t}\subseteq\A^{[0,t)_\T}$ is non-empty, and moreover $A_t = (A_{t,\omega})_{\omega\in\Omega}\in \mc P(\A)^\Omega$ is such that 
    \[ (\ast) \qquad \forall \omega\in\Omega \colon\quad\emptyset \subsetneq A_{t,\omega} \subsetneq \A, \]
    then $c(A_{<t},A_t) \in \ms C_t$. This includes the two following cases: A) $A_t = A_t^{i,g}$ for given $i\in I$ and a map $g\colon \Omega\to\A^i$ because $p^i$ is surjective and $\A^i$ has at least two elements; and B) $A_{t,\omega} = (p^i)^{-1}(G)$ for some $i\in I$ and some set $\emptyset\subsetneq G\subsetneq \A^i$, again because $p^i$ is surjective.
    
    As $A_{<t}\neq\emptyset$ and $(\ast)$ is assumed, there is $w=(\omega,f) \in W=\Omega\times\A^\T$ such that $f|_{[0,t)_\T}\in A_{<t}$ and $f(t) \in A_{t,\omega}$. Hence, $w\in c(A_{<t},A_t)$. Thus, \hyperlink{Ass:AP.C0}{AP.C0} is satisfied.\smallskip

    Regarding \hyperlink{Ass:AP.C1}{AP.C1}, let $w=(\omega,f)\in c(A_{<t},A_t)$. There is $f'\in\A^\T$ with $f'|_{[0,t)_\T} = f|_{[0,t)_\T}$ and $f'(t) \notin A_{t,\omega}$ (by $(\ast)$). Then $w'=(\omega,f')\in \Omega\times\A^\T = W$, hence $w'\in x_t(w)$, but $w'\notin c(A_{<t},A_t)$.\smallskip

    Regarding \hyperlink{Ass:AP.C2}{AP.C2}, let $f\in\A^\T$ with $f|_{[0,t)_\T}\in A_{<t}$ and $\omega\in D_{t,f}$. There is $f'\in\A^\T$ with $f'|_{[0,t)_\T} = f|_{[0,t)_\T}$ and $f'(t) \in A_{t,\omega}$, by $(\ast)$. As $(\omega,f') \in \Omega\times\A^\T = W$, we infer
    \[ (\omega,f') \in x_t(\omega,f) \cap c(A_{<t},A_t). \]

    Hence, $c(A_{<t},A_t)\in\ms C_t$. The general statement above is proven. By considering the case A), we infer that for all $i\in I$ and $g\colon \Omega\to\A^i$, $c(A_{<t},i,g)\in\ms C_t$.\smallskip

    Regarding \hyperlink{Ass:AP.C3}{AP.C3}, suppose that for all $\x\in\tilde\X^i$ with $\mf t(\x)=t$, $\ms C^i_\x$ contains all $c(A_{<t},A_t)$ ranging over all $A_t$ satisfying ($\ms C_\x^i$.$k$), $k=2,3,4$. Let $\x\in\tilde\X^i$ be such that $c(A_{<t},i,g)$ is available at $\x$.
    By Lemma~\ref{lemma:C_t_non-redundant_complete}, there is $(\omega_0,f_0)\in c(A_{<t},i,g)$ such that $\x = \x_t(f_0)$ and $\omega_0\in D_\x$.
    
    Let $\ms G(\A^i) = \ms B(\A^i)\setminus \{\A^i,\emptyset\}$ which is obviously a generator of $\ms B(\A^i)$ stable under non-trivial intersections. Let $G\in\ms G(\A^i)$. Then, the general statement $(\ast)$ above in case B) applies, hence $c(A_{<t},A_t^{i,G})\in\ms C_t$. To complete the proof that $c(A_{<t},A_t^{i,G})\in\ms C_\x^i$, in view of our additional assumption on $\ms C^i_\x$, it remains to prove Property ($\ms C_\x^i$.\ref{def:msC.4}) for $c(A_{<t},A_t^{i,G})$, ($\ms C_\x^i$.\ref{def:msC.3}) having been proven just before and ($\ms C_\x^i$.\ref{def:msC.2}) being evident. For this, let $\omega'\in D_\x$. Then there is $f'\in\A^\T$ such that $f'|_{[0,t)_\T} = f_0|_{[0,t)_\T}\in A_{<t}$ and $p^i \circ f'(t) \in G$, because $G$ is non-empty and $p^i$ surjective. Then, $w' = (\omega',f')\in \Omega\times \A^\T = W$, and thus $w'\in \x_t(f_0)(\omega') \cap c(A_{<t},A_t^{i,G})$. Property ($\ms C_\x^i$.\ref{def:msC.4}) therefore holds true for $c(A_{<t},A_t^{i,G})$ and the proof of the statement $c(A'_{<t},A_t^{i,G})\in\ms C_\x^i\cup\{\emptyset\}$ is complete.\medskip

    (Ad timing problem):~ Let $t\in\T$. Further, let $A_{<t}\subseteq \A^{[0,t)_\T}$ be a non-empty set of componentwise decreasing paths $f_t$ such that $p^i\circ f_t = 1_{[0,t)_\T}$, and $g\colon\Omega\to\{0,1\} = \A^i$. Let $c = c(A_{<t},i,g)$. We are going to prove that $c\in\ms C_t$. 

    Regarding \hyperlink{Ass:AP.C0}{AP.C0}, let $\omega\in\Omega$. By assumption on $A_{<t}$, there is componentwise decreasing $f\colon\T \to \A$ such that $f|_{[0,t)_\T}\in A_{<t}$ and $p^i \circ f(t) = g(\omega)$. Hence, $(\omega,f)\in W$, and even $(\omega,f)\in c$. Thus, $c\neq\emptyset$.\smallskip

    Regarding \hyperlink{Ass:AP.C1}{AP.C1}, let $w=(\omega,f)\in c$. There is componentwise decreasing $f'\in\A^\T$ with $f'|_{[0,t)_\T} = f|_{[0,t)_\T}$ and $p^i\circ f'(t)\neq g(\omega)$, because $p^i\circ f|_{[0,t)_\T} = 1_{[0,t)_\T}$. Then $w'=(\omega,f')\in W$, hence $w'\in x_t(w)$, but $w'\notin c$.\smallskip

    Regarding \hyperlink{Ass:AP.C2}{AP.C2}, let $f\in\A^\T$ with $f|_{[0,t)_\T}\in A_{<t}$ and $\omega\in D_{t,f}$. There is componentwise decreasing $f'\in\A^\T$ with $f'|_{[0,t)_\T} = f|_{[0,t)_\T}$ and $p^i \circ f'(t) = g(\omega)$, by assumption on $A_{<t}$. As $(\omega,f') \in W$, we infer
    \[ (\omega,f') \in x_t(\omega,f) \cap c. \]
    We conclude that $c\in \ms C_t$.\smallskip

    Regarding Assumption~\hyperlink{Ass:AP.C3}{AP.C3}, suppose that for all $\x\in\tilde\X^i$ with $\mf t(\x)=t$, $\ms C^i_\x$ contains all $c(A_{<t},A_t)$ ranging over all $A_t$ satisfying ($\ms C_\x^i$.$k$), $k=2,3,4$. Let $\x\in\tilde\X^i$ be such that $c(A_{<t},i,g)$ is available at $\x$. By Lemma~\ref{lemma:C_t_non-redundant_complete}, there is $(\omega_0,f_0)\in c(A_{<t},i,g)$ such that $\omega_0 \in D_{t,f_0}=D_\x$ and $\x = \x_t(f_0)$. As $c(A_{<t},i,g)\in\ms C_t$, Assumption~\hyperlink{Ass:AP.C1}{AP.C1} implies, that $p^i \circ f_0|_{[0,t)_\T}$ is constant with value $1$. Indeed, there is $f_1\in\A^\T$ with $(\omega,f_1)\in x_t(\omega,f_0)\setminus c(A_{<t},i,g)$. In particular, $f_1$ is componentwise decreasing and takes the same values on $[0,t)_\T$ as $f_0$. If $p^i \circ f_0$ took the value $0$ at some point in $[0,t)_\T$, the monotonicity would imply $p^i\circ f_1(t) = 0 = p^i\circ f_0(t) = g(\omega)$, whence $(\omega,f_1)\in c(A_{<t},i,g)$, in contradiction to the choice of $f_1$.

    Further, let $\ms G(\A^i) = \ms B(\A^i)\setminus \{\A^i,\emptyset\}$ which is obviously a generator of $\ms B(\A^i)$ stable under non-trivial intersections. As $\A^i = \{0,1\}$, $\ms G(\A^i) = \{\{0\},\{1\}\}$. Let $G\in\ms G(\A^i)$. Thus, $G$ is a singleton and $c(A_{<t},A_t^{i,G}) = c(A_{<t},i,g_G)$ for the constant map $g_G$ with value given by the unique element of $G$. We have shown just beforehand that $c(A_{<t},i,g_G)\in\ms C_t$. To complete the proof of the fact that $c(A_{<t},A_t^{i,G})\in\ms C_\x^i$, in view of our additional assumption on $\ms C^i_\x$, it thus remains to show Axiom ($\ms C_\x^i$.\ref{def:msC.4}) in the definition of $\ms C_\x$, because ($\ms C_\x^i$.\ref{def:msC.3}) has just been proven and ($\ms C_\x^i$.\ref{def:msC.2}) is evident by construction. For this, let $\omega'\in D_\x$. As $p^i\circ f_0$ only takes the value $1$ on $[0,t)_\T$, there is componentwise decreasing $f' \colon \T\to\A$ such that $f'|_{[0,t)_\T} = f_0|_{[0,t)_\T}$ and $p^i\circ f'(t) = g_G(\omega')$. Hence, $(\omega',f')\in \x_t(f_0)(\omega')\cap c(A_{<t},i,g_G) = \x(\omega') \cap c(A_{<t},A_t^{i,G})$. This completes the proofs for the timing game example.\medskip

    (Ad up-and-out option exercise example):~ Let $t\in\R_+$, $A_{<t} = \{1\}^{[0,t)}$ and $D$ be the set of $\omega\in\Omega$ such that $\max_{u\in [0,t]} P_u(\omega) < 2$. We suppose that $D\neq\emptyset$. Let $g\colon D\to\{0,1\}$ be a map and let $c = c(A_{<t},i,g)$. We are first going to prove that $c\in\ms C_t$. 

    We start with the proof of \hyperlink{Ass:AP.C0}{AP.C0}. There is $\omega_0\in D$. As $P$ is continuous, there is $\e>0$ such that $\max_{u\in [0,t+\e]} P_u(\omega_0) < 2$. Then, regardless of the value of $g(\omega_0)$, there is decreasing $f\colon{\R_+} \to \A$ such that $f|_{[0,t)} = 1_{[0,t)}$, $f(t) = g(\omega_0)$, and $f(t+\e) = 0$. Hence, $(\omega_0,f)\in W$, and even $(\omega_0,f)\in c$. Thus, $c\neq\emptyset$.\smallskip

    Regarding \hyperlink{Ass:AP.C1}{AP.C1}, let $w=(\omega,f)\in c$. Then, $\max_{u\in [0,t]} P_u(\omega) < 2$, and by continuity of $P$, there is $\e>0$ such that $\max_{u\in [0,t+\e]} P_u(\omega) < 2$. Hence, regardless of the value of $g(\omega)$, there is decreasing $f'\in\A^{\R_+}$ with $f'|_{[0,t)} = f|_{[0,t)} = 1_{[0,t)}$, $f'(t)\neq g(\omega)$, and $f'(t+\e) = 0$. Then $w'=(\omega,f')\in W$, hence $w'\in x_t(w)$, but $w'\notin c$.\smallskip

    Regarding \hyperlink{Ass:AP.C2}{AP.C2}, let $f\in\A^{\R_+}$ with $f|_{[0,t)}\in A_{<t}$ and $\omega\in D_{t,f}$. Then there is decreasing $\tilde f\colon\R_+\to\A$ with $(\omega,\tilde f)\in x_t(\omega,f)$, that is, $\tilde f|_{[0,t)} = f|_{[0,t)}$ and $\max_{u\in [0,t]} P_u(\omega) < 2$. By continuity of $P$, there is $\e>0$ such that $\max_{u\in [0,t+\e]} P_u(\omega) < 2$. Hence, regardless of the value of $g(\omega)$, there is decreasing $f'\in\A^{\R_+}$ with $f'|_{[0,t)} = f|_{[0,t)} = 1_{[0,t)}$, $f'(t) = g(\omega)$, and $f'(t+\e) = 0$. Hence, $(\omega,f') \in W$, and we infer
    \[ (\omega,f') \in x_t(\omega,f) \cap c. \]
    We conclude that $c\in \ms C_t$.\smallskip

    Regarding Assumption~\hyperlink{Ass:AP.C3}{AP.C3}, suppose that for all $\x\in\tilde\X^i$ with $\mf t(\x)=t$, $\ms C^i_\x$ contains all $c(A_{<t},A_t)$ ranging over all $A_t$ satisfying ($\ms C_\x^i$.$k$), $k=2,3,4$. Let $\x\in\tilde\X^i$ be such that $c(A_{<t},i,g)$ is available at $\x$. By Lemma~\ref{lemma:C_t_non-redundant_complete}, there is $(\omega_0,f_0)\in c(A_{<t},i,g)$ such that $\x = \x_t(f_0)$ and $\omega_0\in D_\x$. As $c(A_{<t},i,g)\in\ms C_t$, Assumption~\hyperlink{Ass:AP.C1}{AP.C1} implies, that $f_0|_{[0,t)}$ is constant with value $1$. Indeed, there is $f_1\in\A^\T$ with $(\omega,f_1)\in x_t(\omega,f_0)\setminus c(A_{<t},i,g)$. In particular, $f_1$ is componentwise decreasing and takes the same values on $[0,t)_\T$ as $f_0$. If $f_0$ took the value $0$ at some point in $[0,t)$, the monotonicity would imply $f_1(t) = 0 = f_0(t) = g(\omega)$, whence $(\omega,f_1)\in c(A_{<t},i,g)$, in contradiction to the choice of $f_1$.

    Let $\ms G(\{0,1\}) = \ms B(\{0,1\})\setminus \{\{0,1\},\emptyset\}$ which is obviously a generator of $\ms B(\{0,1\})$ stable under non-trivial intersections. Note that $\ms G(\{0,1\}) = \{\{0\},\{1\}\}$. Let $G\in\ms G(\{0,1\})$.
    Thus, $G$ is a singleton and $c(A_{<t},A_t^{i,G}) = c(A_{<t},i,g_G)$ for the constant map $g_G$ with value given by the unique element of $G$.

    We have shown just above that $c(A_{<t},i,g_G)\in\ms C_t$. To complete the proof of the fact that $c(A_{<t},i,g_G)\in\ms C_\x^i$, in view of our additional assumption on $\ms C^i_\x$, it remains to show Axiom ($\ms C_\x^i$.\ref{def:msC.4}), because ($\ms C_\x^i$.\ref{def:msC.3}) has just been proven and ($\ms C_\x^i$.\ref{def:msC.2}) is evident by construction. For this, let $\omega'\in D_\x = D_{t,f_0}$. In particular, $\max_{u\in [0,t]} P_u(\omega') < 2$. By continuity of $P$, there is $\e>0$ such that $\max_{u\in [0,t+\e]} P_u(\omega') < 2$. Then, regardless of the value of $g(\omega)$, and as $f_0$ only takes the value $1$ on $[0,t)$, there is componentwise decreasing $f' \colon {\R_+}\to\A$ such that $f'|_{[0,t)} = f_0|_{[0,t)}$, $f'(t) = g_G(\omega')$, and $f'(t+\e) = 0$. Hence, $(\omega',f')\in \x(\omega')\cap c(A_{<t},i,g_G)$. This completes the proofs for the up-and-out option exercise problem example.
\end{proof}